\documentclass[12pt]{article}

\usepackage{amsmath,amsfonts,bm}

\def\1{\bm{1}}

\DeclareMathAlphabet{\mathsfit}{\encodingdefault}{\sfdefault}{m}{sl}
\SetMathAlphabet{\mathsfit}{bold}{\encodingdefault}{\sfdefault}{bx}{n}

\def\gA{{\mathcal{A}}}

\def\gC{{\mathcal{C}}}
\def\gD{{\mathcal{D}}}

\def\gF{{\mathcal{F}}}

\def\gI{{\mathcal{I}}}

\def\gM{{\mathcal{M}}}

\def\gR{{\mathcal{R}}}

\def\gT{{\mathcal{T}}}

\def\gX{{\mathcal{X}}}
\def\gY{{\mathcal{Y}}}

\newcommand{\E}{\mathbb{E}}

\newcommand{\Eqmark}[2]{\stackrel{(#1)}{#2}}

\def\X{{\bf{X}}}

\usepackage{graphicx}       
\graphicspath{ {./} }       
\usepackage{amsmath}        
\usepackage{amssymb}        
\usepackage[utf8]{inputenc} 
\usepackage[T1]{fontenc}    
\usepackage{url}            
\usepackage{booktabs}       
\usepackage{amsfonts}       
\usepackage{nicefrac}       
\usepackage{microtype}      
\usepackage{todonotes}      
\usepackage[colorlinks=true,linkcolor=blue,urlcolor=blue,citecolor=blue]{hyperref}
\usepackage{enumitem}       
\usepackage{mathtools}      
\usepackage{natbib}         
\usepackage{multirow}
\setcitestyle{authoryear}
\usepackage{dsfont}
\usepackage{caption}
\usepackage{subfigure}
\usepackage{parskip}
\usepackage{geometry}
\usepackage{algorithm}
\usepackage{algorithmic}

\usepackage{amsthm}
\newtheorem{theorem}{Theorem}
\newtheorem{lemma}{Lemma}

\newtheorem{assumption}{Assumption}
\newtheorem{proposition}{Proposition}[section]
\newtheorem{remark}{Remark}[section]

\newtheorem{corollary}{Corollary}[section]
\newtheorem{definition}{Definition}[section]

\def\FDR{\operatorname{FDR}}
\def\mFDR{\operatorname{mFDR}}

\def\FDP{\operatorname{FDP}}

\def\E{{\mathbb{E}}}
\def\P{\mathbb{P}}
\def\I{{\mathbb{I}}}
\def\X{{\bf{X}}}
\def\x{{\bf{x}}}
\def\hdelta{\boldsymbol{\delta}}

\def\X{{\bf X}}
\def\x{{\bf x}}

\def\E{{\mathbb{E}}}

\def\P{\mathbb{P}}

\def\I{{\mathbb{I}}}
\def\gI{{\mathcal{I}}}
\def\gA{{\mathcal{A}}}
\def\gD{{\mathcal{D}}}

\def\gX{{\mathcal{X}}}
\def\gY{{\mathcal{Y}}}
\def\gR{{\mathcal{R}}}
\def\gF{{\mathcal{F}}}

\def\gM{{\mathcal{M}}}

\def\gT{{\mathcal{T}}}
\def\gC{{\mathcal{C}}}

\def\hdelta{{\bm \delta}}

\usepackage{xr}    
\usepackage{booktabs}
\usepackage{multirow}
\usepackage{tabularx}
\usepackage{array}
\usepackage{geometry}
\usepackage{booktabs}
\usepackage{multirow}
\usepackage{tabularx}
\usepackage{pifont}
\usepackage{makecell} 
\usepackage{geometry}
\usepackage{graphicx}
\usepackage{subcaption}
\title{\LARGE\bf Feedback-Enhanced Online Multiple Testing with Applications to Conformal Selection}
\date{}
\author{
}

\begin{document}

\author{Lin Lu$^{*1}$,\ Yuyang Huo$^{*1}$,\ Haojie Ren$^{\text{†}2}$, Zhaojun Wang$^{1}$,\  and Changliang Zou$^{1}$              \\
{$^1${\small\it School of Statistics and Data Sciences,  LPMC, KLMDASR and LEBPS,}} \\
{{\small \it Nankai University, Tianjin, China}} \\
{$^2${\small\it{School of Mathematical Sciences, Shanghai Jiao Tong University, Shanghai, China}}}\\
}
\date{}
\maketitle
\def\thefootnote{*}\footnotetext{The first two authors equally contributed to this work.}\def\thefootnote{\arabic{footnote}}
\def\thefootnote{†}\footnotetext{Correspondence to: Haojie Ren <haojieren@sjtu.edu.cn> }\def\thefootnote{\arabic{footnote}}

\bigskip
\begin{abstract}

This work studies online multiple testing with feedback, where decisions are made sequentially, and the true state of the hypothesis is revealed after decisions are made, either instantly or with a delay, {\color{black}and under either full or bandit feedback}. We propose Generalized alpha-investing with feedback (GAIF) \textcolor{black}{along with its adaptive variants}, a feedback-enhanced framework that dynamically adjusts thresholds using revealed outcomes, ensuring finite-sample false discovery rate (FDR)/marginal FDR (mFDR) control. {\color{black}We further extend GAIF to online conformal testing by constructing valid conformal $p$-values and developing feedback-enhanced testing rules with finite-sample mFDR control. We also propose a feedback-driven score selection criterion to adaptively choose the candidate score that is most effective for the testing procedure, together with a theoretical analysis of its optimality.} Numerical simulations and real-data applications demonstrate the effectiveness of our methods.
\end{abstract}

\noindent%


\section{Introduction}
Real-time decision making plays a critical role in a growing number of modern applications, such as online recruitment for job hiring \citep{faliagka2014line}, real-time alignment of large language models \citep{huang2025survey}, and time-series anomaly detection \citep{rebjock2021online}. These tasks 
can be naturally formulated as online multiple testing problems \citep{foster2008alpha}. 
{\color{black}Consider a potentially infinite stream of null hypotheses \(\{\mathbb{H}_{01}, \mathbb{H}_{02}, \dots,\mathbb{H}_{0t},\dots\}\) that are tested sequentially based on incoming test statistics such as $p$-values \(\{p_1, p_2, \dots p_t,\dots\}\)}. At each time \(t\), a real-time decision must be made based on the current test statistic. 
Let \(\theta_t\) denote the true state of the hypothesis at time \(t\), where \(\theta_t = 0\) if \(\mathbb{H}_{0t}\) is true and \(\theta_t = 1\) otherwise. The testing problem at time \(t\) is
\[
\mathbb{H}_{0t}: \theta_t = 0 \quad \text{vs.} \quad \mathbb{H}_{1t}: \theta_t = 1.
\]
Let \(\delta_t \in \{0,1\}\) be the decision at time \(t\), with \(\delta_t = 1\) indicating rejection of \(\mathbb{H}_{0t}\). To ensure the reliability of the testing procedure, it is essential to control a well-defined error rate. Define the false discovery proportion (FDP) and false discovery rate (FDR) \citep{Benjamini1995ControllingTF} at time \(t\) by
\[
\FDP(t) = \frac{V(t)}{1 \vee R(t)} 
    \;:=\; \frac{\sum_{j=1}^t \delta_j \bigl(1 - \theta_j\bigr)}{1 \vee \sum_{j=1}^t \delta_j}, 
\quad
\FDR(t)= \mathbb{E}\left\{\FDP(t)\right\},
\]
where 
$V(t)$ and $R(t)$ represent the numbers of false rejections (discoveries) and rejections at time $t$, respectively. Given a target level \(\alpha \in (0,1)\), classical work on online multiple testing \citep{ramdas2017online, ramdas2018saffron, tian2019addis} aims to guarantee $\sup_t\FDR(t)\leq \alpha$ or its weaker variant  \citep{foster2008alpha, Zrnic2021asynchronous}, the marginal FDR (mFDR), i.e.,
$\operatorname{mFDR}(t) 
    = {\mathbb{E}\{V(t)\}}/{\mathbb{E}\{1 \vee R(t)\}}$.

{\color{black} Unlike the classical setting, we consider another realistic scenario in which feedback on past hypotheses is revealed after decisions are made. Specifically, after issuing $\delta_t$, the corresponding state $\theta_t$ may be observed either immediately or after a delay, and in some settings it may be observed only partially. 
For instance, in the bandit setting \citep{neu2010online}, feedback $\theta_t$ is observed only when the decision $\delta_t=1$ is made. Our work encompasses all these feedback-available settings, including full or partial feedback, as well as immediate or delayed feedback.} 

Such feedback-available settings naturally arise in many practical applications. To demonstrate its relevance, we present three {motivating examples}: online conformal selection, real-time alignment of large language models, and time-series anomaly detection.
\begin{itemize}
  \item \textbf{Online conformal selection.} Conformal selection aims to identify valuable individuals whose unknown label $Y$ satisfies a pre-specified requirement by leveraging machine learning predictions \citep{wu2024optimal,jin2025model}. 
  The online conformal selection setting naturally aligns with the feedback-available online multiple testing framework. For example, in diabetes risk prediction, a patient’s true condition may later be confirmed by an expert, offering feedback that could improve future decisions. Some works \citep{huo2024realtime, xu2024online} considered online conformal testing problems but did not exploit feedback information.

  \item \textbf{Real-time LLM alignment.} 
  Large language models (LLMs) are increasingly used in high-stakes domains such as healthcare, finance, and law, where outputs must be reliable and factual. However, LLMs can hallucinate—producing plausible but incorrect content \citep{huang2025survey}. A remedy is to
  filter or certify LLM outputs \citep{guiconformal, bai2024optimized}. 
  Some alignment approaches based on conformal testing are proposed for this use but are generally designed for offline environments, whereas many applications demand immediate and trustworthy screening. 
  In such applications, follow-up feedback is usually available, 
  which, if incorporated, could continuously improve alignment.

  \item \textbf{Time series anomaly detection.} 
  Detecting anomalies in time series is crucial for industrial monitoring, fraud detection, and healthcare analytics. To ensure reliability, prior work \citep{rebjock2021online, gang2021structure,  kronert2023fdr} addresses online FDR control but typically ignores real-time feedback. In practice, such feedback is often available: once an anomaly is flagged, a subsequent system failure, user verification of a fraudulent transaction, or expert annotation may confirm or refute the alarm. 
\end{itemize}

These observations underscore that feedback is a foundational element of adaptive decision-making, playing a pivotal role across a wide range of online applications. Despite its importance, current online multiple testing procedures seldom incorporate feedback in a systematic way. This naturally raises a fundamental and compelling question: \emph{Can available feedback be effectively incorporated into online multiple testing procedures in a way that allows us to enhance statistical power while still ensuring valid error rate control?} 

{\color{black}To this end, we develop a feedback-enhanced framework that systematically integrates revealed feedback into the generalized alpha-investing (GAI) framework \citep{aharoni2014generalized,javanmard2018online}, 
achieving significant performance gains without compromising statistical validity. }


\subsection{Our contributions}

To the best of our knowledge, this is the first work to incorporate feedback information directly into the construction of testing thresholds for online FDR procedures and to employ this idea to online conformal testing. Our main contributions are summarized as follows: 

\begin{itemize}
{\color{black}
    \item \textbf{GAIF and its variants.} We introduce Generalized Alpha-Investing with Feedback (GAIF), a feedback-enhanced extension of the GAI 
    framework, together with its adaptive variants. The key idea is to use revealed feedback to sharpen the FDP estimator by discounting confirmed non-null hypotheses, and, in the adaptive version, {\color{black}to recycle the $\alpha$-wealth that would otherwise be permanently lost whenever a non-null hypothesis yields a large $p$-value.} The framework applies to a broad range of feedback settings, including full or partial feedback, as well as immediate or delayed feedback, and it also extends to locally dependent scenarios. We establish online FDR control under independence and online mFDR control under local dependence for GAIF and its variants.
    \item \textbf{Online conformal testing with feedback.}
    We extend the GAIF to \textit{online conformal testing}. We first provide an explicit construction of valid, independent $p$-values (for null hypotheses) by dynamically updating the calibration dataset. We then combine them with suitably modified GAIF rules to obtain \emph{online conformal testing with feedback} (OCTF), which enjoys finite-sample online mFDR control. This extension bridges the gap between traditional online multiple testing and conformal inference, yielding distribution-free, model-agnostic tools for real-time decision-making.
   \item  \textbf{Feedback-driven score selection.}  We further introduce a strategy to adaptively select among candidate predictive models or conformity scores. This is particularly important under non-stationary non-null distributions, where the optimal score may vary over time. By locally averaging recent non-null $p$-values to track the best-performing score, our selection criterion can accommodate such shifts while maintaining validity. }
\end{itemize}


We provide rigorous proofs for all proposed methods. Extensive simulations and real‐data experiments demonstrate that our procedures substantially outperform existing approaches while effectively controlling the online FDR when feedback information is provided.

\subsection{Related works}

Our work is situated at the intersection of online multiple testing and conformal inference. We review key developments in each area and highlight gaps that motivate our contribution.

\paragraph{Online multiple testing under independence.} Early works on online multiple testing began with the alpha-investing strategy of \citet{foster2008alpha}, later generalized by \citet{aharoni2014generalized} and \citet{javanmard2018online} into the GAI framework, which led to the LORD algorithm. Building on this line of work, \citet{ramdas2017online} introduced LORD++, an improved version of GAI tailored for online FDR control. Subsequent refinements include SAFFRON \citep{ramdas2018saffron}, which adapts to the proportion of non-nulls, and ADDIS \citep{tian2019addis}, which adjusts for conservative null $p$-values. These methods, including LORD++, SAFFRON, and ADDIS, achieve online FDR control when null $p$-values are independent of all other $p$-values. 
Separately, \citet{gang2021structure} developed structure-adaptive rules based on local FDR, which improve power but only ensure asymptotic FDR control under correct model specification. For a comprehensive review, see \citet{robertson2023online}. All of these methods determine thresholds solely from past rejections, without considering real-time feedback.

\paragraph{Online multiple testing under dependence.} In practice, hypotheses often exhibit dependence, and applying methods designed for independence can lead to inflated error rates. To address arbitrary dependence, \citet{xu2024online} proposed \textit{e}-LOND, an FDR-controlling procedure based on \textit{e}-values. \citet{zhang2025egai} extended this approach to \textit{e}-GAI, achieving improved power by dynamically allocating the
 testing levels. Alternatively, research has focused on local dependence structures. \citet{Zrnic2021asynchronous} introduced \textsc{LORD}\textsubscript{dep} and \textsc{SAFFRON}\textsubscript{dep}, establishing mFDR control under local dependence; \citet{rebjock2021online} later adapted these methods to time-series anomaly detection. Recently, \citet{fisher2024online} showed that \textsc{LORD}++  with suitable local modifications can maintain FDR control under certain dependence, while \citet{fischer2024online} proposed an online procedure for p-values {\color{black}satisfying positive regression dependence on a subset (PRDS).} Despite these advances,
 existing dependence-aware methods do not incorporate any real-time feedback. 
 
\paragraph{Conformal inference and conformal multiple testing.}

Conformal inference \citep{vovk2005algorithmic} offers a model-agnostic way to quantify prediction uncertainty. 
In the multiple testing setting, early works constructed conformal $p$-values and applied the Benjamini–Hochberg (BH) procedure \citep{Benjamini1995ControllingTF} to achieve finite-sample FDR control \citep{bates2021testing,jin2025model}. Subsequent extensions addressed covariate shift \citep{jin2023selection}, constrained selection \citep{wu2024optimal,nair2025diversifying}, and conditional testing \citep{wu2025conditional}, as well as model selection \citep{bai2024optimized,gui2025acs}. However, these contributions remain confined to offline settings. The few existing efforts to extend conformal multiple testing to the online domain {\color{black}\citep{huo2024realtime, xu2024online,liu2026online}} have not taken feedback information into account. 
Although related research on the construction of {\color{black}online conformal prediction sets \citep{gibbs2021adaptive, gibbs2024conformal,xi2025exploring}} has considered feedback, multiple testing problems with FDR control remain largely unexplored. One very recent work, by \citet{humbert2025onlineselectiveconformalinference}, establishes asymptotic online FDP control through an online learning strategy, while we achieve feedback-enhanced testing based on the GAI framework.

\subsection{Organization of the paper}
The remainder of this paper is organized as follows. Section \ref{sec:GAIF} introduces the  GAIF and establishes FDR control under independence. {\color{black}Section \ref{sec:exten-GAIF} extends the GAIF to achieve adaptive $\alpha$-wealth allocation and to address local dependence.} In Section \ref{sec:OCT}, we construct explicit online conformal $p$-values and apply modified GAIF rules to online conformal testing, providing finite-sample theoretical guarantees. We also extend the proposed framework to achieve score selection and address distribution shifts. Simulation and real-data experiment results are presented in Sections \ref{sec:simu} and \ref{sec:real-data}, respectively. Finally, we conclude the paper in Section \ref{sec:conclu}.

\section{Generalized Alpha-Investing with Feedback (GAIF)}\label{sec:GAIF}
{\color{black} In this section, we begin by revisiting the generalized alpha-investing (GAI) framework 
in Section~\ref{subsec:GAI}. In Section~\ref{subsec:GAIF}, we introduce our feedback-enhanced framework, termed generalized alpha-investing with feedback (GAIF), which incorporates feedback information to improve FDP estimation. In Section~\ref{subsec:finite-FDR-GAIF}, we establish finite-sample online FDR control for GAIF procedures. 
}


\subsection{Recap: Generalized alpha-investing}\label{subsec:GAI}


Early approaches to online error-rate control were based on so-called \emph{alpha-investing} strategies \citep{foster2008alpha}, which were later generalized into the GAI framework \citep{aharoni2014generalized,javanmard2018online}. {\color{black}These procedures are based on the principle of \emph{alpha-wealth}: the testing process starts with an initial wealth, spends wealth to conduct hypothesis tests, and earns wealth whenever a rejection is made, thereby generating resources for future tests.} Accordingly, each incoming $p$-value $p_t$ is compared against a dynamically chosen threshold $\alpha_t$, and the testing decision is given by $\delta_t=\mathbb{I}\{p_t\le\alpha_t\}$.

Based on GAI, \citet{ramdas2017online} proposed controlling online FDR by ensuring that an estimator of the FDP remains below a pre‐specified level~$\alpha$.  A specific example is the LORD++ algorithm, where the estimated $\operatorname{FDP}$ at time $t$ is given by
\begin{equation}\label{eq:LORD_FDP}
  \widehat{\operatorname{FDP}}_{\text{LORD}}(t)=\frac{\sum_{j\leq t}\alpha_j}{1\vee\sum_{j\leq t}\delta_j}\overset{(i)}{\geq} \operatorname{FDP}^*(t)= \frac{\sum_{j\leq t,j\in \mathcal{H}_0} \alpha_j}{1\vee\sum_{j\leq t}\delta_j}\approx\operatorname{FDP}(t):=\frac{\sum_{j\leq t,j\in \mathcal{H}_0}\mathbb{I}\{p_j\leq \alpha_j\}}{1\vee\sum_{j\leq t}\delta_j},  
\end{equation}
  where $\mathcal{H}_0$ is the index set of null hypotheses.
    To ensure that $\operatorname{FDR}(t)\leq\alpha$, it suffices to enforce $\widehat{\operatorname{FDP}}_{\text{LORD}}(t)=\frac{\sum_{j\leq t}\alpha_j}{R(t)\vee 1}\leq \alpha.$
    {\color{black} The corresponding testing level of LORD++ is 
    \begin{equation*}\label{eq:alpha-LORD++}
        \alpha_t^{\text{LORD++}}=\gamma_t s_0+(\alpha-s_0)\gamma_{t-\tau_1}\mathbb{I}\{\tau_1<t\}+\alpha\sum_{j:\tau_j<t,\tau_j\neq \tau_1}\gamma_{t-\tau_j},
    \end{equation*}
    where $\{\gamma_t\}_{t=1}^{\infty}$ is a given infinite non-increasing sequence of positive constants that sums to one, $\tau_j$ is the time of the $j$-th rejection, and $s_0>0$ is the pre-specified initial wealth.
    } 
    
    Building on this, some adaptive versions of LORD++ were proposed subsequently, such as SAFFRON \citep{ramdas2018saffron} and ADDIS \citep{tian2019addis}. Specifically, 
    SAFFRON introduces a user-specified parameter $\lambda\in[0,1]$ to estimate {\color{black}the proportion of true nulls in the online testing stream}, and obtains a less conservative estimated FDP as    \begin{equation}\label{eq:FDP_SAFFRON}
\widehat{\operatorname{FDP}}_{\text{SAFFRON}}(t)=\frac{\sum_{j\leq t}\alpha_j\frac{\mathbb{I}\{p_j>\lambda\}}{(1-\lambda)}}{1\vee\sum_{j\leq t}\delta_j}.
    \end{equation}
    
These GAI‐based procedures 
guarantee online mFDR control under the {\it conditional super‐uniformity} assumption for null $p$‐values:
\begin{align}\label{csua}
\Pr\bigl(p_t \le \alpha_t \,\bigm|\, \mathcal{G}_{t-1}\bigr)
\;\le\;
\alpha_t
\quad
\text{for all } t \in \mathcal{H}_0,
\end{align}
{\color{black}where $\mathcal{G}_{t} = \sigma(\delta_{1:t})$ and each threshold $\alpha_t = f_t(\delta_{1:t-1})$ is measurable with respect to past decisions. Here $\delta_{1:t}:=(\delta_{1},\dots,\delta_{t})$.} Furthermore, if all null $p$‐values are independent of other $p$‐values, these procedures also ensure strict FDR control.


Note that the inequality (i) in \eqref{eq:LORD_FDP} is tight if and only if all past hypotheses are null. 
This conservative step is one of the major sources of the gap between the realized FDR of LORD++ and the target level $\alpha$. 
{\color{black}When feedback becomes available over time, it allows us to sharpen the FDP estimator and accordingly reduce the resulting slack.} 
{\color{black}Intuitively, once the $j$th test is confirmed non-null, its contribution \(\alpha_j\) can be removed from the FDP bound and reallocated. 
Our proposed GAIF and its adaptive variant extend LORD++ and SAFFRON by systematically recycling such recovered wealth and reinvesting it into future tests.} 





\subsection{Boosting GAI via feedback: GAIF}\label{subsec:GAIF}
{\color{black}We now propose generalized alpha-investing with feedback (GAIF), which leverages revealed feedback information to improve FDP estimation. The key idea is that once some past hypotheses are known to be non-nulls, their contribution to the FDP upper bound can be reduced accordingly, leaving more $\alpha$-wealth available for future tests. 

Formally, let
\[
\mathcal{I}_t = \{ j\in[t] : \theta_j \text{ is revealed by time } t \}
\]
denote the index set of hypotheses whose labels are available by time \(t\), and let \(\bar{\mathcal{I}}_t = [t] \setminus \mathcal{I}_t\), where \([t] = \{1, 2, \dots, t\}\). 
Realistic scenarios can lead to diverse feedback structures. 
For example, feedback may be available only when a rejection is made or may arrive after a delay \(d \geq 0\). 
Our framework encompasses a variety of practically relevant cases: 
\begin{itemize}
  \item Full and instant feedback: \(\mathcal{I}_t = \{j\in[t] : j \le t-1\}\); 
  \item Bandit and instant feedback: \(\mathcal{I}_t = \{j\in[t] : j \le t-1, \delta_j = 1\}\);
  \item Full and delayed feedback: \(\mathcal{I}_t = \{j\in[t] : j \le t-d-1\}\), where \(d\) is the delay; 
  \item Bandit and delayed feedback: \(\mathcal{I}_t = \{j\in[t] : j \le t-d-1, \delta_j = 1\}\). 
\end{itemize}

At time $t$, the available feedback reveals the true values of $\{\theta_j\}_{j \in \mathcal{I}_t}$. 
We leverage this information to refine the FDP estimator as follows.


\begin{definition}[GAIF]\label{def:GAIF}
 Define the feedback-enhanced FDP estimator as
\[
\widehat{\mathrm{FDP}}_{\mathrm{GAIF}}(t)
:=
\frac{
 \sum_{j \in \mathcal{I}_t} (1 - \theta_j)\,\alpha_j
\;+\;
 \sum_{j \in \bar{\mathcal{I}}_t} \alpha_j
}{
1 \,\vee\, \sum_{j=1}^t \delta_j
}.
\] A procedure is said to be a \emph{GAI with Feedback (GAIF)} procedure if it assigns test levels $\{\alpha_t\}$ such that, for every $t$,
\[
\widehat{\mathrm{FDP}}_{\mathrm{GAIF}}(t)\le \alpha.
\]
\end{definition}

The GAIF framework aligns with intuition. For $j\in\mathcal{I}_t$, feedback reveals whether the corresponding hypothesis is null or non-null, so their contribution to the FDP bound can be accounted for precisely through $(1-\theta_j)\alpha_j$. For indices not yet revealed, i.e.,  $j\in\bar{\mathcal{I}}_t$, we keep the contribution $\alpha_j$.  Thus, GAIF improves the FDP estimator relative to the no-feedback case.

Importantly, this refinement leads directly to a less conservative constraint on the test levels $\alpha_t$. Recall that LORD++ enforces $\alpha_t\leq \alpha(1\vee \sum_{j\leq t} \delta_j)-\sum_{j\leq t-1}\alpha_j$. In contrast, a GAIF procedure selects $\alpha_t$ to satisfy
\begin{equation}\label{eq:alphat_construction_GAIF}
    \alpha_t \le \alpha(1\vee\sum_{j=1}^t\delta_j) - \sum_{j\in\mathcal{I}_t}\alpha_j(1-\theta_j)-\sum_{j\in\bar{\mathcal{I}}_t\setminus \{t\}}\alpha_j.
\end{equation}
The presence of revealed feedback enlarges the admissible upper bound for \(\alpha_t\). Indeed, since \(1-\theta_j \le 1\) for every \(j\in\mathcal{I}_t\), \eqref{eq:alphat_construction_GAIF} implies
\[
\sum_{j\in\mathcal{I}_t}\alpha_j(1-\theta_j)+\sum_{j\in\bar{\mathcal{I}}_t\setminus \{t\}}\alpha_j\le\sum_{j\in\mathcal{I}_t}\alpha_j+\sum_{j\in\bar{\mathcal{I}}_t\setminus \{t\}}\alpha_j = \sum_{j=1}^{t-1}\alpha_j.
\]
Therefore, the upper bound of $\alpha_t$ under GAIF is always at least as large as that under LORD++. 
In this sense, an appropriate procedure with feedback is expected to improve power by sharpening the FDP estimator and releasing more $\alpha$-wealth for future allocation. 

We next present a concrete instance of GAIF as follows:
\begin{equation}\label{eq:alpha-GAIF}
        \alpha_t^{\text{GAIF}}=\gamma_t s_0+(\alpha-s_0)\gamma_{t-\tau_1}\mathbb{I}\{\tau_1<t\}+\alpha\sum_{j:\tau_j<t,\tau_j\neq \tau_1}\gamma_{t-\tau_j}+\sum_{j:j\in\mathcal{I}_t}{\gamma_{t-j}}\alpha_{j}\theta_j.
    \end{equation}
When no feedback is observed, i.e., $\mathcal{I}_t = \emptyset$, the above reduces exactly to the LORD++ update. Consequently, the test levels $\alpha_t^{\text{GAIF}}$ are never smaller than those of LORD++ whenever feedback is available. {\color{black}The last term represents the additional $\alpha$-wealth recovered from revealed non-nulls and reallocated to future tests.}

\textcolor{black}{
The following result states that the $\alpha_t^{\text{GAIF}}$ construction in Eq. (\ref{eq:alpha-GAIF}) satisfies Definition \ref{def:GAIF}. 
\begin{proposition}\label{prop:LF_fdphat_control}
    Assume the testing levels $\alpha_t^{\text{GAIF}}$ are generated by Eq.~(\ref{eq:alpha-GAIF}), the sequence $\{\gamma_j\}_{j=1}^\infty$ is non-negative and satisfies $\sum_{j=1}^{\infty} \gamma_j \leq 1$, the initial wealth parameter satisfies $s_0 \leq \alpha$. Then, the following is satisfied:
    \[\widehat{\mathrm{FDP}}_{\mathrm{GAIF}}(t) \le \alpha.\]
\end{proposition}
}

 
Algorithm~\ref{alg:GAIF-LORD-SAF} summarizes the implementation of this concrete GAIF procedure under a general feedback regime. The only difference across settings lies in how the feedback-available set \(\mathcal{I}_t\) is updated. We denote the resulting procedures under full and instant feedback, bandit and instant feedback, full and delayed feedback, and bandit and delayed feedback by LF, LF-BI, LF-FD, and LF-BD, respectively. 


{\color{black}
\begin{algorithm}[h!]
    \small
    \captionsetup{font=small}
    \caption{GAIF procedures}
    \label{alg:GAIF-LORD-SAF}
\begin{algorithmic}[1]
    \REQUIRE Target FDR level $\alpha$; pre-specified parameters for constructing test levels.
    \FOR{$t = 1,2,\dots$}
        \STATE Observe the $p$-value $p_t$.
        \STATE Update $\mathcal{I}_t$ according to the specific feedback regime.
        \STATE Construct $\alpha_t=\alpha_t^{\mathrm{GAIF}}$ using Definition~\ref{def:GAIF} with the current $\mathcal{I}_t$.
        \STATE \textbf{if} $p_t\le\alpha_t$ \textbf{then} set $\delta_t=1$; \textbf{else} set $\delta_t=0$.
        \STATE Reveal feedback by obtaining $\theta_j$ for all $j\in\mathcal{I}_{t+1}\setminus\mathcal{I}_t$.
    \ENDFOR
    \ENSURE Rejection set $\mathcal{R}=\{t:\delta_t=1\}$.
\end{algorithmic}
\end{algorithm}
}

\subsection{Theoretical guarantee under independence}\label{subsec:finite-FDR-GAIF}

The following result states that  GAIF procedures guarantee online mFDR control under conditional super-uniformity of null $p$-values
\begin{align}\label{csuag}
\Pr\bigl(p_t \le \alpha_t \,\bigm|\, \mathcal{F}_{t-1}\bigr)
\;\le\;
\alpha_t
\quad
\text{for all } t \in \mathcal{H}_0,
\end{align}
where the enlarged filtration {\color{black}$\mathcal{F}_{t-1}:=\sigma(\delta_{1:t-1};\{\theta_j\}_{j\in{\mathcal{I}}_t})$}. Under stronger independence and monotonicity assumptions, GAIF further achieves online FDR control. Here, the test level sequence $\{\alpha_t\}_{t\in\mathbb{N}}$ is said to be a monotonic function of the past if, for all $t \in \mathbb{N}$, $\alpha_t$ is coordinate-wise non-decreasing in the past decisions $\{\delta_j : j < t\}$ and revealed feedback $\{\theta_j : j \in\mathcal{I}_t\}$ for the GAIF procedure.}

{\color{black}
\begin{theorem}[Online mFDR and FDR control for GAIF]\label{the:FDR_GAIF_ind}
Let $\{\alpha_t\}_{t\in \mathbb{N}}$ be a sequence of test levels satisfying $\widehat{\mathrm{FDP}}_{\rm GAIF}(t)\le \alpha$. Then:
\begin{itemize}
    \item[(a)] If the null $p$-values are conditionally super-uniformly distributed as in \eqref{csuag}, then the procedure guarantees 
    \[
        \operatorname{mFDR}(t)\le \alpha \quad \text{for all } t\in\mathbb{N}.
    \]
    
    \item[(b)] If the null $p$-values are mutually independent and also independent of the non-null $p$-values, and if the test levels $\{\alpha_t\}$ are a monotone sequence of functions of the past for all $t$, then the procedure satisfies
    \[
        \operatorname{FDR}(t)\le \alpha \quad \text{for all } t\in\mathbb{N}.
    \]

\end{itemize}
\end{theorem}
 The test levels of GAIF in \eqref{eq:alpha-GAIF} inherit the monotonicity structure from LORD++, and therefore satisfy the monotonicity condition needed for Theorem \ref{the:FDR_GAIF_ind} (b).
The proof of Theorem~\ref{the:FDR_GAIF_ind} is provided in Appendix~\ref{proof:them1}.}

{\color{black}
\section{Extensions of GAIF}\label{sec:exten-GAIF}
In this section, we discuss two important extensions of GAIF. 
We first develop the Adaptive GAIF framework in Section~\ref{subsec:Adaptive-GAIF}, which further enhances GAIF by enabling adaptive \(\alpha\)-wealth allocation. We then extend our methods to the settings with local dependence in Section~\ref{subsec:GAIF-dep}.

\subsection{Adaptive GAIF: improve GAIF via adaptive $\alpha$-wealth allocation}\label{subsec:Adaptive-GAIF}

Section~\ref{subsec:GAIF} shows that feedback improves power by sharpening the FDP estimate. In online testing, power also depends on how the available $\alpha$-wealth is allocated over time, since the test levels $\alpha_t$ must be chosen sequentially. This motivates an improvement direction beyond feedback-enhanced FDP estimation, namely adaptive $\alpha$-wealth allocation.

To motivate this idea, we revisit the difference between LORD++ and SAFFRON. The SAFFRON weighting function $\kappa(p):=\frac{\mathbb{I}\{p>\lambda\}}{1-\lambda}$ has a {\color{black}well-known interpretation} as a null proportion estimator \citep{storey2004strong,dohler2023unified}. Here we highlight a complementary perspective: $\kappa(\cdot)$ also functions as an $\alpha$-wealth allocation rule, concentrating testing budget on hypotheses with small $p$-values. Indeed, under LORD++, the test levels satisfy $\alpha_t \le \alpha \sum_{j=1}^t \delta_j-\sum_{j=1}^{t-1}\alpha_j$. By contrast, SAFFRON enforces
\begin{equation}\label{eq:SAFFRON-alphat-bound-rev}
\alpha_t \le \alpha\,\frac{1-\lambda}{\mathbb{I}\{p_t>\lambda\}}
\sum_{j=1}^t \delta_j
-\frac{1}{\mathbb{I}\{p_t>\lambda\}}\sum_{j=1}^{t-1}\alpha_j\mathbb{I}\{p_j>\lambda\}.
\end{equation}
When $p_t\le \lambda$, the upper bound is interpreted as $+\infty$, permitting aggressive investment at promising tests; when $p_t>\lambda$ and \(p_j>\lambda\) for all \(j\le t-1\), the bound is strictly smaller than the LORD++ bound, ensuring less wealth is spent on unpromising tests. 
{\color{black}This observation 
clarifies the distinct roles of FDP estimation and $\alpha$-wealth allocation in improving power, thereby helping to explain how GAIF and SAFFRON achieve power gains over LORD++ through different mechanisms.
}


{\color{black}However, SAFFRON may allocate $\alpha$-wealth inefficiently when a non-null hypothesis yields a large $p$-value. 
 Specifically, when $p_j>\lambda$, the hypothesis contributes to the FDP estimator through $\kappa(p_j)=1/(1-\lambda)$. In the absence of feedback, this conservative penalty is unavoidable. Once feedback reveals $\theta_j=1$, we know retrospectively that this test carried no false discovery risk. The wealth $\alpha_j$ could have been reallocated to future tests. 
Adaptive GAIF is designed to address this issue. It retains the feedback-based refinement of GAIF while incorporating the SAFFRON-style allocation rule through $\kappa(\cdot)$.} 
\begin{definition}[Adaptive GAIF]\label{def:GAIF-extend}
Define the adaptive feedback-enhanced FDP estimator as
\[
\widehat{\mathrm{FDP}}_{\mathrm{Ada\text{-}GAIF}}(t)
:=
\frac{
 \sum_{j \in \mathcal{I}_t} (1 - \theta_j)\,\alpha_j\kappa(p_j) 
\;+\;
 \sum_{j \in \bar{\mathcal{I}}_t} \alpha_j \kappa(p_j)
}{
1 \vee \sum_{j=1}^t \delta_j
},
\]
where $\kappa(p)=\frac{\mathbb{I}\{p > \lambda\}}{1 - \lambda}$ and $\lambda\in[0,1)$ is a user-chosen parameter for identifying large \(p\)‑values. A procedure is called an \emph{Adaptive GAIF} procedure if it assigns test levels \(\{\alpha_t\}\) such that, for every \(t\),
\begin{equation}\label{eq:adaptive-GAIF-control}
\widehat{\mathrm{FDP}}_{\mathrm{Ada\text{-}GAIF}}(t) \le \alpha.
\end{equation}
\end{definition}


The key difference between GAIF and Adaptive GAIF is the introduction of the $\kappa(p)$ in the FDP estimate. If feedback is ignored so that \(\mathcal{I}_t=\emptyset\), the FDP estimator of Adaptive GAIF reduces to that of SAFFRON.


\begin{remark}
In principle, other constructions of $\kappa(\cdot)$ may further optimize $\alpha$-wealth allocation while ensuring online FDR control. This online optimization remains an open problem and a direction for future work \citep{gang2021structure}. Instead, we focus on the SAFFRON-type weight $\frac{\mathbb{I}\{p > \lambda\}}{1 - \lambda}$, which offers a  theoretically grounded and practically appealing compromise with robust empirical performance.
\end{remark}

A concrete instance of Adaptive GAIF follows the same structure as SAFFRON, augmented with the feedback correction term. {\color{black}For $t = 1$, $\alpha_1^{\text{Ada-GAIF}} = \min\{\gamma_1 s_0 (1 - \lambda),\, \lambda\}$. For $t > 1$,}
\begin{equation}\label{eq:alpha-SAFFRONF}
\begin{aligned}
\alpha_t^{\text{Ada-GAIF}} 
&= \min\Biggl\{\lambda,\;
(1-\lambda)\Biggl[
s_0\, \gamma_{t - C_{0+}}  + (\alpha - s_0)\, \gamma_{t - \tau_1 - C_{1+}}  + \alpha \sum_{j \geq 2} \gamma_{t - \tau_j - C_{j+}}
\Biggr] \\
&\quad + \sum_{j:\, j \in \mathcal{I}_t} 
{\color{black}\gamma_{t - j - C_{j+}^{*}}\, \alpha_j\, \theta_j\, \mathbb{I}\{p_j > \lambda\}}
\Biggr\},
\end{aligned}
\end{equation}
where $C_{j+} = C_{j+}(t) = \sum_{i=\tau_j+1}^{t-1} \mathbb{I}\{p_i \leq \lambda\}$ and {\color{black}$C_{j+}^{*} = C_{j+}^{*}(t) = \sum_{i=j+1}^{t-1} \mathbb{I}\{p_i \leq \lambda\}$. The first term coincides with the standard SAFFRON update; the second term recycles the wealth recovered through feedback.} We denote the resulting procedures under full and instant, bandit and instant, full and delayed, and bandit and delayed feedback by SF, SF-BI, SF-FD, and SF-BD, respectively.

\textcolor{black}{
The Proposition \ref{prop:SF_fdphat_control} states that the $\alpha_t^{\text{Ada-GAIF}}$ construction in Eq. (\ref{eq:alpha-SAFFRONF}) satisfies Definition \ref{def:GAIF-extend}. 
\begin{proposition}\label{prop:SF_fdphat_control}
    Assume the testing levels $\alpha_t^{\text{Ada-GAIF}}$ are generated by Eq.~(\ref{eq:alpha-SAFFRONF}), the sequence $\{\gamma_j\}_{j=1}^\infty$ is non-negative and satisfies $\sum_{j=1}^{\infty} \gamma_j \leq 1$, and the initial wealth parameter satisfies $s_0 \leq \alpha$. Then, the following is satisfied:
    \[\widehat{\mathrm{FDP}}_{\mathrm{Ada\text{-}GAIF}}(t) \le \alpha.\]
\end{proposition}
}

Figure~\ref{fig:alpha_Gaussian} depicts the average testing thresholds $\{\alpha_t\}$ over time for LF, SF, LORD++, SAFFRON and LOND \citep{javanmard2015online} applied to Gaussian observations under full and instant feedback. {\color{black}Across all non-null proportions $\pi_1$}, both LF and SF yield strictly larger thresholds than their feedback-free counterparts, confirming that feedback enables more effective $\alpha$-wealth utilization. 
{\color{black}When $\pi_1$ is small, revealed non-nulls are rare and feedback-recovered wealth is limited; adaptive $\alpha$-wealth allocation then dominates, and SF outperforms LF---at $\pi_1=0.2$. 
As $\pi_1$ increases, feedback becomes more informative and wealth recycling grows, allowing LF to surpass SF.}

\begin{figure}[htbp!]
		\centering
		\includegraphics[width=\textwidth]{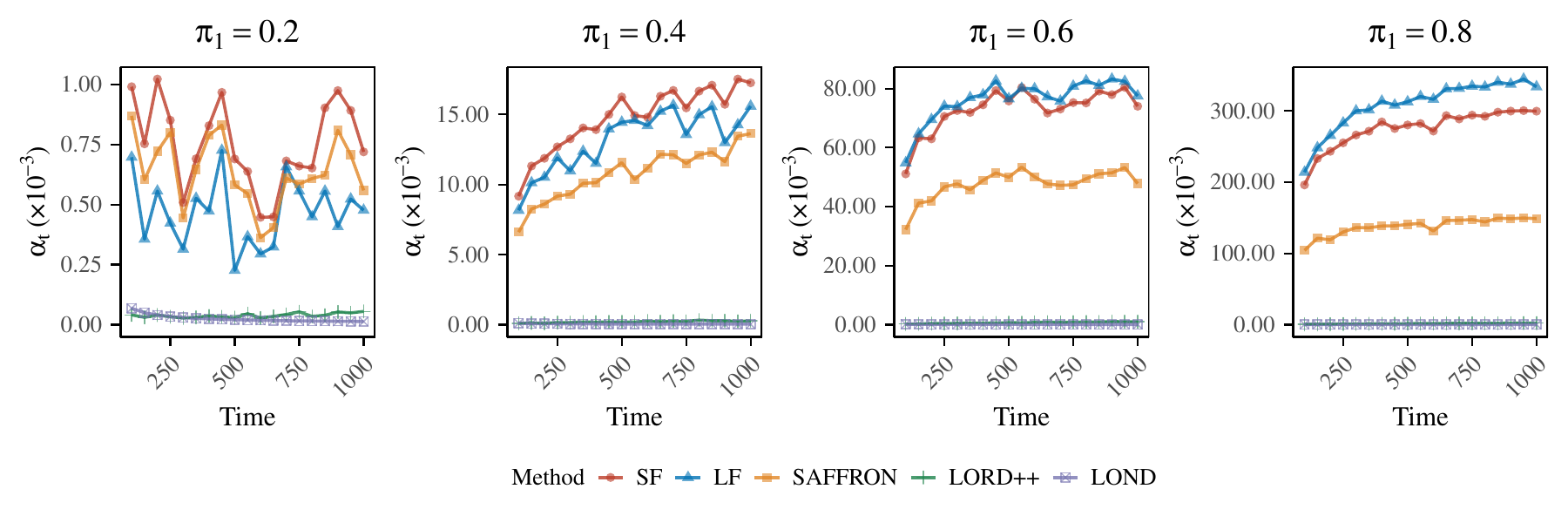}
		\caption{\small Average testing thresholds $\alpha_t$ over time under various procedures, based on 500 replications across non-null proportions {\color{black}$\pi_1 \in \{0.2, 0.4, 0.6, 0.8\}$}. Data are generated under Scenario~I of Section~\ref{subsec:simu-GAIF} with full and instant feedback and signal strength $\mu = 2$. }
  \label{fig:alpha_Gaussian}
	\end{figure}


{\color{black}
We now establish the theoretical results for Adaptive GAIF. Similar to \eqref{csuag}, the conditional super-uniformity condition on the null $p$-values is
    \begin{align}\label{csuag-SF}
\Pr\bigl(p_t \le \alpha_t \,\bigm|\, \mathcal{J}_{t-1}\bigr)
\;\le\;
\alpha_t
\quad
\text{for all } t \in \mathcal{H}_0,
\end{align}
where {\color{black}$\mathcal{J}_{t-1}:=\sigma(\delta_{1:t-1};C_{1:t-1};\{\theta_j\}_{j\in{\mathcal{I}}_t})$} is the enlarged filtration, where $C_j = \mathbb{I}\{p_j \leq \lambda\}$. 
And we extend the monotonicity definition in Theorem \ref{the:FDR_GAIF_ind} by additionally requiring $\alpha_t$ to be coordinate-wise non-decreasing in the indicators $\{C_j : j < t\}$. 
\begin{theorem}[Online mFDR and FDR control for Adaptive GAIF]\label{the:FDR_Ada_GAIF_ind}
Let $\{\alpha_t\}_{t\in \mathbb{N}}$ be a sequence of test levels satisfying  $\widehat{\mathrm{FDP}}_{\rm Ada\text{-}GAIF}(t)\le \alpha$. Then:
\begin{itemize}
    \item[(a)] If the null $p$-values are conditionally super-uniform \eqref{csuag-SF}, the procedure guarantees 
    \[
        \operatorname{mFDR}(t)\le \alpha \quad \text{for all } t\in\mathbb{N}.
    \]
    
    \item[(b)] If the null $p$-values are mutually independent and independent of the non-null $p$-values, and if the test levels $\{\alpha_t\}$ form a monotone sequence of functions of the past for all $t$, then the procedure satisfies
    \[
        \operatorname{FDR}(t)\le \alpha \quad \text{for all } t\in\mathbb{N}.
    \]

\end{itemize}
\end{theorem}
The test levels of Adaptive GAIF in \eqref{eq:alpha-SAFFRONF} inherit the monotonicity structure from SAFFRON, and therefore satisfy the monotonicity condition needed for Theorem \ref{the:FDR_Ada_GAIF_ind} (b).
The proof of Theorem~\ref{the:FDR_Ada_GAIF_ind} is provided in Appendix~\ref{proof:them1-Ada-GAIF}.
}

\subsection{GAIF under local dependence}\label{subsec:GAIF-dep}

This subsection focuses on extending the GAIF and Adaptive GAIF framework to settings with local dependence (as defined in Definition \ref{def:local_dep}). This extension is useful in applications where nearby tests may be dependent, while long-range dependence is absent or weak. For a strategy that addresses more general dependence structures via e-values, the reader may refer to Appendix \ref{app:exten-GAIF}.

\begin{definition}[Local dependence; \cite{Zrnic2021asynchronous}]\label{def:local_dep}
    We say that $p$-values $p_1,p_2,\dots,p_t,\dots$ are locally dependent if
    \begin{equation}\label{eq:local-dep}
        \text{for all}\; t>0,\; \text{there exists}\; L_t\in\mathbb{N}\; \text{such that}\; p_t\perp p_{t-L_t-1},p_{t-L_t-2},\dots,p_1,
    \end{equation}
    where $\{L_t\}_{t\in\mathbb{N}}$ is a fixed sequence of parameters which we refer to as lags.
\end{definition}

 \cite{Zrnic2021asynchronous} pioneered a strategy to handle this local dependence in online multiple testing, leading to the development of the $\text{LORD}_{\text{dep}}$ and $\text{SAFFRON}_{\text{dep}}$ methods. 
{\color{black}We extend the idea to the feedback-enhanced setting. Under local dependence, 
the dependence--adjusted FDP estimates for GAIF and Adaptive GAIF are, respectively: 
\begin{align*}
\widehat{\operatorname{FDP}}_{\mathrm{GAIF}_{\mathrm{dep}}}(t)
&=\frac{\sum_{j\in \mathcal{I}_t}(1-\theta_j)\alpha_j+\sum_{j\in\bar{\mathcal{I}}_t}\alpha_j}
{ \bigl(\sum_{j\leq t,\; j\notin \{t-L_t,\dots,t-1\}}\delta_j \bigr)\vee 1},
\end{align*}
and
\begin{equation*}
\widehat{\operatorname{FDP}}_{\mathrm{Ada\text{-}GAIF}_{\mathrm{dep}}}(t) = \frac{ \sum_{j < t-L_t} \frac{\alpha_j \mathbb{I}\{p_j>\lambda\}}{1-\lambda} \mathcal{W}_j + \sum_{t-L_t \le j \le t} \frac{\alpha_j}{1-\lambda} \mathcal{W}_j }{ \bigl(\sum_{j\leq t,\; j\notin \{t-L_t,\dots,t-1\}}\delta_j \bigr)\vee 1 },
\end{equation*}
 where $\mathcal{W}_j = (1-\theta_j)\mathbb{I}\{j \in \mathcal{I}_t\} + \mathbb{I}\{j \in \bar{\mathcal{I}}_t\}$. The corresponding testing levels are chosen so that 
$\widehat{\mathrm{FDP}}_{\mathrm{GAIF}_{\mathrm{dep}}}(t)\leq\alpha$ or $\widehat{\mathrm{FDP}}_{\mathrm{Ada\text{-}GAIF}_{\mathrm{dep}}}(t)\leq\alpha$. The resulting testing levels are provided in Appendix~\ref{appendix:testing-levels}.
Theorem \ref{the:mFDR-control-dep} shows that $\text{GAIF}_{\rm dep}$ and $\text{Ada-GAIF}_{\rm dep}$ methods control mFDR under local dependence.}
\begin{theorem}[Online mFDR control under local dependence]\label{the:mFDR-control-dep}
Suppose that the null $p$-values are locally dependent as defined in (\ref{eq:local-dep}). Then if the test levels $\{\alpha_t\}_{t\in\mathbb{N}}$ satisfy {\color{black}$\widehat{\operatorname{FDP}}_{\operatorname{GAIF}_{\rm{dep}}}(t)\leq \alpha$ or $\widehat{\operatorname{FDP}}_{\operatorname{Ada-GAIF}_{\rm{dep}}}(t)\leq \alpha$}, we have
$\operatorname{mFDR}(t)\leq \alpha$ for all \(t \in \mathbb{N}\).
\end{theorem}

\section{Applications of GAIF on Online Conformal Selection}\label{sec:OCT}

Online conformal selection \citep{huo2024realtime,xu2024online} provides a canonical and intuitive instantiation of the GAIF framework. Here, decisions must be made in real time about whether an incoming observation satisfies a pre-specified requirement, while the true label is typically observed after the decision and thus serves as feedback. \textcolor{black}{Throughout this section, we focus on the full and instant feedback setting.} 

Consider a data pair \((\X, Y) \in \gX \times \gY\), with a historical calibration dataset \(\mathcal{D}_{\gC} = \{(\X_i, Y_i)\}_{i=-n+1}^0\) of size $n$, with index set $\gC$. Test samples \((\X_t, Y_t)\) arrive sequentially for \(t = 1,2,\dots\). At each time \(t\), the covariate \(\X_t \) is observed but responses \(Y_t\) remain hidden until a real-time decision \(\delta_t \in \{0,1\}\) is made. The goal is to determine whether \(Y_t\) lies in a target region \(\gA \subseteq \gY\), for example, \(\gA=[a,b]\)  or \(\gA=[b,\infty)\) in regression tasks. This can be framed as an online testing problem with \(\theta_t = \mathbb{I}\{Y_t \in \gA\}\). {\color{black}The GAIF framework developed in the preceding sections offers a principled basis for analyzing this setting and enables rigorous theoretical guarantees.} 

\textcolor{black}{We apply the GAIF framework to online conformal selection as follows. Since no off-the-shelf \(p\)-values are available, we first construct feedback-assisted online conformal \(p\)-values in Section~\ref{subsec:p-val-construct}. We then incorporate them into a modified GAIF (or Adaptive GAIF) rule to obtain finite-sample mFDR control in Section~\ref{subsec:OCTF}. We further use feedback to adaptively select the score model for conformal testing in Section~\ref{subsec:Opt-OCTF}, and analyze its optimality in Section~\ref{subsec:optimality}.}

\subsection{Construction of online conformal $p$-values} \label{subsec:p-val-construct}

\textcolor{black}{We begin by constructing suitable online \(p\)-values for this online conformal selection setting.} 
Let $V(\X)$ be a non-conformity score, where larger values indicate a higher likelihood that $\theta_i=0$. Typically, $V(\X)$ is a monotone transformation of the prediction $\widehat{\mu}(\X)$, assumed pre-trained to estimate \(Y_t\). For example, if $\theta=\mathbb{I}\{Y_t \geq b\}$, one can take $V(X)=b-\widehat{\mu}(\X)$. For simplicity, we write $V_i=V(\X_i)$. 

A natural approach to achieving online FDR control is to compute conformal \( p \)-values with the given calibration set $\gC$ \citep{bates2021testing, jin2023selection} and then apply GAIF or existing GAI rules. However, 
standard conformal \( p \)-values are not suitable here, since they do not satisfy the conditional super-uniformity in \eqref{csuag} and independence assumptions.
To circumvent these issues, we adopt \emph{online} conformal $p$-values, which are constructed by sequentially updating the calibration set. 
Specifically, at time $t$, the current point \((\X_{t}, Y_{t})\) serves as the test point, but from time \(t+1\) onward, it is added to the calibration set. Let $\gC_{0t}$ denote the dynamically updated index set of calibration samples with 
\begin{equation}
    \gC_{0t} = \{ -n<i<t : \theta_i=0 \}.\nonumber
\end{equation}
For test data $\X_{t}$ at time $t$, the online conformal $p$-value $p_t$ is defined as  
\begin{equation}\label{eq:conf_p}
    p_t=\frac{\sum_{i\in{\gC}_{0t}} \I\{V_{i}<V_{t}\}+\xi_t\cdot \left[1+\sum_{i\in \gC_{0t}} \mathbb{I}\{V_{i}=V_{t}\}\right]}{1+|\gC_{0t}|},
\end{equation}
where  $\xi_t\overset{\text{i.i.d.}}{\sim} \operatorname{Unif}[0,1]$ is an independent random variable for tie-breaking. 
Similar forms of online conformal \( p \)-values have been studied for testing exchangeability \citep{vovk2003testing, vovk2021testing} and constructing online conformal prediction intervals \citep{angelopoulos2024theoretical}. 

Under exchangeability, this construction yields null $p$-values that are independent, thereby avoiding the complex dependence induced by  structures inherent in offline conformal methods, which rely on 
shared calibration sets and produce $p$-values that are PRDS \citep{bates2021testing}.
We impose the following assumption.
\begin{assumption}[Exchangeability in conformal setting]\label{assump:conformal_setting}
    The null data $\big\{(\X_i,Y_i):  i> -n,\theta_i=0\big\}$ are exchangeable conditional on the non-null data $\big\{(\X_i,Y_i):  i> -n,\theta_i=1\big\}$.
\end{assumption}
This assumption is common in conformal inference \citep{adadetect}, and is weaker than requiring that the full data $\{(\X_i,Y_i)\}_{i>-n}$ are independent and identically distributed (i.i.d.). We formally state the validity and independence under the null in Proposition~\ref{prop-inde-p}. 

\begin{proposition}[{\color{black}Validity and independence of online conformal \( p \)-values}]\label{prop-inde-p}
Suppose Assumption~\ref{assump:conformal_setting} holds. Then under the null, the online conformal \( p \)-value \( p_t \) defined in \eqref{eq:conf_p} is uniformly distributed on \([0,1]\), and the null p-values \(\{p_t:t\in\mathbb{N},\theta_t=0\}\) are mutually independent.
\end{proposition}

\textcolor{black}{The mutual independence above indeed relies on the full and instant feedback setting, since only in this case does \(\gC_{0t}\) contain the complete set of past null samples, allowing the sequential rank-based independence argument to hold. For delayed feedback, additional adjustments are needed; see the finite-sample analysis in Appendix~\ref{appen:OCTF-delay}. We also allow the score function to vary over time and be data-dependent, provided that it remains symmetric on the null samples; see Appendix~\ref{appen_sub:online_cp}. This flexibility enables leveraging non-null feedback to improve score quality, facilitating score selection and adaptation to distribution shift.
}

\subsection{Online conformal testing with feedback}\label{subsec:OCTF}
Having established the properties of online conformal \( p \)-values, we now introduce the corresponding testing procedures. Although applying GAIF directly to those \( p \)-values yields satisfactory FDR control empirically, a finite-sample theoretical guarantee remains lacking. The key challenge is that null online conformal \( p \)-values may still be dependent on non-null decisions through the adaptive testing levels.

{\color{black}
To resolve this, we construct safe variants of LF and SF in which the rejection rewards are triggered only by confirmed null rejections. Specifically, we replace the usual rejection times $\tau_j$ with $\tilde{\tau}_j$, the time of the $j$-th null rejection: 
\[
\tilde{\tau}_j
=
\inf\left\{
t\in\mathbb{N}:
\sum_{i\leq t}\delta_i(1-\theta_i)\geq j
\right\}.
\]
Since $\tilde{\tau}_j$ is determined solely by null rejections, the resulting test levels do not depend on non-null decisions. We denote the resulting procedures by LFS and SFS, where ``S'' stands for ``safe.'' The corresponding \emph{Online Conformal Testing with Feedback} (OCTF) procedure is given in Algorithm~\ref{alg:OCTF-ind-conservative}, and the following theorem establishes its finite-sample mFDR control.}

 \begin{algorithm}[h!]
    \small
    \captionsetup{font=small}
    \caption{Online conformal testing with feedback (OCTF)}
    \label{alg:OCTF-ind-conservative}
\begin{algorithmic}[1]
    \REQUIRE{Initial calibration data $\gD_\gC=\{(\X_i, Y_i)\}_{i=-n+1}^0$, target region $\mathcal{A}$, FDR target level $\alpha \in (0,1)$, non-conformity score function $V(\cdot)$, parameter $s_0$,  parameter sequence $\{\gamma_t\}_{t\in\mathbb{N}}$,  constant $\lambda\in(0,1)$.} \\[0.5ex]
      \STATE Initialize $\gC_{01}=\{i\in\gC:\theta_i=0\}$.
    \FOR{$t=1,2,\ldots$}  
    \STATE Observe test covariate $\X_{t}$.
        \STATE Compute conformity scores $\{V_{i}\}_{i\in\gC_{0t}\cup\{t\}}$.
        \STATE Compute online conformal $p$-value $p_t$ via~\eqref{eq:conf_p}.
      \STATE  Update $\alpha_t={\alpha}_t^{\text{LFS}}$ in \eqref{eq:alpha-LORDF-conservative} (or $\alpha_t={\alpha}_t^{\text{SFS}}$ in \eqref{eq:alpha-SAFFRONF-conservative}).
    \STATE \textbf{if} $p_t\leq \alpha_t$ \textbf{then} $\delta_t=1$, \textbf{else} $\delta_t=0$. 
    \STATE Obtain the revealed feedback $Y_{t}$, and thus obtain $\theta_t$.
    \STATE Update the calibration dataset $\gC_{0t}$.
    \ENDFOR 
    \ENSURE{Rejection set $\mathcal{R}=\{t:\delta_t=1\}$.}
\end{algorithmic}
\end{algorithm}

\begin{theorem}[Finite-sample online mFDR control for OCTF]\label{the:FDR-OCTF}
     Suppose Assumption \ref{assump:conformal_setting} holds.
     The output $\gR$ of Algorithm \ref{alg:OCTF-ind-conservative} satisfies $\mFDR(t)\leq \alpha$ for all $t\in\mathbb{N}$.
\end{theorem}

{\color{black}

Theorem~\ref{the:FDR-OCTF} is established under a dependence regime that naturally arises from online conformal $p$-values: the null $p$-values are mutually independent, whereas their dependence with non-null $p$-values is generally unspecified. This regime is not directly covered by the standard assumptions used in classical online testing procedures such as LOND, LORD, and SAFFRON. 
In this setting, LFS and SFS  trade some potential power for a rigorous theoretical guarantee. LF and SF retain the more operationally natural reward structure and serve as more aggressive alternatives when empirical power is the primary consideration. Although they maintain empirical FDR control in our simulations (Section~\ref{sec:simu}), their finite-sample validity under this dependence regime remains open.

\begin{remark}
Theorem~\ref{the:FDR-OCTF} controls mFDR rather than FDR. Although the safe construction removes past non-null rejections from the threshold update, they remain in the random FDP denominator and may be dependent on a current null $p$-value. The resulting dependence makes extending mFDR control to FDR nontrivial 
\citep{Zrnic2021asynchronous}. Finite-sample FDR control under the present assumptions therefore remains open.
\end{remark}

}

\subsection{Optimized online conformal testing with feedback}\label{subsec:Opt-OCTF}

The performance of online conformal testing depends on the choice of the non-conformity scores \( V(\X) \) and the prediction model, e.g., random forests, neural networks, support vector machines, or regularized linear models \citep{bai2024optimized,gasparin2024conformal}. In practice, the most effective score may also change over time, especially when the non-null distribution drifts \citep{gang2021structure,huo2024realtime}. This motivates a feedback-driven score-selection step built on top of the OCTF framework.

{\color{black}Suppose there are \( K \) pre-trained candidate score functions \( \{V(\cdot; k) : \gX  \to \mathbb{R} \}_{k=1}^K \). At each time $t$, a score is chosen as 
\[
\widehat{k}_t = \underset{k \in [K]}{\arg\min} \;\gM(\gD_t;k),
\]
where $\gM(\gD_t;k)$ is a given criterion evaluated by the currently observed data $\gD_t=\{(\X_i,Y_i):-n< i< t\}\cup\{\X_t\}$. 
Our goal is to choose the score that is most favorable for detecting non-nulls, which motivates the design of a reliable criterion $\gM(\gD_t;k)$. Since the testing threshold \(\alpha_t\) depends recursively on the past rejection path, directly optimizing rejection probability is difficult. We therefore use the proxy $$\E[p_{k,t} \mid \theta_t=1],$$
where $p_{k,t}$ is the $t$-th conformal $p$-value obtained from the score function $V(\cdot;k)$. Smaller non-null \(p\)-values generally correspond to higher rejection probability $\Pr(p_{k,t}\le \alpha_t\mid \theta_t=1)$, so this provides a natural threshold-free criterion.}


{\color{black}To estimate $\E[p_{k,t}\mid\theta_t=1]$ under shifted non-null distribution while preserving symmetry, we employ an exponentially weighted moving average (EWMA) of past auxiliary non-null \(p\)-values. Specifically, at time $t$, for each score function $k$, we construct \textit{auxiliary non-null $p$-values} as
\begin{equation}\label{eq:aux_nonnull_p_value}
\tilde{p}_{k,j} = \frac{\sum_{s \in \gC_{0t} \cup \{t\}} \mathbb{I}\{V(\mathbf{X}_s; k) \leq V(\mathbf{X}_j; k)\}}{1 + |\gC_{0t}|}, \quad j \in \gC_{1t},
\end{equation}
where $\gC_{1t}=\{-n<i< t:\theta_i=1\}$ denotes the set of online non-null samples observed prior to time $t$. When the non-null distribution varies slowly over time, these auxiliary non-null $p$-values are approximately distributed like the current non-null $p$-value $p_{k,t}$. The construction in \eqref{eq:aux_nonnull_p_value} is carefully designed so that the resulting auxiliary non-null $p$-values are invariant under permutations of $\gC_{0t} \cup \{t\}$, which is crucial for ensuring valid inference after score selection. We then define the EWMA criterion as
\begin{equation*}
\gM^{\mathrm{EWMA}}_t(\gD_t;k)=\frac{\sum_{j\in\gC_{1t}} \rho^{t-1-j}\, \tilde{p}_{k,j}\,}{\sum_{j\in\gC_{1t}}\rho^{t-1-j}},
\end{equation*}
where \(\rho\in (0,1)\) is the user-specified decay parameter that downweights past observations, enabling dynamic adaptation to the recent non‑null distribution and precise estimation of $\E[p_{k,t}\mid\theta_t=1]$. The selected score function is accordingly
as $\widehat{k}_t={\arg\min}_{k \in [K]}\gM^{\mathrm{EWMA}}_t(\gD_t;k).$ 

The online conformal \(p\)-value after score selection, denoted as $p^{\rm opt}_t$, is then computed by \eqref{eq:conf_p} using the selected score $V(\mathbf{X}_t;\widehat{k}_t)$. 
Then running the OCTF procedures in Section \ref{subsec:OCTF} with these optimized $p$-values yields the whole optimized procedure, which we refer to as Opt-OCTF algorithm. }


\begin{corollary}[{\color{black}Finite-sample online mFDR control for Opt-OCTF}]\label{the:FDR-Opt} Under the assumptions in Theorem \ref{the:FDR-OCTF}, the Opt-OCTF procedure by $\gM^{\rm EWMA}_t$ satisfies 
$\mFDR(t)\leq \alpha$ for all $t\in\mathbb{N}.$
\end{corollary}

{\color{black}

Although Corollary~\ref{the:FDR-Opt} is stated for the EWMA criterion, the underlying framework for optimized online conformal testing can be extended more broadly. In particular, we allow more general selection criteria \(\gM\) by using a full permutation strategy over \(\gC_{0t}\cup\{t\}\); see Appendix~\ref{appen:model-selection}. In addition, data-driven score functions obtained by online updating 
are also compatible with the framework under a suitable symmetry-preserving training scheme. However, we focus on selection among \(K\) pre-trained candidates in the main text for computational reasons; 
see Appendix~\ref{appen:OCTF-onlinelearning} for details. 

\begin{remark}
In \eqref{eq:aux_nonnull_p_value}, we include the current time point \(t\) in the null calibration set to preserve symmetry over \(\gC_{0t}\cup\{t\}\), which is essential for the validity of the optimized conformal \(p\)-value after score selection. Although this may slightly affect calibration when \(t\) is non-null, the effect is typically minor in practice. A simple truncated variant can further reduce this influence while preserving validity; see Appendix~\ref{appen_subsec:simu_auxi_modelsel} for details.
\end{remark}

\subsection{Optimality analysis of EWMA criterion}\label{subsec:optimality}

We now give a supporting analysis for the EWMA criterion by establishing certain optimality of the online score selection step.
To facilitate analysis, we work under the standard two-group model, which is widely used to study the performance of online testing procedures \citep{gang2021structure,humbert2025onlineselectiveconformalinference}.

\begin{assumption}[Two-group model]\label{assum:two group}
    Suppose the score functions follow the two-group model for each $k\in[K]$:
$$V(\X_t;k)\overset{\rm indep}{\sim} \pi_t F_{k}^{(0)}+(1-\pi_t)F^{(1)}_{k,t},\quad \theta_t\overset{\rm indep}{\sim} \text{Bernoulli}(1-\pi_t),$$
where $F_{k}^{(0)}$ and $F^{(1)}_{k,t}$ are the null and (time-varying) non-null score distributions for $k$-th model. 
\end{assumption}

Under this model, the oracle conformal $p$-value associated with score $V(\cdot;k)$ is $F_{k}^{(0)}(V(\X_t;k))$. Given $\theta_t=1$, a smaller value of $F_{k}^{(0)}(V(\X_t;k))$ indicates higher detection power. And we define the optimal score function at time  $t$ as
$$k_t^*=\arg\min_{k\in[K]}\E[F_{k}^{(0)}(V(\X_t;k))\mid\theta_t=1].$$
Note that $\E[F_{k}^{(0)}(V(\X_t;k))\mid\theta_t=1]$ is the population analogue of the empirical proxy $\E[p_{k,t}\mid\theta_t=1]$. 
To show that the EWMA criterion tracks this oracle target and selects ${k}^*_t$
 reliably, we impose standard regularity conditions.

\begin{assumption}[Slowly varying distribution]\label{assum:slowdrift}
    For each $k\in[K]$ and $t\in\mathbb{N}$, $F^{(1)}_{k,t}$ has a slowly varying distribution $\|F^{(1)}_{k,t}-F^{(1)}_{k,t-1}\|_\infty\leq \gamma$ such that $\gamma\to 0$.
\end{assumption}

\begin{assumption}[Non-degeneracy]\label{assum:non_degen}
There exist constants \(0<\underline{\pi}\le \overline{\pi}<1\) such that null proportion satisfies $\pi_t\in[\underline{\pi},\overline{\pi}]$.
\end{assumption}

\begin{assumption}[Margin]\label{assum:margin}
    For any $t$ and a constant $c>0$, $\min_{i\neq k_t^*}|\E[F_i^{(0)}(V(\X_t;i))\mid\theta_t=1]-\E[F_{k_t^*}^{(0)}(V(\X_t;k_t^*))\mid\theta_t=1]|\geq c$.
\end{assumption}
Assumption~\ref{assum:slowdrift} allows the non-null distribution $F^{(1)}_{k,t}$ to drift over time. 
Assumption~\ref{assum:non_degen} ensures that neither nulls nor non-nulls become too rare over time. Assumption~\ref{assum:margin} guarantees the optimal score is identifiable. 

\begin{theorem}[Optimality of online score selection]\label{the:opt-FDR-OCTF}
Suppose Assumptions \ref{assum:two group}-\ref{assum:non_degen} hold.
\begin{enumerate}
    \item[(i)]   If $t+n\geq 1+1/(1-\rho)$, then for any $\varepsilon_1\in\left(0,\sqrt{\underline{\pi}/2}\right)$ and $\varepsilon_2\in(0,1-\rho)$,
     \begin{align*}
        &\Pr\left(|\gM^{\rm EWMA}(k,\gD_t)-\E[F^k_0(V(\X_t;k))\mid\theta_t=1]|\leq \varepsilon_1+\varepsilon_2+\frac{4\gamma}{(1-\overline{\pi})(1-\rho)}+\frac{2}{(t+n-1)\underline{\pi}}\right)\\
        &\geq 1-3\exp\left\{-(t+n-1)\underline{\pi}\varepsilon_1^2\right\}-4\exp\left\{-\frac{\varepsilon_2^2({1-\overline{\pi}})^2}{6(1-\rho)}\right\}.
    \end{align*} 

Further assume that Assumption \ref{assum:margin} holds. Then the following results in (ii)-(iii) hold.
\item[(ii)] If $t+n\geq 1+{16}/({c\underline{\pi}})$ and ${\gamma}/({1-\rho})\leq {(1-\overline{\pi})c}/{32}$, we have the finite-sample error bound for the score selection
$$\Pr(\hat k_t\neq k_t^*\mid \theta_t=1)
\le\;
3K\exp\!\left\{-(t+n-1)\underline\pi\frac{c^2}{64}\right\}+
4K\exp\!\left\{-\frac{c^2(1-\overline\pi)^2}{384(1-\rho)}\right\}.$$

\item[(iii)] Moreover, the score selection consistency for all non-null points holds as
$$\Pr(\exists i\le t:\theta_i=1,\ \hat k_i\neq k_i^*)\leq\frac{3K(1-\underline{\pi})\exp\!\left\{-n\underline\pi{c^2}/{64}\right\}}{1-\exp\!\left\{-\underline\pi{c^2}/{64}\right\}}
+4K(1-\underline{\pi})t\exp\!\left\{-\frac{c^2(1-\overline\pi)^2}{384(1-\rho)}\right\}.$$
In particular, if $(1-\rho)\log (t)\to0$, $\gamma/(1-\rho)\to 0$ and $n\to\infty$, $\Pr(\exists i\le t:\theta_i=1,\ \hat k_i\neq k_i^*)\to 0$.
\end{enumerate}
 
\end{theorem}
Theorem \ref{the:opt-FDR-OCTF} (i) shows that the EWMA criterion concentrates around the oracle power proxy. The bound decomposes naturally into: an estimation term scaling with $t+n$, a non-null effective size term in EWMA operation by $(1-\rho)^{-1}$ and a drift bias of order $\gamma/(1-\rho)$ reflecting how well the EWMA estimator can track a drifting non-null distribution. Parts (ii)–(iii) translate this concentration into optimal score selection guarantees. The rate conditions reflect the classic EWMA trade-off: $\rho$ must be sufficiently close to 1 to ensure an adequate effective sample size, yet small enough to allow the procedure to adapt to distribution drift. 


}

\section{Numerical Experiments}\label{sec:simu}
In this section, we present extensive synthetic experiments to demonstrate the validity and efficiency of our proposed methods. 
First, Section~\ref{subsec:simu-GAIF} studies GAIF and Adaptive GAIF in the non-conformal setting. 
Second, Sections~\ref{subsec:simu-OCTF} and~\ref{subsec:simu-opt-OCTF} present the results for OCTF and Opt-OCTF under the conformal setting, respectively. 
Following the setup in \citet{robertson2023online}, we fix $\lambda = 0.5$, $\gamma_j \propto j^{-1.6}$ for all $j\in\mathbb{N}$, and $s_0 = \alpha/2$ throughout. 
{\color{black}Code for implementing our methods and reproducing the experiments and figures in our paper is
available at \url{https://github.com/lulin2023/GAIF}.}

\subsection{Synthetic experiments under non-conformal settings}\label{subsec:simu-GAIF}
{\color{black}We begin by presenting experiments on online multiple testing problems under three scenarios. Throughout, let $\pi_1$ denote the non-null proportion parameter in the data-generating model; its value varies across experiments.}
\begin{itemize}
    \item {\bf Scenario I (Testing with Gaussian observations)}
We simulate $T$ independent test statistics $Z_t\sim N(\mu_t,1)$ and test hypotheses $\mathbb{H}_{0t}:\mu_t=0$ using one-sided $p$-values $p_t=\Phi(-Z_t)$, where $\Phi$ is the standard Gaussian CDF. The signal strengths $\mu_t$ follow a mixture model:
 \begin{equation}
    \mu_t = \begin{cases}
        0 \quad \text{with probability}\; 1-\pi_1\\
        F_1 \quad \text{with probability}\; \pi_1,
    \end{cases}
\end{equation}
 where the random variable $F_1\sim N(2.5,1)$.
 \item {\bf Scenario II (Testing with Beta alternatives)} We generate $p$-values according to the following model:
\begin{equation}
    p_t \sim \begin{cases}
        \text{Unif}[0,1] \quad \text{with probability}\; 1-\pi_1\\
        \text{Beta}(0.5,4) \quad \text{with probability}\; \pi_1.
    \end{cases}
\end{equation}
\item {\bf Scenario III (Testing under local dependence)} We simulate correlated test statistics $(Z_1, \dots, Z_T)^\top \sim N(\boldsymbol{\mu}, \Sigma)$ where $\boldsymbol{\mu}=(\mu_1,\ldots,\mu_T)^\top$ with $\mu_t=2$ for a randomly chosen fraction $\pi_1$ of indices and $\mu_t=0$ otherwise. The covariance matrix $\Sigma$ has a block-diagonal structure: coordinates are grouped into blocks of size $n_{\mathrm{block}}=10$, with within-block correlation $\rho=0.8$ and diagonal elements as 1.
We test hypotheses $\mathbb{H}_{0t}:\mu_t=0$ using one-sided $p$-values $p_t=\Phi(-Z_t)$.


\end{itemize}

We set the stopping time as $T = 1000$ and the FDR level at \(\alpha = 0.1\). We evaluate the performance via empirical FDR and power averaged over 500 independent replications. Appendix~\ref{appen:add-experiment} confirms that mFDR and FDR show similar trends.  

{\color{black}Firstly, we consider the full and instant feedback setting, and compare our methods \text{SF} and \text{LF} with state-of-the-art algorithms for online FDR control, namely \text{LOND} \citep{javanmard2015online}, \text{LORD++}, and \text{SAFFRON}, using their default parameters (\(\lambda = 0.5\) chosen for \text{SAFFRON}).} 

Figure \ref{fig:GAIF_Gaussian} shows results for varying non-null proportion $\pi_1\in[0.1,0.8]$ under Scenarios I and II. {\color{black}All benchmark methods 
maintain valid empirical FDR control across different $\pi_1$ under both Scenarios.} Our feedback-based SF and LF methods enhance detection power while preserving valid empirical online FDR control:  SF consistently outperforms SAFFRON, while LF yields higher power than LORD++. {\color{black}When \(\pi_1\) is small, SF benefits from adaptive wealth allocation and tends to outperform LF. When \(\pi_1\) is large, LF gains more from frequent non-null feedback and can recycle more released \(\alpha\)-wealth and thereby surpassing SF in power.} In contrast, LORD++ and LOND remain conservative across different values of \(\pi_1\). {\color{black}Numerical values together with Monte Carlo standard errors are reported in Appendix \ref{appen:add-experiment} (Table \ref{tab:sim_results}).
}

\begin{figure}[h!]
		\centering
		\includegraphics[width=0.75\textwidth]{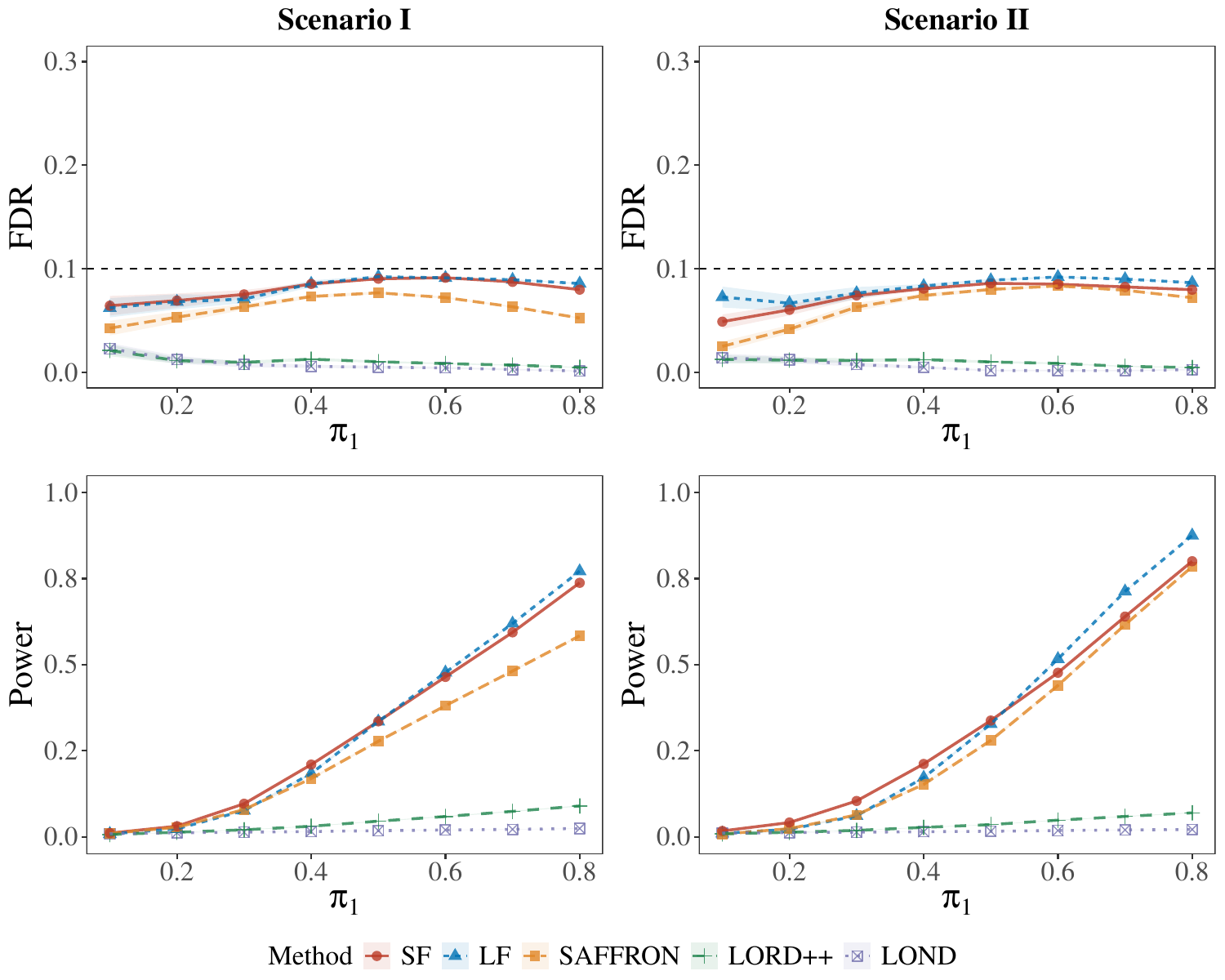}
		\caption{\small Results for Scenario I and Scenario II under full and instant feedback: FDR and Power at stopping time with non-null proportion $\pi_1$ ranging from $0.1$ to $0.8$ across $500$ replications;  The black dashed lines denote the FDR level $\alpha=0.1$. Shaded areas show $\pm 1$ standard error.}
  \label{fig:GAIF_Gaussian}
	\end{figure}

Figure~\ref{fig:GAIF_dep_only2} presents results under a positive local dependence structure in Scenario III. The baseline methods—LOND, LORD++, and SAFFRON—are included for comparison. The proposed procedures, $\text{SF}_{\rm dep}$ and $\text{LF}_{\rm dep}$, correspond to $\text{Ada\text{-}GAIF}_{\rm dep}$ and $\rm GAIF_{\rm dep}$ under full and instant feedback, respectively. In addition, we include the existing methods \(\text{SAFFRON}_{\rm dep}\) and \(\text{LORD}_{\rm dep}\) from \cite{Zrnic2021asynchronous}.  We find that SF, LF, and SAFFRON fail to control the FDR under dependence. In contrast, the dependence-aware procedures \(\mathrm{SF}_{\mathrm{dep}}\) and \(\mathrm{LF}_{\mathrm{dep}}\) successfully achieve valid FDR control {\color{black}and attain higher power than their feedback-ignoring counterparts, \(\mathrm{SAFFRON}_{\mathrm{dep}}\) and \(\mathrm{LORD}_{\mathrm{dep}}\). See more detailed numerical values in Appendix \ref{appen:add-experiment} (Table \ref{tab:sim_results_dep}).
}

    \begin{figure}[h!]
  \centering
  \includegraphics[width=0.8\textwidth]{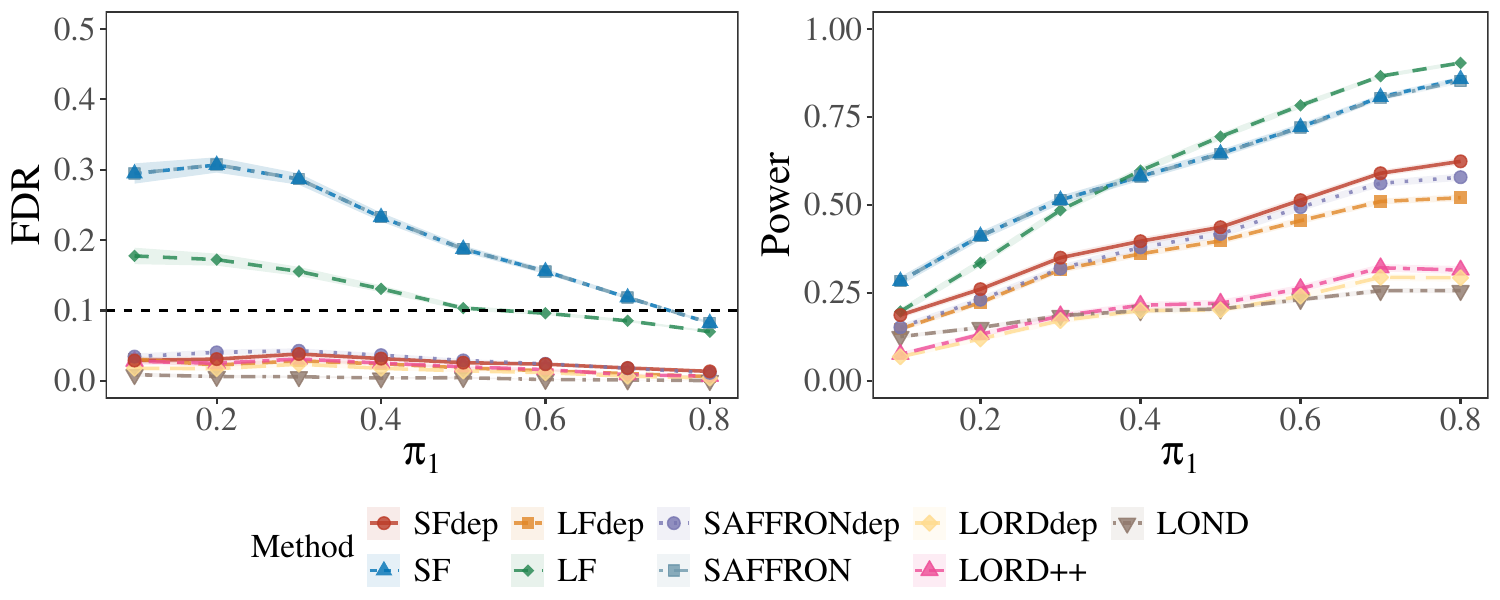}
  \caption{\small Results for Scenario III under full and instant feedback: FDR and Power at stopping time with non-null proportion $\pi_1$ ranging from 0.1 to 0.8  across 500 replications. The black dashed line indicates the target FDR level $\alpha = 0.1$. Shaded areas show $\pm 1$ standard error.}
  \label{fig:GAIF_dep_only2}
\end{figure}

Next, we evaluate the performance of our proposed methods under other feedback settings: bandit and instant, and full and delayed. Taking the classical methods LOND, LORD++, and SAFFRON as benchmarks, we implement the proposed \text{LF-BI} and \text{SF-BI} under bandit and instant feedback setting. The results in Scenario I are shown in Figure~\ref{fig:GAIF_Bandit}. All methods maintain valid empirical FDR control. {\color{black} Under bandit feedback, LF-BI remains more powerful than LORD++ by recycling
wealth from rejected non-nulls. SF-BI coincides with SAFFRON because every
rejected hypothesis satisfies \(p_t\le\alpha_t\le\lambda\), so the feedback term
\(\mathbf{1}\{p_t>\lambda\}\) is zero. 
Numerical results are reported in Table \ref{tab:sim_results_bi}.} 

 \begin{figure}[h!]
  \centering
  \includegraphics[width=0.8\textwidth]{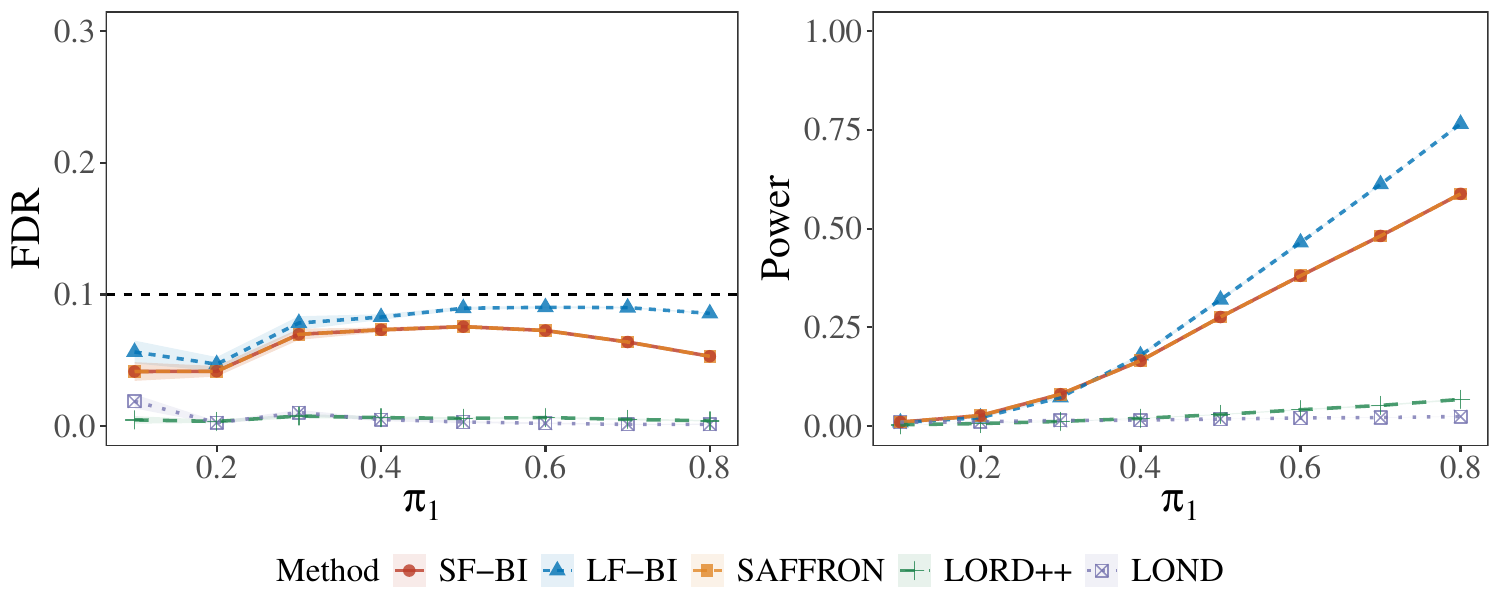}
  \caption{\small {\color{black}Results for Scenario I under bandit and instant feedback: FDR and Power at stopping time with non-null proportion $\pi_1$ ranging from 0.1 to 0.8 across 500 replications. The black dashed line indicates the target FDR level $\alpha = 0.1$. Shaded areas show $\pm 1$ standard error.}}
  \label{fig:GAIF_Bandit}
\end{figure}

For the full and delayed feedback setting, we consider delays {\color{black}$d \in \{0, 10, 100\}$} and compare the proposed SF-FD and LF-FD against feedback-free baselines. Figure~\ref{fig:GAIF_Delayed} displays the results under Scenario~I. As the delay $d$ increases, the power of our methods decreases gradually, yet they continue to outperform the baselines by a substantial margin while maintaining valid empirical FDR control. {\color{black}When $d = 0$ and the non-null proportion is large, LF achieves the highest power, as immediate feedback enables aggressive $\alpha$-wealth recycling from the abundant revealed non-nulls. Numerical results are reported in Table~\ref{tab:sim_results_fd_scenario1}.}

 \begin{figure}[h!]
  \centering
  \includegraphics[width=0.8\textwidth]{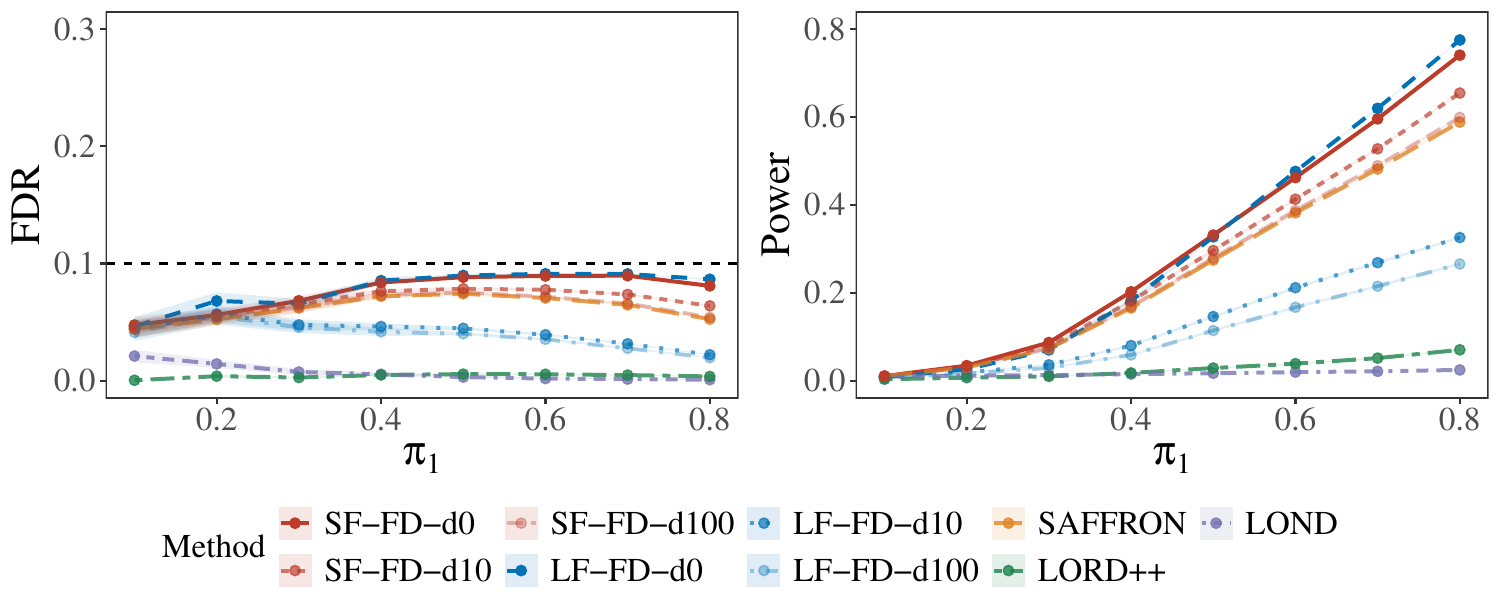}
  \caption{\small {\color{black}Results for Scenario I under full and delayed feedback: FDR and power at stopping time with non-null proportion $\pi_1$ ranging from 0.1 to 0.8 across 500 replications. The black dashed line indicates the target FDR level $\alpha = 0.1$. Shaded areas show $\pm 1$ standard error.}}
  \label{fig:GAIF_Delayed}
\end{figure}

\subsection{Results for online conformal testing with feedback}\label{subsec:simu-OCTF}

We next evaluate OCTF with a fixed score in a real-time binary classification task. 

\begin{itemize}
    \item \textbf{{Scenario IV}}: Data are from
$\mathbf{X}\mid Y=0 \sim \mathcal{N}_4\left(\boldsymbol{\mu}_1, \mathbf{I}_4\right)$, and $\mathbf{X}\mid Y=1 \sim \mathcal{N}_4\left(\boldsymbol{\mu}_2, \mathbf{I}_4\right)$, where $\boldsymbol{\mu}_1=(2,0,0,0)^{\top}, \boldsymbol{\mu}_2=(0,0,-1,-1)^{\top}$.   The target region is $\mathcal{A}=\{1\}$. 
\end{itemize}
We set the stopping time to $T = 1000$ and the target $\operatorname{FDR}$ level to $\alpha = 0.2$. The initial calibration size is $n_{\text{cal}} = 1000$.  We consider the score as $V(\mathbf{X}) = 1 - \hat{\mu}(\mathbf{X})$, where the predictive model $\hat{\mu}(\mathbf{X})$ is a support vector machine classifier trained on $n_{\text{tr}} = 1000$. Appendix~\ref{appen:add-experiment} reports additional results for varying \(n_{\mathrm{cal}}\), alternative training algorithms, and a regression setting.

We evaluate the proposed LF and SF against the benchmark procedures  LOND, LORD++, and SAFFRON, all implemented within the OCTF workflow using their own testing levels $\{\alpha_t\}$. Figure~\ref{fig:cla_stop_A_vary_prop} reports the empirical FDR and power as the non-null proportion $\pi_1$ varies over $[0.1, 0.8]$. All five procedures control the empirical FDR below the nominal level $\alpha$ across the entire range. As expected, the feedback-enhanced procedures deliver substantial power gains over their non-feedback counterparts: SF improves consistently over SAFFRON across all $\pi_1$, {\color{black}while LF trails SF for small $\pi_1$ but overtakes it once $\pi_1$ exceeds roughly $0.5$, ultimately attaining the highest power among all methods considered. We also report results for the safe variants \text{SFS} and \text{LFS} in Appendix~\ref{add:numericalvalue}; these procedures are more conservative, 
consistent with their design of prioritizing finite-sample validity via a more restrictive wealth update.
}

\begin{figure}[h!]
		\centering
		\includegraphics[width=0.8\textwidth]{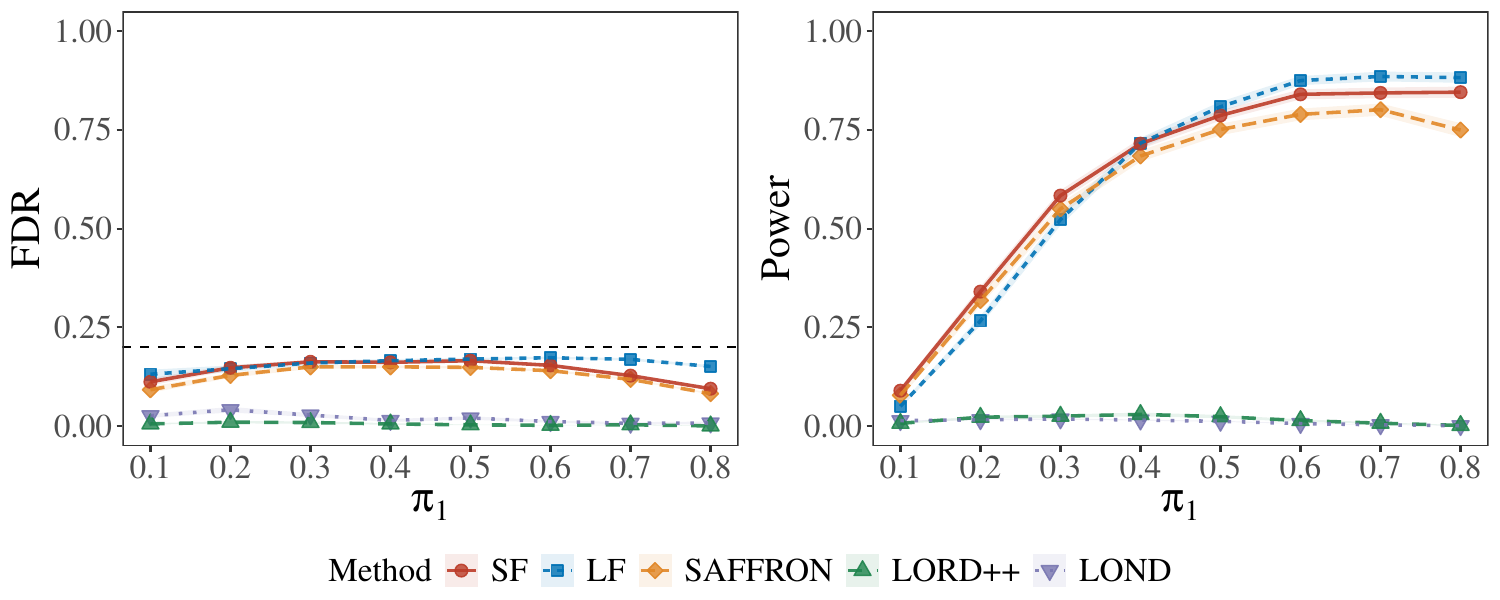}
        \caption{\small Results for Scenario {IV}: FDR and power at stopping time with non-null proportion $\pi_1$ ranging from 0.1 to 0.8 across 500 replications. The black dashed lines denote the FDR level $\alpha=0.2$. Shaded areas show $\pm 1$ standard error.}
  \label{fig:cla_stop_A_vary_prop}
	\end{figure}

\subsection{Results for optimized online conformal testing with feedback}\label{subsec:simu-opt-OCTF}

We now demonstrate the effectiveness of the proposed score selection strategy from Section \ref{subsec:Opt-OCTF}. The goal here is to examine whether feedback can be used not only for valid testing, but also for adaptive score selection when non-null distributions shift. 

{\color{black}We employ the binary classification example in Scenario IV with $\boldsymbol{\mu}_1=(3,0,0,0)^{\top}, \boldsymbol{\mu}_2=(0,0,-1,-1)^{\top}$.} For the historical labeled dataset, the non-null proportion is fixed as $0.2$. For the testing data, we consider a \text{sine pattern} shift, where the non-null proportion varies as \(\pi_t=\operatorname{Pr}(Y_t=1)= \{\sin(100\pi t/T) + 1\}/4\), oscillating between 0 and 0.5. 


{\color{black}For the implementation of the score selection}, we consider $K=3$ candidate models: random forest (RF), neural network (NN), and support vector machine (SVM) and employ the EWMA criterion $\gM^{\text{EWMA}}_t(k,\gD_t)$ with exponential weight parameter $\rho=0.95$. To reduce computational cost, the auxiliary non-null $p$-value is computed by a sliding window of fixed length $L=100$: for $t > L$, we define $\gC_{1t} = \{i : t - L \leq i < t, \theta_i = 1\}$; for $t \leq L$, we use the full set $\gC_{1t} = \{i < t : \theta_i = 1\}$. 

{\color{black}To evaluate the score-selection step, we compare the proposed optimized rule with random score selection under different testing procedures. Here, ``Opt-'' means that the score model is selected adaptively by the proposed EWMA criterion, while ``Ran-'' means that the score model is chosen uniformly at random from the candidate library. For example, Opt-LF denotes the optimized version of LF, whereas Ran-LF uses the same testing rule with random score selection.}

{\color{black}Figure~\ref{fig:cla_stop_A_vary_prop_opt_sine} shows results across time $t$ under distribution shift; results for the \text{SFS}- and \text{LFS}-based variants are deferred to Appendix~\ref{appen:add-experiment}. All methods successfully control empirical FDR at or below the nominal level $\alpha = 0.05$.} In terms of power, the optimized (\text{Opt}) variants consistently outperform their randomly selected (\text{Ran}) counterparts, confirming that optimized score selection is beneficial under distribution shift. Additional results for scenarios without distribution shifts are presented in Appendix~\ref{appen:add-experiment}, and similarly highlight the advantage of optimized score selection.

\begin{figure}[h!]
		\centering
		\includegraphics[width=1.0\textwidth]{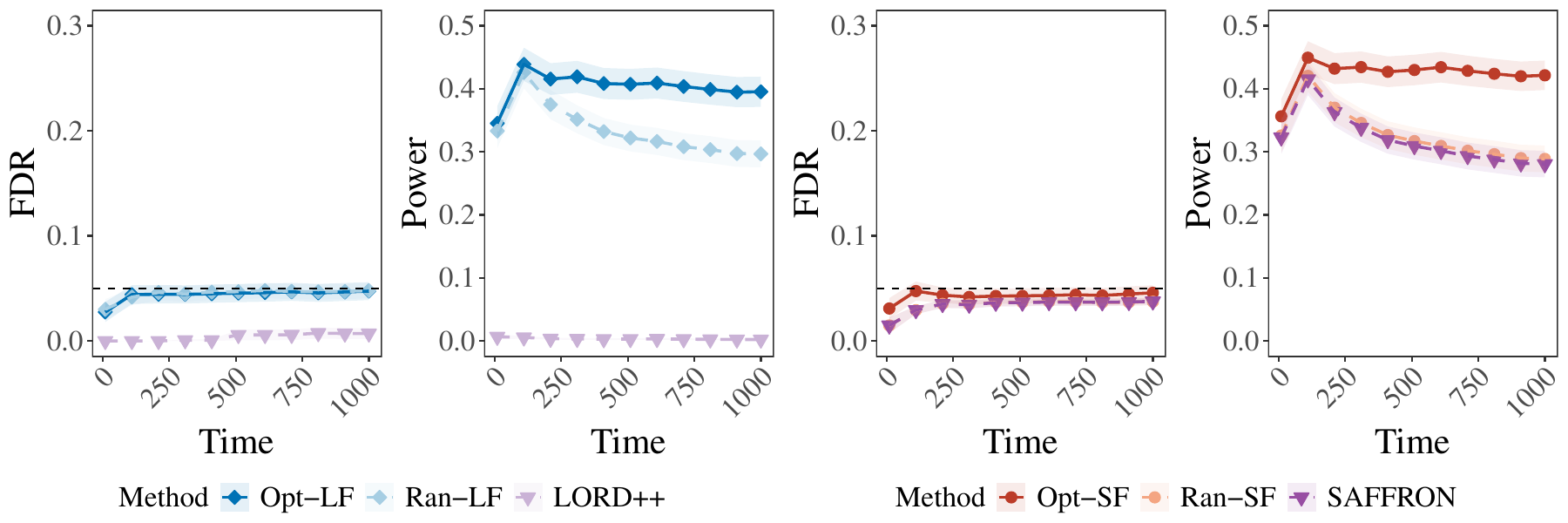}
		\caption{\small Results for Scenario {IV} (sine pattern shifts): Empirical $\FDR(t)$  and $\operatorname{Power}(t)$ across different time $t$. The black dashed lines denote the FDR level $\alpha=0.05$. Shaded areas show $\pm 1$ standard error.}
  \label{fig:cla_stop_A_vary_prop_opt_sine}
	\end{figure}

\section{Real Data Applications}\label{sec:real-data}

In this section, we evaluate our proposed methods on four real-world datasets to illustrate their practical benefits in diverse online decision-making tasks. 


\begin{itemize}
    \item  \textbf{Task 1: Online Candidate Screening.} The first task is real-time candidate screening for selecting the candidates who can get into the interview process. 
    We use the recruitment dataset \citep{Candidate}, which contains $N=45,372$ candidates with 11 attributes, including education status, handicapped or not, and gender. The target binary variable is whether the candidate passes the job interview. \textcolor{black}{In this task, the non-null proportion is  $\pi_1=0.69$.}
    \item \textbf{Task 2: High-Risk Diabetes Identification.} The second task focuses on health screening, where selecting individuals at high risk of diabetes is critical for early intervention. We use the diabetes health indicators dataset \citep{Healthcare_Diabetes_Data}, which contains $N=70,692$ samples and 22 covariates, including demographic attributes (e.g., sex, age, BMI), lifestyle-related features, and several binary health indicators. The target binary variable is whether an individual has diabetes. \textcolor{black}{The non-null proportion is $\pi_1=0.50$ in the dataset.}
    \item \textbf{Task 3: High-Income Individual Selection.}  The third task involves using the 1994 Census Bureau dataset \citep{misc_adult_2} to identify a subset of individuals with high incomes (i.e., income $>50K$) for precision marketing. This dataset contains $N=32,561$ records with 14 attributes, including gender, race, marital status, education level, and more.  \textcolor{black}{The non-null proportion is $\pi_1=0.25$.}
    \item \textbf{Task 4: Airfoil Noise Detection.}  The final task involves using the airfoil dataset \citep{misc_airfoil_self-noise_291} from the UCI Machine Learning Repository to identify samples with high sound pressure. This dataset contains $N = 1,503$ observations with five physical covariates (log frequency, angle of attack, chord length, free-stream velocity, and suction side log displacement thickness). The response variable $Y$ represents the scaled sound pressure, and we aim to test $\mathbb{H}_{0t}: Y_t \in [-\infty,c)$, where $c$ is the $(1 - \pi_1)$-quantile of $Y$, with \textcolor{black}{$\pi_1=0.60$}.
\end{itemize}

For Tasks 1-3, we randomly sample three subsets from the full dataset: $n_{\text{tr}} = 1{,}000$ training, $n_{\text{cal}} = 1{,}000$ calibration, and $n_{\text{te}} = 1{,}000$ test samples. For Task 4, we set $n_{\text{tr}} = 300$, $n_{\text{cal}} = 300$, and $n_{\text{te}} = 1{,}000$. {\color{black}We compare Opt-SF and Opt-LF with three feedback-free benchmarks, SAFFRON, LORD++, and LOND.}
At each time step, $K=3$ candidate training algorithms (RF, NN, and SVM) are employed, and the score selection criterion is $\gM^{\text{EWMA}}_t(k,\gD_t)$ with $\rho=0.9$. Our methods use the optimized conformity score, while the LORD++ and SAFFRON use a randomly selected score at each time $t$. 
All results are averaged over 500 replications. 



Figure~\ref{fig:FDR_power_real-data} summarizes the empirical FDR and power over time $t$, and Table~\ref{tab:all_fdr_power} reports the corresponding values at the stopping time $T = 1{,}000$ at the target FDR level $\alpha=0.3$. {\color{black}Across all four tasks, all methods maintain empirical FDR below the nominal level. 

The relative performance of \text{Opt-LF} and \text{Opt-SF} is consistent with our earlier simulation findings. For Tasks 1, 2, and 4, where the non-null proportion is relatively large ($\pi_1 = 0.69$, $0.50$, and $0.60$ respectively), \text{Opt-LF} achieves the highest power ($0.308$, $0.163$, and $0.684$), as abundant non-null feedback enables more aggressive $\alpha$-wealth recycling. For Task 3 ($\pi_1 = 0.25$), \text{Opt-SF} instead achieves the best power ($0.151$), consistent with the observation that adaptive $\alpha$-wealth allocation becomes more consequential when non-nulls are sparse and feedback-recovered wealth is limited.} 
\text{SAFFRON} and \text{LORD++} exhibit substantially lower power across all tasks, with \text{LORD++} and \text{LOND} achieving near-zero power in most settings. {\color{black}Results for the safe variants \text{Opt-SFS} and \text{Opt-LFS} are deferred to Appendix~\ref{appen:add-experiment} (Table~\ref{tab:appendix_safe_fdr_power}); both achieve empirical FDR control and have conservative power across all tasks.}


\begin{figure}[htbp!]
		\centering
		\includegraphics[width=1.0\textwidth]{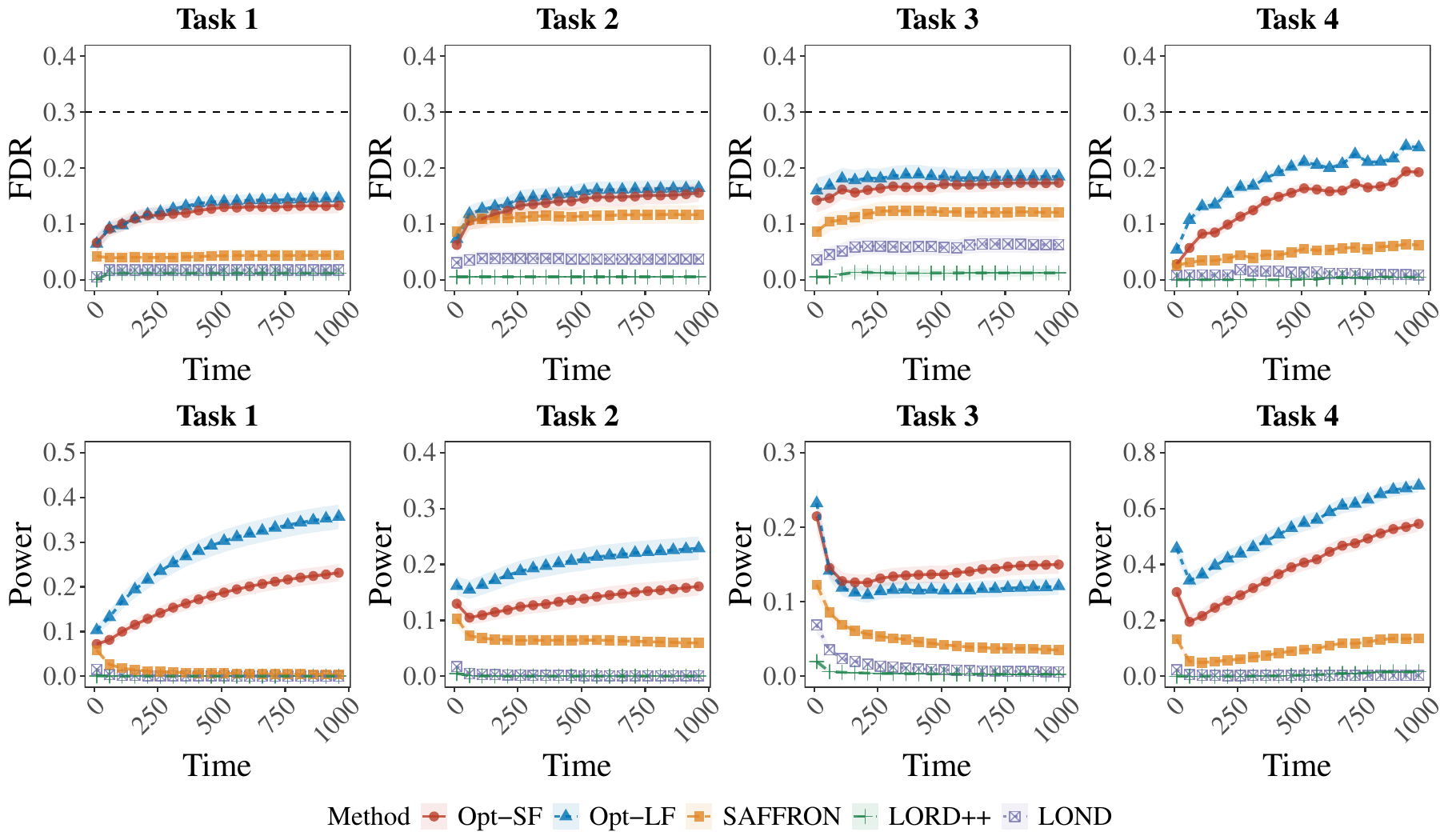}
		\caption{\small  Results for real-data applications: the values of $\operatorname{FDR}(\boldsymbol{\delta}^t)$ and $\operatorname{Power}(\boldsymbol{\delta}^t)$ over time $t$ for five methods across four tasks (Task~1: candidate screening, $\pi_1=0.69$; Task~2: diabetes identification, $\pi_1=0.50$; Task~3: income selection, $\pi_1=0.25$; Task~4: airfoil noise detection, $\pi_1=0.60$). The black dashed lines indicate the FDR level $\alpha = 0.3$. Shaded areas show $\pm 1$ standard error.}
  \label{fig:FDR_power_real-data}
	\end{figure}

\begin{table}[!ht]
\centering
\setlength{\heavyrulewidth}{0.5pt}
\setlength{\lightrulewidth}{0.3pt}
\caption{\small $\operatorname{FDR}(T)$ and $\operatorname{Power}(T)$ for different tasks across four datasets (Candidate, Diabetes, Income, Airfoil). The target FDR level is $\alpha = 0.3$. Bold numbers represent the best results.}
\label{tab:all_fdr_power}
\resizebox{\textwidth}{!}{
{\tiny
\begin{tabular}{lcccccccc}
\toprule
\textbf{Method} 
& \multicolumn{2}{c}{\textbf{Task 1}} 
& \multicolumn{2}{c}{\textbf{Task 2}} 
& \multicolumn{2}{c}{\textbf{Task 3}} 
& \multicolumn{2}{c}{\textbf{Task 4}} \\
\cmidrule(lr){2-3} \cmidrule(lr){4-5} \cmidrule(lr){6-7} \cmidrule(lr){8-9}
& FDR & Power & FDR & Power & FDR & Power & FDR & Power \\
\midrule
Opt-SF  & .118\tiny{(.014)} & .181\tiny{(.018)} & .135\tiny{(.016)} & .121\tiny{(.015)} & .166\tiny{(.014)} & \textbf{.151}\tiny{(.012)} & .166\tiny{(.008)} & {.499}\tiny{(.028)} \\
Opt-LF  & .152\tiny{(.014)} & \textbf{.308}\tiny{(.026)} & .138\tiny{(.016)} & \textbf{.163}\tiny{(.019)} & .181\tiny{(.017)} & {.126}\tiny{(.012)} & .228\tiny{(.007)} & \textbf{.684}\tiny{(.024)} \\
SAFFRON & .034\tiny{(.010)} & .001\tiny{(.000)} & .102\tiny{(.016)} & .022\tiny{(.005)} & .095\tiny{(.013)} & .031\tiny{(.004)} & .050\tiny{(.007)} & .136\tiny{(.015)} \\
LORD++  & .000\tiny{(.000)} & .000\tiny{(.000)} & .010\tiny{(.007)} & .000\tiny{(.000)} & .011\tiny{(.007)} & .001\tiny{(.000)} & .002\tiny{(.001)} & .012\tiny{(.004)} \\
LOND    & .015\tiny{(.009)} & .000\tiny{(.000)} & .045\tiny{(.014)} & .000\tiny{(.000)} & .035\tiny{(.011)} & .004\tiny{(.000)} & .007\tiny{(.004)} & .003\tiny{(.000)} \\
\bottomrule
\end{tabular}
}
}
\end{table}

\section{Concluding Remarks}\label{sec:conclu}

We study online multiple testing with feedback, aiming to develop reliable machine learning methods for real-time decision-making with rigorous FDR and mFDR control. Our key contribution is \text{GAIF} \textcolor{black}{and its adaptive variants}, a {feedback-enhanced Generalized Alpha-Investing framework} that guarantees online FDR control under the standard assumption of independence between null and all other \(p\)-values. Building on GAIF, we further develop \text{OCTF} for online conformal testing by constructing independent online conformal \(p\)-values for null hypotheses, which achieves finite-sample mFDR control. To further enhance performance and address distribution shifts among non-nulls, we propose a feedback-driven score selection strategy  \textcolor{black}{and provide theoretical analysis characterizing its optimality properties}. 

We highlight several potential directions for future work. 
First, the current framework mainly focuses on distribution shifts in the non-null data. To address alpha-death and piggybacking under such shifts, one can consider controlling weighted FDR using forgetting factors or decaying memory mechanism \citep{ramdas2017online}. Second, extending our framework to accommodate more general forms of distribution shift is another promising direction. One potential approach is to incorporate adaptive strategies that exploit local information, in the spirit of the adaptive conformal inference \citep{gibbs2021adaptive}. Finally, to establish guarantees under weaker assumptions, one could relax the definition of valid online FDR control and develop new error-rate notions, analogous to the average-coverage criterion studied in online conformal inference \citep{gibbs2021adaptive,humbert2025onlineselectiveconformalinference}.

\section*{Acknowledgments}

We would like to acknowledge the action editor and three anonymous referees for their valuable comments and suggestions, which have improved the manuscript greatly.





\newpage

\appendix
\newpage
\begin{center}
    \Large{\bf Supplementary Material for ``Feedback-Enhanced Online Multiple Testing with Applications to Conformal Selection''}
\end{center}

{\color{black}This supplementary material includes:
\begin{itemize}
    \item Notation and   Algorithmic Details (Appendix \ref{appen:notation})
    \item Additional Discussions on GAIF and Adaptive GAIF (Appendix \ref{appen:add-details})
    \item Additional Discussions on OCTF (Appendix \ref{appen:OCTF})
     \item Applications on Real-time LLM Alignment (Appendix \ref{appen:LLM})
       \item Additional Experimental Results (Appendix \ref{appen:add-experiment})
      \item  Proofs of all the theoretical results (Appendix \ref{appen:proofs})
\end{itemize}
}

\section{Notation and   Algorithmic Details}\label{appen:notation}
This appendix gathers the notation and additional algorithmic details used throughout the paper. 
For completeness, we first review several preliminary terms and benchmark procedures, and then 
provide the concrete forms of the proposed testing levels and related implementation details.

\subsection{Preliminary terms for self-containment}
\label{appen:terms}
{
Here, we list the preliminary terms we use in the paper for the sake of clarity and self-containment.
\begin{itemize}
    \item \textbf{$\operatorname{FDR}$} \citep{Benjamini1995ControllingTF}, false discovery rate, a widely-adopted error rate notion in the field of multiple testing, is defined as the expected proportion of incorrectly rejected null hypotheses as follows:
   \begin{align*}
    \FDP(t) = \frac{V_t}{1\vee R(t)}:=\frac{\sum_{j=1}^t \delta_j(1-\theta_j)}{1 \vee \sum_{j=1}^t \delta_j}\quad\text{and}\quad \FDR(t):=\E[\FDP(t)],
\end{align*}
    where $R(t)$ represents the number of rejections until time $t$ and $V(t)$ is the number of false discoveries.
    \item \textbf{$\mFDR$}, modified false discovery rate, or marginal false discovery rate, is defined as:
    \[\operatorname{mFDR}(t):=\frac{\E[V(t)]}{\E[1\vee R(t)]}=\frac{\E[\sum_{j=1}^t \delta_j(1-\theta_j)]}{\E[1\vee \sum_{j=1}^t \delta_j]}.\]
   \item \textbf{Testing levels of existing online multiple testing methods.} For the online methods, denote the decision rule as $\delta_t=\mathbb\{p_t\leq \alpha_t\}$, where $p_t$ is the corresponding conformal $p$-value at time $t$ for our problem. The test levels $\{\alpha_t\}$ for LOND \citep{javanmard2015online}, LORD++ \citep{ramdas2017online}, SAFFRON \citep{ramdas2018saffron} and ADDIS \citep{tian2019addis}, $\text{LORD}_{\text{dep}}$ \citep{Zrnic2021asynchronous}, $\text{SAFFRON}_{\text{dep}}$ \citep{Zrnic2021asynchronous} are listed as follows:
   \begin{itemize}
    \item[1.] LOND: $\alpha_t=\gamma_t(R(t-1)+1)$, where $\{\gamma_t\}_{t=1}^{\infty}$ is a given infinite non-increasing sequence of positive constants  that sums to $\alpha$ and $R(n)=\sum_{t=1}^{n}\delta_t$ denotes the number of discoveries in the first $n$ hypotheses tested.
    \item[2.] LORD++:   $$\alpha_t^{\text{LORD++}}=\gamma_t s_0+(\alpha-s_0)\gamma_{t-\tau_1}\mathbb{I}\{\tau_1<t\}+\alpha\sum_{j:\tau_j<t,\tau_j\neq \tau_1}\gamma_{t-\tau_j},$$
    where $\{\gamma_t\}_{t=1}^{\infty}$ is a given infinite non-increasing sequence of positive constants that sums to one; $\tau_j$ is the time of the $j$-th rejection. 
    \item[3.] SAFFRON: 	At each time $t$, define $C_{j+}=C_{j+}(t)=\sum_{i=\tau_j+1}^{t-1}C_i$, where $C_t=\mathbb{I}\{p_t\leq \lambda\}$. For $t=1$, $\alpha_1=\min\{\gamma_1(1-\lambda)s_0,\lambda\}$; For $t=2,3,\dots$, $\alpha_t:=\min\{\lambda,\tilde{\alpha}_t\}$, where 
		\[\tilde{\alpha}_t=(1-\lambda)s_0\gamma_{t-C_{0+}}+((1-\lambda)(\alpha-s_0))\gamma_{t-\tau_1-C_{1+}}+(1-\lambda)\alpha\sum_{j\geq2}\gamma_{t-\tau_j-C_{j+}}.\]
    \item[4.] ADDIS: The testing levels for ADDIS are given by $\alpha_t=\min\{\lambda,\hat{\alpha}_t\}$, where 
    \[\hat{\alpha}_t=(\eta-\lambda)[\omega_0\gamma_{S^t-C_{0+}}+(\alpha-\omega_0)\gamma_{S^t-\tau_1^*-C_{1+}}+\alpha\sum_{j\geq 2}\gamma_{S^t-\tau_j^*-C_{j+}}]\]
    and $S^t=\sum_{i<t}\mathbb{I}\{p_i\leq \eta\}$, $\tau_j^*=\sum_{i\leq \tau_j}\mathbb{I}\{p_i\leq \eta\}$.
    \item[5.] $\text{LORD}_{\text{dep}}$ and $\text{SAFFRON}_{\text{dep}}$: Define $r_k$ under local dependence as:
\[r_k=\min\{i\in[t]: \sum_{j=1}^{i-L_{i+1}}\delta_j\geq k\}.\] The corresponding test levels for $\text{LORD}_{\text{dep}}$ and $\text{SAFFRON}_{\text{dep}}$ are as follows:

    \begin{equation*}\label{eq:alpha-LORD-dep}
    \alpha_t^{\text{LORD}_{\text{dep}}}=\gamma_t s_0+(\alpha-s_0)\gamma_{t-r_1}\mathbb{I}\{r_1<t\}+\alpha\sum_{j=2}^\infty\gamma_{t-r_j}.
\end{equation*}
    \begin{equation*}\label{eq:alpha-SAFFRON-dep}
    \alpha_t^{\text{SAFFRON}_{\text{dep}}}:=\min\left\{\lambda, (1-\lambda)\left(s_0\gamma_{t-C_{0+}}+(\alpha-s_0)\gamma_{t-r_1-C_{1+}}+\alpha(\sum_{j\geq 2}\gamma_{t-r_j-C_{j+}})\right)\right\},
\end{equation*}
where $C_{j+}=\sum_{i=r_j+1}^{t-L_t+1}C_i$.
\end{itemize}

\item \textbf{Conformal $p$-values.} The notion of conformal \( p \)-values was originally introduced by \citet{vovk2005algorithmic} for constructing prediction intervals. A conformal \( p \)-value quantifies how well a new observation conforms to a reference set, based on a chosen nonconformity score function. More recently, several works have applied conformal \( p \)-values to sample selection from a multiple testing perspective~\citep{jin2025model,bates2021testing,wang2024conformalized}. The conformal $p$-values are defined as
\begin{equation}
    p_t = \frac{ \sum_{i \in \gC_{0}} \I\{V_{i} < V_{t}\}+\xi_t[1+\sum_{i\in\gC_0}\I\{V_i=V_t\}]}{1 + |\gC_{0}|},
\end{equation}  
where $\gC_0$ is a hold-out calibration dataset and $V(\cdot)$ is a nonconformity score function. However, in these approaches, the conformal \( p \)-values are constructed using a fixed offline calibration dataset, which limits their flexibility in online or adaptive settings.

\end{itemize}}

\subsection{Testing levels of the proposed approaches}\label{appendix:testing-levels}
\begin{itemize}
    \item[1.] Testing levels for $\text{GAIF}_{\text{dep}}$ and $\text{Ada-GAIF}_{\text{dep}}$:
    
    Define $r_k$ under local dependence as:
\[r_k=\min\{i\in[t]: \sum_{j=1}^{i-L_{i+1}}\delta_j\geq k\}.\] The corresponding test levels for $\text{LF}_{\text{dep}}$ and $\text{SF}_{\text{dep}}$ are constructed as follows:
\begin{equation*}\label{eq:alpha-LORDF-dep}
    \alpha_t^{\text{GAIF}_{\text{dep}}}=\gamma_t s_0+(\alpha-s_0)\gamma_{t-r_1}\mathbb{I}\{r_1<t\}+\alpha\sum_{j=2}^\infty\gamma_{t-r_j}+\textcolor{black}{\sum_{j:j\in\mathcal{I}_t, j<t-L_t}\gamma_{t-j}\alpha_{j}\theta_{j}},
\end{equation*}
\begin{equation*}\label{eq:alpha-SAFFRONF-dep}
    \begin{aligned}
        \alpha_t^{\text{Ada-GAIF}_{\text{dep}}} := \min\Biggl\{ & \lambda, (1-\lambda)\biggl( s_0\gamma_{t-C_{0+}} + (\alpha-s_0)\gamma_{t-r_1-C_{1+}} + \alpha\sum_{j\geq 2}\gamma_{t-r_j-C_{j+}} \biggr) \\
        & + \textcolor{black}{\sum_{j \,:\, j \in \mathcal{I}_t,j<t-L_t} \gamma_{t - j-C_{j+}^{\#}}\, \alpha_j\, \theta_j\mathbb{I}\{p_j>\lambda\}} \Biggr\},
    \end{aligned}
\end{equation*}
where $C_{j+}=\sum_{i=r_j+1}^{t-L_t}C_i$, {\color{black}$C_{j+}^{\#}=\sum_{i=j+1}^{t-L_t}C_i$.}
\item[2.] Testing levels for $\text{LFS}$ and $\text{SFS}$:
Recall that the proposed test levels of GAIF under full and instant feedback, i.e., $\alpha_t^{\rm LF}$, are given by:
\begin{equation}
\alpha_t^{\text{LF}}=\gamma_t s_0+(\alpha-s_0)\gamma_{t-\tau_1}\mathbb{I}\{\tau_1<t\}+\alpha\sum_{j:\tau_j<t,\tau_j\neq \tau_1}\gamma_{t-\tau_j}+\sum_{j:j<t}\gamma_{t-j}\alpha_{j}\theta_j,
\end{equation}
where \( \tau_j \) denotes the time of the \( j \)-th rejection. We revise this rule to:
\begin{equation}\label{eq:alpha-LORDF-conservative}
{\alpha}_t^{\text{LFS}}=\gamma_t s_0+(\alpha-s_0)\gamma_{t-\tilde{\tau}_1}\mathbb{I}\{{\tilde\tau}_1<t\}+\alpha\sum_{j:\tilde{\tau}_j<t,\tilde{\tau}_j\neq \tilde{\tau}_1}\gamma_{t-\tilde{\tau}_j}+\sum_{j:j<t}\gamma_{t-j}{\alpha}_{j}\theta_j,
\end{equation}
where \( \tilde{\tau}_j \) denotes the time of the \( j \)-th rejection under the null, defined as
\[
\tilde{\tau}_j = \inf\left\{t\in\mathbb{N}:\sum_{i\leq t}\delta_i(1-\theta_i)\geq j\right\}.
\]
And similarly, the test levels of SFS are:
\begin{equation}\label{eq:alpha-SAFFRONF-conservative}
    \alpha_t^{\text{SFS}}:=\min\{\lambda,\tilde{\alpha}_t^{\text{SFS}}\},
\end{equation}
where $$
       \textcolor{black}{\tilde{\alpha}_t^{\text{SFS}}=(1-\lambda)\Bigl[
            s_0\, \gamma_{t - C_{0+}}
            + (\alpha - s_0)\, \gamma_{t - \tilde\tau_1 - C_{1+}}
            + \alpha \sum_{j \geq 2} \gamma_{t - \tilde\tau_j - C_{j+}}
        \Bigr]+ \sum_{j \,:\, j <t} \gamma_{t - j-C_{j+}^{*}}\, \alpha_j\, \theta_j\mathbb{I}\{p_j>\lambda\}. }
   $$
\end{itemize}

\subsection{Choice of non-conformity score function $V$}

In terms of the non-conformity score function, denote $W_t=\widehat{\mu}_t(\X_t)$, in classification settings, we set $V(W_t)=1-W_t$.  In regression settings, if $\mathcal{A}=[b,+\infty)$, we can use $V(W_j)=b-W_j$. If $\mathcal{A}=(-\infty,a]\cup[b,+\infty)$, then we can choose $V(W_t)=\max\{W_t-a,b-W_t\}$.

\section{Additional Discussions on GAIF and Adaptive GAIF} \label{appen:add-details}

In this section, we provide additional details of our GAIF and Adaptive GAIF  algorithms in Sections \ref{sec:GAIF} and \ref{sec:exten-GAIF}.


\subsection{Explanation about utilizing feedback in SAFFRON}\label{subsec:explain-saffronf}

From another perspective, one may propose using feedback to estimate the null proportion and operate in a manner similar to SAFFRON, but in a form different from ours in Section \ref{subsec:Adaptive-GAIF}, where the construction of $\alpha_t$ is
\begin{equation}
       {\alpha}^{\rm SF-variant}_t=s_0\gamma_{t-D_{0+}}+(\alpha-s_0)\gamma_{t-\tau_1-D_{1+}}+\alpha\sum_{j\geq2}\gamma_{t-\tau_j-D_{j+}}, 
   \end{equation}
here $D_{j+}=\sum_{i=\tau_j+1}^{t-1}\I\{\theta_i=1\}$. This approach constructs $\alpha_t$ using $\{\theta_i\}_{i=1}^{t-1}$ only, and can provide the same finite-sample online FDR guarantee.

Here, we explain why this form is not adopted in the proposed GAIF. In practice, we typically have $D_j^+<C_j^+$, since the set filtered by $p$-values is larger than the true non-null set. Consequently, due to the exponential decay design of $\gamma_t$, we have $$\gamma_{t-\tau_j-D_{j+}}\ll\gamma_{t-\tau_j-C_{j+}}.$$ 
Although the original SAFFRON includes an adjustment factor $1-\lambda$ which reduces the $\alpha_t$, this term becomes negligible compared to the exponential decay of $\gamma_{t-\tau_1-C_{1+}}$. As a result, the original SAFFRON achieves higher power than this naive feedback-based variant, especially when $\lambda$ is large. 

Our empirical results in Figure~\ref{fig:SAFFRON_lambda_vary} further confirm this phenomenon. In the Gaussian setting described in Section~\ref{subsec:simu-GAIF} with $\mu=2$ and $\pi_1=0.3$, we observe that when $\lambda>0.3$, the original SAFFRON outperforms the variant in terms of power. Moreover, its highest power occurs at $\lambda=0.8$, which is relatively large. This highlights the trade-off between $\gamma_{t-\tau_j-C_{j+}}$ and $1-\lambda$, where the rapidly decaying $\gamma_{t-\tau_j-C_{j+}}$ dominates, rendering $1-\lambda$ less influential.

In conclusion, directly using feedback to estimate the null proportion is not advisable, as it leads to a poor construction of $\alpha_t$. By contrast, our SF strategy leverages the $p$-values to achieve adaptive $\alpha$-wealth allocation and demonstrates superior performance.

\begin{figure}[h!]
		\centering
\includegraphics[width=0.7\textwidth]{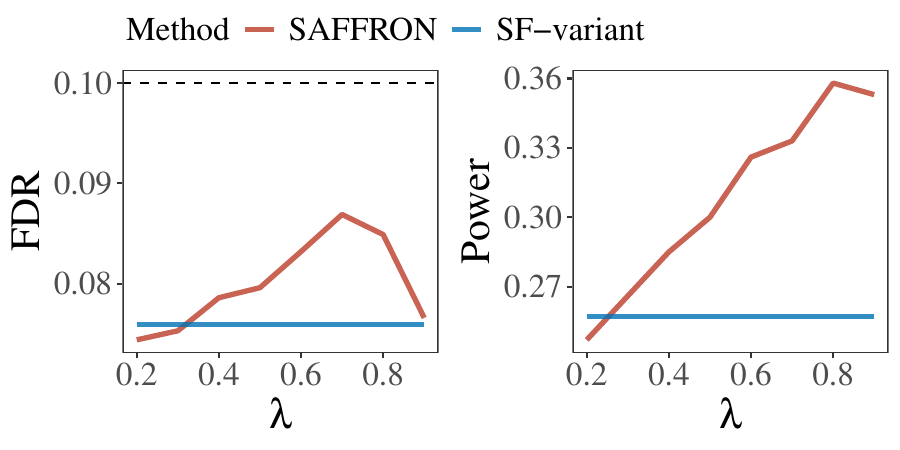}
		\caption{\small The FDR and Power for SAFFRON at stopping time $600$ under different $\lambda$ value for target FDR level $\alpha=0.1$. }
  \label{fig:SAFFRON_lambda_vary}
	\end{figure}

\subsection{Extensions of GAIF based on $e$-values}\label{app:exten-GAIF}

{\color{black}Although feedback cannot be directly used to improve $e$-LOND \citep{xu2024online}, it can enhance $e$-LORD and $e$-SAFFRON \citep{zhang2025egai} through a feedback-driven approach analogous to the extension from GAI to GAIF. Denote $\delta_j=\mathbb{I}\{e_j\geq 1/\alpha_j\}$. Denote $$\operatorname{FDP}^*(t)=\sum_{j\leq t,j\in \mathcal{H}_0}\frac{\alpha_j}{1+R({j-1})}\leq \sum_{j\leq t}\frac{\alpha_j}{1+R({j-1})}=\widehat{\operatorname{FDP}}_{\text{e-LORD}} .$$ 

Similar to the GAIF framework, we propose the estimates $\widehat{\FDP}_{\text{e-GAIF}}$ and $\widehat{\FDP}_{\text{e-Ada-GAIF}}$ as follows:
$$\operatorname{FDP}^*(t)=\frac{\sum_{j\leq t,j\in \mathcal{H}_0}\alpha_j}{1\vee \sum_{j\leq t}\delta_j}\leq \widehat{\operatorname{FDP}}_{\text{e-GAIF}} =\sum_{j\in \mathcal{I}_t}\frac{\alpha_j(1-\theta_j)}{1+R({j-1})}+\sum_{j\in\bar{\mathcal{I}}_t}\frac{\alpha_j}{1+R({j-1})}.$$

$$\widehat{\FDP}(t)_{ \text{e-Ada-GAIF}}=\sum_{j\in \mathcal{I}_t}\frac{\alpha_j(1-\theta_j)}{1+R({j-1})}\frac{\mathbb{I}\{e_j\leq 1/\lambda\}}{1-\lambda}+\sum_{j\in\bar{\mathcal{I}}_t}\frac{\alpha_j}{1+R({j-1})}\frac{\mathbb{I}\{e_j\leq 1/\lambda\}}{1-\lambda}.$$

Then we require $\widehat{\FDP}_{\text{e-GAIF}}\leq \alpha$ or $\widehat{\FDP}_{\text{e-Ada-GAIF}}\leq \alpha$ when constructing the testing levels. The corresponding testing levels are as follows:
\begin{equation}\label{eq:alpha-e-LF}
    \alpha_t^{\text{e-GAIF}}=\omega_t\left(\alpha-\sum_{j\in\mathcal{I}_t}\frac{\alpha_j(1-\theta_j)}{1+R({j-1})}-\sum_{j\in\bar{\mathcal{I}}_t}\frac{\alpha_j}{1+R({j-1})}\right)(R({t-1})+1),
\end{equation}
\begin{equation}\label{eq:alpha-e-SF}
    \alpha_t^{\text{e-Ada-GAIF}}=\omega_t\left(\alpha(1-\lambda)-\sum_{j\in\mathcal{I}_t}\frac{\alpha_j(1-\theta_j)\mathbb{I}\{e_j<1/\lambda\}}{1+R({j-1})}-\sum_{j\in\bar{\mathcal{I}}_t}\frac{\alpha_j\mathbb{I}\{e_j<1/\lambda\}}{1+R({j-1})}\right)(R({t-1})+1),
\end{equation}
where $\omega_t\in(0,1)$ is updated by
\begin{equation}\label{eq:omega_update}
    \omega_{t+1}=\omega_t+\omega_1\varphi^{t-R(t)}(1-\delta_t)-\omega_1\psi^{R(t)}\delta_t
\end{equation}
with a user-defined initial allocation coefficient $\omega_1\in(0,1)$, and user-defined parameters $\varphi>0,\psi>0$.

The Generalized Alpha-Investing with Feedback procedure based on $e$-values is summarized in Algorithm \ref{alg:e-GAIF}. The $e$-GAIF and $e$-Ada-GAIF can control online FDR validly  if the null $e$-values are conditional valid.

\begin{algorithm}[h!]
    \small
    \captionsetup{font=small}
    \caption{$e$-GAIF ($e$-GAIF and $e$-Ada-GAIF)}
    \label{alg:e-GAIF}
\begin{algorithmic}[1]
    \REQUIRE{Target $\operatorname{FDR}$ level $\alpha$, parameters $\lambda,\varphi$ and $\psi\in(0,1)$, initial allocation coefficient $\omega_1\in(0,1)$.} \\[0.5ex]
    \FOR{$t=1,2,\dots$}  
\STATE  Observe $e$-value $e_t$\;
\STATE  Update $\alpha_t=\alpha_t^{\text{e-GAIF}}$ in Eq.(\ref{eq:alpha-e-LF}) (or $\alpha_t=\alpha_t^{\text{e-Ada-GAIF}}$ in Eq.(\ref{eq:alpha-e-SF})) \; 
\STATE \textbf{if} $e_t\geq 1/\alpha_t$ \textbf{then} $\delta_t=1$, \textbf{else} $\delta_t=0$ 
\STATE Update $R(t)=R(t-1)+\delta_t$ and $\omega_{t+1}$ by (\ref{eq:omega_update})
 \STATE Obtain the revealed feedback $\theta_t$
    \ENDFOR 
   
    \ENSURE{Rejection set $\{t:\delta_t=1\}$.}
\end{algorithmic}
\end{algorithm}
    
\begin{theorem}[Online FDR control for e-GAIF and e-Ada-GAIF]\label{the:FDR_e-GAIF}
\begin{itemize}
    \item[(a.)] If the null $e$-values are conditional valid, i.e.,
    \begin{equation}
        \E\left[e_t\mid \mathcal{F}_{t-1}\right]\leq 1\; \text{for all}\; t\in\mathcal{H}_0,
    \end{equation}
     where $\mathcal{F}_{t-1}=\sigma(\delta_{1:t-1};\{\theta_j\}_{j\in\mathcal{I}_{t}})$ is the sigma field generated from past rejections and feedback, then if the parameters $\{\alpha_t\}_{t\in \mathbb{N}}$ are selected such that $\widehat{\operatorname{FDP}}_{\operatorname{e-GAIF}}(t)\leq \alpha$, then we have
    \[\operatorname{FDR}(t)\leq \alpha\quad \text{for all} \quad t\in\mathbb{N}.\]
    \item[(b.)] Denote $C_t=\mathbb{I}\{e_t<1/\lambda\}.$ If the null $e$-values are conditional valid, i.e.,
    \begin{equation}
        \E\left[e_t\mid \mathcal{J}_{t-1}\right]\leq 1\; \text{for all}\; t\in\mathcal{H}_0,
    \end{equation}
     where $\mathcal{J}_{t-1}=\sigma(\delta_{1:t-1};\{\theta_j\}_{j\in\mathcal{I}_{t}};C_{1:t-1})$ is the enlarged filtration, then if the parameters $\{\alpha_t\}_{t\in \mathbb{N}}$ are selected such that $\widehat{\operatorname{FDP}}_{\operatorname{e-Ada-GAIF}}(t)\leq \alpha$ , then we have
    \[\operatorname{FDR}(t)\leq \alpha\quad \text{for all} \quad t\in\mathbb{N}.\]
\end{itemize}
\end{theorem}
}

\subsection{Discussion about improving LOND/e-LOND with feedback}
\label{appen:LOND_disscussion}
In this subsection, we discuss the feasibility of improving LOND or $e$-LOND through our feedback-enhanced framework.

Recall the setup of the LOND or $e$-LOND algorithm. Given a non-negative sequence $\{\gamma_j\}_{j=1}^\infty$ such that $\sum_{j=1}^\infty \gamma_j=1$, the test levels are set as
\[\alpha_t^{\text{LOND}}=\alpha\gamma_t\left(\sum_{j=1}^{t-1}\delta_j\vee 1\right)=\alpha\gamma_t\left({R}(t-1)\vee 1\right).\]
 \cite{Zrnic2021asynchronous} proved that LOND controls the FDR under PRDS. 

\paragraph{Improve LOND or $e$-LOND with feedback? (negative result)} 

Here we prove that the feedback cannot be used to improve LOND and $e$-LOND.

 
        Denote \( a_j = \sum_{i=1}^{j-1} \delta_i \theta_i \) and \( b_j = \sum_{i=1}^{j-1} \delta_i (1 - \theta_i) \), so that \({R}({j-1}) = a_j + b_j \).  
For LOND or e-LOND, the significance level is given by  
\[
\alpha_j = \alpha \gamma_j (a_j + b_j + 1).
\]  
If \( \theta_1, \dots, \theta_{j-1} \) are known, we consider the modified significance level  
\[
\tilde{\alpha}_j = \alpha \gamma_j (2w_j a_j + 2(1 - w_j) b_j + 1),
\]  
where \( w_j \in [0,1] \) is a weight that differentiates between true and false discoveries. Setting \( w_j = 1/2 \) recovers \( \tilde{\alpha}_j = \alpha_j \).  
To control \(\operatorname{FDR}(t) \leq \alpha\), we require  
\[
\sum_{j \in \mathcal{H}_0 \cap [t]} \frac{\tilde{\alpha}_j}{{R}({j-1}) + 1} \leq \alpha.
\]  
This implies  
\[
2w_j a_j + 2(1 - w_j) b_j + 1 \leq a_j + b_j + 1.
\]  
To improve power, we require  
\[
\tilde{\alpha}_j > \alpha_j,
\]  
which translates to  
\[
2w_j a_j + 2(1 - w_j) b_j + 1 > a_j + b_j + 1.
\]  
However, this contradicts the FDR control condition, making the approach infeasible.

\section{Additional Discussions on OCTF}\label{appen:OCTF}
{\color{black}This section collects several supplementary discussions and extensions of OCTF that complement Section~\ref{sec:OCT}. 
These include delayed feedback (Appendix~\ref{appen:OCTF-delay}), a general framework for score selection (Appendix~\ref{appen:model-selection}), extensions to online-updated score functions (Appendix~\ref{appen:OCTF-onlinelearning}), and the discussion of OCTF under null distribution shift (Appendix~\ref{appen:nullshift}).}

{\color{black}\subsection{Online conformal selection with delayed Feedback}\label{appen:OCTF-delay}

In this section, we study how to adapt OCTF to settings with delayed feedback with rigorous guarantees.
We focus on the case of a fixed delay $d$, where at time $t$ only the feedback
$\{\theta_j\}_{j=1}^{t-d-1}$ is available. To maintain finite-sample mFDR control
under delayed feedback, both the construction of the test levels $\alpha_t$
and the calibration data sets must be modified accordingly.

Specifically, note that the feedback associated with test $t$ is revealed at time
$t+d+1$. To accommodate this delay, we partition the time indices into $(d+1)$
disjoint sub-streams:
\[
S_1=\{1,\,2+d,\,3+2d,\,4+3d,\ldots\},
\]
\[
S_2=\{2,\,3+d,\,4+2d,\,5+3d,\ldots\},
\]
\[
\vdots
\]
\[
S_{d+1}=\{1+d,\,2+2d,\,3+3d,\,4+4d,\ldots\}.
\]
Within each sub-stream, the feedback for a test is revealed at the next time point
belonging to the same sub-stream. For each sub-stream $j\in\{1,2,\ldots,d+1\}$, we maintain an independent online
multiple testing procedure based on safe GAIF {\color{black}with an adjusted target level $\alpha/(d+1)$}, which is used to compute and update the
test levels for that sub-stream. In parallel, we maintain a separate calibration
data set
\[
\gC_{0t}^j=\{\, i< t : i\in S_j,\ \theta_i=0 \,\}.
\]

For a test observation $\X_t$ arriving at time $t\in S_j$, the test level $\alpha_t$
is determined by the current state of the GAIF procedure associated with
sub-stream $j$. The corresponding online conformal $p$-value is defined as
\begin{equation}\label{eq:conf_p_delay}
    p_t
    =
    \frac{
        \sum_{i\in \gC_{0t}^j} \mathbb{I}\{V_i < V_t\}
        + \xi_t \left(1 + \sum_{i\in \gC_{0t}^j} \mathbb{I}\{V_i = V_t\}\right)
    }{
        1 + |\gC_{0t}^j|
    },
\end{equation}
where $\xi_t \overset{\text{i.i.d.}}{\sim} \mathrm{Unif}[0,1]$ is an independent
random variable used for tie-breaking.

The resulting online conformal testing procedure with delayed feedback,
referred to as OCTF-delay, is summarized in Algorithm~\ref{alg:OCTF-delay}.
Finite-sample mFDR control for OCTF-delay is established in
Corollary~\ref{the:FDR-OCTF-delay}.

\begin{corollary}[Finite-sample online mFDR control for OCTF-delay]\label{the:FDR-OCTF-delay}
Suppose that Assumption~\ref{assump:conformal_setting} holds.
Then the rejection set $\gR$ produced by Algorithm~\ref{alg:OCTF-delay} satisfies
\[
\mFDR(t)\le \alpha \quad \text{for all } t\in\mathbb{N}.
\]
\end{corollary}

The proof relies on a sub-stream decomposition argument.
Under fixed delay $d$, the time indices are partitioned into $(d+1)$ disjoint
sub-streams. Within each sub-stream, feedback for a test is revealed before the
next test in the same sub-stream, so the corresponding online conformal $p$-values are
super-uniform and independent under the null conditional on the sub-stream filtration.
Applying safe GAIF separately to each sub-stream therefore guarantees finite-sample
mFDR control at level $\alpha/(d+1)$ within every sub-stream. \textcolor{black}{Since the
sub-streams are disjoint, summing the corresponding bounds and using
$\sum_{r=1}^{d+1}(1\vee R_r(t))\le (d+1)(1\vee R(t))$ yields
$\E[V(t)]\le \alpha\E[1\vee R(t)]$, and therefore the combined procedure
controls the overall mFDR at level $\alpha$.} 

}

{\color{black}





\begin{algorithm}[h!]
\small
\captionsetup{font=small}
\caption{Online conformal testing with delayed feedback (OCTF-delay)}
\label{alg:OCTF-delay}
\begin{algorithmic}[1]
\REQUIRE Initial calibration data $\gD_\gC=\{(\X_i,Y_i)\}_{i=-n+1}^0$, 
target region $\mathcal{A}$, FDR target level $\alpha\in(0,1)$, 
non-conformity score function $V(\cdot)$, initial wealth $s_0$, 
sequence $\{\gamma_t\}_{t\in\mathbb{N}}$, delay $d$.

\STATE Partition $\mathbb{N}$ into $(d+1)$ sub-streams $\{S_j\}_{j=1}^{d+1}$.
\STATE For each sub-stream $j$, initialize an independent GAIF procedure with initial wealth {\color{black}$s_0/(d+1)$ and target mFDR level $\alpha/(d+1)$}.
\STATE Initialize calibration sets $\gC_{01}^j=\{i\le 0: \theta_i=0\}$ for $j=1,\ldots,d+1$.

\FOR{$t=1,2,\ldots$}
    \STATE Observe test data $\X_t$
    \STATE Let $j$ be the index such that $t\in S_j$
    \STATE Compute conformity score $V_t$
    \STATE Compute conformal $p$-value $p_t$ using $\gC_{0t}^j$ via \eqref{eq:conf_p_delay}
    \STATE Compute test level $\alpha_t$ using the SFS (or LFS) update rule of sub-stream $j$ {\color{black}at level $\alpha/(d+1)$}
    \STATE \textbf{if} $p_t\le \alpha_t$ \textbf{then} $\delta_t=1$ \textbf{else} $\delta_t=0$
    \STATE Record decision $\delta_t$
    \STATE \textbf{if} feedback for some $i<t$ is revealed at time $t$:
        \STATE \quad update $\theta_i$ and the GAIF state of the corresponding sub-stream
        \STATE \quad if $\theta_i=0$, add $i$ to $\gC_{0t}^j$
\ENDFOR

\ENSURE Rejection set $\mathcal{R}=\{t:\delta_t=1\}$.
\end{algorithmic}
\end{algorithm}


We evaluate the proposed procedures in Scenario~IV under delayed feedback. We compare the sub-stream-based procedures with direct delayed-feedback implementations. At time $t$, these implementations form the calibration set from the null samples whose feedback has already been revealed, compute the corresponding online conformal $p$-value, and then apply LF-FD or SF-FD without the sub-stream modification. 

Figure~\ref{fig:cla-FDR-delayed} reports the results. All sub-stream-based procedures (SF-sub, LF-sub, SFS-sub, and LFS-sub) maintain empirical FDR below the target level under delayed feedback. Notably, SF-sub and LF-sub retain substantial power despite the sub-stream partition. Although the direct SF-FD and LF-FD implementations do not admit a finite-sample guarantee, they also maintain empirical FDR control in these simulations and attain higher power than their sub-stream-based counterparts. Thus, the direct implementations may be attractive when empirical power is the main concern and the lack of guarantee is acceptable, whereas the sub-stream construction provides a principled alternative when rigorous finite-sample error control is required.

\begin{figure}[htbp!]
		\centering
		\includegraphics[width=0.8\textwidth]{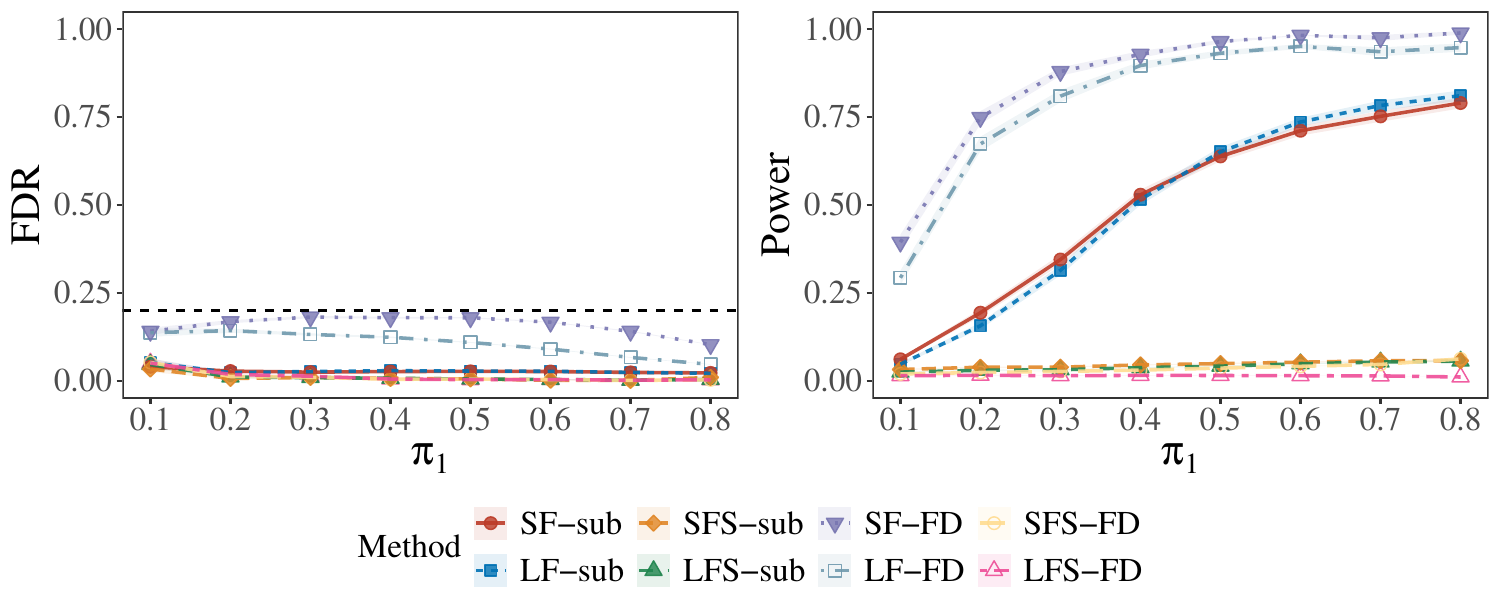}
		\caption{\small Results for Scenario {IV} with delayed feedback:  Line charts of FDR and Power with varying non-null proportion $\pi_1$ from $0.1$ to $0.8$ after $500$ replications; The black dashed lines denote the target FDR level $\alpha=0.2$. {Delay $d=5$}. }
  \label{fig:cla-FDR-delayed}
	\end{figure}

\subsection{A general validity framework and practical principles for optimized OCTF}\label{appen:model-selection}

We present a general validity-preserving framework for online conformal testing after score selection. The key idea is that, even when the score selection criterion is arbitrary, one can still restore the symmetry required for conformal validity by applying the same score selection rule under all permutations of the current null calibration set and test point.

At time $t$, let $\Omega_t$ denote the set of all permutations of $\gC_{0t}\cup\{t\}$ that leave indices outside $\gC_{0t}\cup\{t\}$ unchanged. For each $\sigma\in\Omega_t$, we apply the same score selection rule to the permuted data $(\gD_t)_\sigma$ and obtain
\[
\hat{k}_t^\sigma=\arg\min_{k\in[K]}\gM\bigl((\gD_t)_\sigma;k\bigr).
\]
We then define the permutation-based conformal $p$-value
\begin{equation}\label{eq:opt__perm_pvalue}
p_t^{\rm opt}
=
\frac{1}{|\Omega_t|}
\left[
\sum_{\sigma\in\Omega_t}
\mathbb{I}\{V(\X_{\sigma(t)};\hat{k}_t^\sigma)<V(\X_t;\hat{k}_t)\}
+
\xi_t\sum_{\sigma\in\Omega_t}
\mathbb{I}\{V(\X_{\sigma(t)};\hat{k}_t^\sigma)=V(\X_t;\hat{k}_t)\}
\right].
\end{equation}
Intuitively, under each permutation, the candidate test point $\sigma(t)$ is evaluated using the score function that would have been selected under that same permutation. In this way, the score selection step and the conformal ranking step are coupled symmetrically across all elements of $\gC_{0t}\cup\{t\}$, restoring the exchangeability needed for conformal validity.

With this construction, arbitrary score selection criteria can be incorporated without sacrificing finite-sample error control.

\begin{theorem}[{\color{black}Finite-sample online mFDR control based on full permutations}]\label{the:FDR-Opt-perm}
Under the same assumptions as in Theorem~\ref{the:FDR-OCTF}, for any score selection criterion $\gM$, the procedure in Algorithm~\ref{alg:OCTF-ind-conservative} with $\{p_t\}$ replaced by $\{p_t^{\rm opt}\}$ satisfies
\[
\mFDR(t)\le \alpha,\qquad \forall t\in\mathbb{N}.
\]
\end{theorem}

\paragraph{Practical principles for score selection}

Theorem~\ref{the:FDR-Opt-perm} shows that online conformal testing with score selection is possible in full generality. However, this result is mainly a theoretical benchmark: the full-permutation correction can be computationally infeasible in online settings and may also degrade power when the selected score function varies substantially across permutations. This observation motivates the search for score selection criteria that are not only valid in principle, but also practically effective.

Guided by this perspective, we highlight three principles that a practically useful criterion $\gM(\gD_t;\cdot)$ should satisfy.

\begin{itemize}
    \item \textbf{Power alignment.}
    At time $t$, a natural oracle objective for power is to choose the score function $k$ maximizing
    \[
    \Pr(p_{k,t}\le \alpha_t\mid \theta_t=1).
    \]
    A useful criterion should therefore be aligned with this target, in the sense that models with stronger non-null rejection ability receive better scores.

    \item \textbf{Permutation robustness.}
    The criterion should not vary drastically under permutation of $\gC_{0t}\cup\{t\}$. Otherwise, the selected score function $\hat{k}_t^\sigma$ may differ substantially across $\sigma$, so that $p_t^{\rm opt}$ compares scores computed under different score models. Such cross-model comparisons can inject extra variability, weaken the non-null signal, and reduce power. A sufficient condition avoiding this issue is that $\gM(\gD_t;\cdot)$ is symmetric with respect to $\gC_{0t}\cup\{t\}$, in which case $\hat{k}_t^\sigma\equiv \hat{k}_t$ for all $\sigma$.

    \item \textbf{Computational efficiency.}
    In principle, full permutation requires enumerating $|\Omega_t|=(|\gC_{0t}|+1)!$,
    which is prohibitive in online applications. Criteria with sufficient symmetry, especially permutation-invariant ones, allow one to avoid explicit enumeration and reduce computation essentially to the standard conformal ranking over $|\gC_{0t}|+1$ scores.
\end{itemize}

\paragraph{Alternative criteria}

These principles clarify why some seemingly natural criteria are unattractive in practice.

First, threshold-dependent criteria involving $\alpha_t$ typically violate both permutation robustness and computational efficiency. Since $\alpha_t$ is generated by the online testing rule and depends on the past rejection trajectory, such criteria are generally not symmetric with respect to $\gC_{0t}\cup\{t\}$. Moreover, a full-permutation correction would require recomputing the threshold process under each permutation, making the procedure both computationally burdensome and statistically unstable.

A second illustrative example is the criterion $\gM(\gD_t;k)=p_{k,t}$. At first glance, this choice appears attractive, since smaller $p$-values are more likely to lead to rejection. However, if used directly for score selection, it leads to severe validity issues: under the null, the procedure  selects the minimum of 
$K$ candidate $p$-values, which is generally not super-uniform.

Although Theorem~\ref{the:FDR-Opt-perm} implies that validity can still be restored through full permutation, this criterion remains practically problematic. Indeed, under this choice, for each permutation $\sigma\in\Omega_t$,
\[
\gM((\gD_t)_\sigma;k)
=
\frac{\sum_{i\in \gC_{0t}\cup\{t\}}\I\{V(\X_i;k)<V(\X_{\sigma(t)};k)\}
+\xi_t\cdot\sum_{i\in \gC_{0t}\cup\{t\}}\I\{V(\X_i;k)=V(\X_{\sigma(t)};k)\}}
{1+|\gC_{0t}|}
.
\]
Hence, when $\sigma(t)\neq t$ and denote $i=\sigma(t)$, the quantity $\gM((\gD_t)_\sigma;k)$ is simply a null conformal $p$-value for $V(\X_i;k)$ based on $\{V(\X_s;k)\}_{s\in\gC_{0t}\cup\{t\}}$. Score selection is then driven by null $p$-values, which contain little useful information for identifying a powerful model. As a result, $\hat{k}_t^\sigma$ may fluctuate nearly randomly across permutations, and the final $p_t^{\rm opt}$ in \eqref{eq:opt__perm_pvalue} is formed by ranking scores obtained under different selected score functions, possibly on different scales. This can substantially degrade power; empirically, we find that such instability may even make the procedure perform worse than random score selection.

In contrast, our EWMA criterion is designed to satisfy the above three principles: it is aligned with power through a threshold-free proxy, robust to the required permutation symmetry, and computationally efficient in online implementation. For this reason, although Theorem~\ref{the:FDR-Opt-perm} provides a general validity result for arbitrary criteria, we recommend the EWMA criterion for practical use.

Our framework also allows other power-aligned surrogates. For example, one may consider oracle-aligned objectives of the form $\E[\phi(p_{k,t})\mid \theta_t=1]$, where $\phi(\cdot)$ is a deterministic nondecreasing function that emphasizes small $p$-values, such as $\phi(p)=\min(p,c)$ for some constant $c$. Such choices could move the objective closer to the tail probability $\Pr(p_{k,t} \le \alpha_t \mid \theta_t=1)$. However, selecting $c$ appropriately, especially when $\alpha_t$ varies over time, would require additional tuning and theoretical analysis. In this paper, we therefore focus on $\mathcal{M}_t^{\mathrm{EWMA}}$ as a simple choice.

\subsection{Optimized OCTF with online-updated score functions}\label{appen:OCTF-onlinelearning}

Beyond selecting from a pool of pre-trained score functions, our framework also accommodates
\emph{data-driven score functions} that are updated online. The key requirement is the same as in Proposition~\ref{prop:online_conf_val_exch}: at time $t$, the score function used to construct the conformal $p$-value must be symmetric with respect to the samples in $\gC_{0t}\cup\{t\}$. Equivalently, when fitting the score model for the $t$-th test point, the current sample must be inserted as a provisional null sample and treated exchangeably with the historical null samples in $\gC_{0t}$.

We focus on the binary classification setting with $Y\in\{0,1\}$. The regression case can be handled similarly by replacing the response with the binary label $\mathbb{I}\{Y\in\mathcal{A}\}$. At each time $t$, for each candidate setup $k\in[K]$, we fit a one-step-ahead predictive model $\hat{\mu}_k^{(t)}$ by distinguishing the non-null samples $\{(\X_i,Y_i)\}_{i\in\gC_{1t}}$ from the null samples $\{(\X_i,Y_i)\}_{i\in\gC_{0t}}\cup\{(\X_t,0)\}$. The associated score function is then given by $V(\cdot;\hat{\mu}_k^{(t)})$.

This construction is valid as long as the training algorithm is symmetric to its input data (or more precisely, symmetric to the treated null samples $\{(\X_i,Y_i)\}_{i\in\gC_{0t}}\cup\{(\X_t,0)\}$), which is satisfied by most batch learning procedures, such as random forests and neural networks. In contrast, order-dependent updating rules, such as standard online gradient descent, generally violate this symmetry requirement and therefore cannot be directly used here.

In this perspective, we can extend our framework of score selection. The $K$ candidate setups may correspond to different score function classes or to the same score function with different hyper-parameters. When $K=1$, the procedure reduces to pure online model/score updating. When $K>1$, we update all $K$ score functions in parallel and then select the score function used for inference at time $t$. Specifically, we choose
\[
\hat{k}_t=\arg\min_{k\in[K]} \gM_t^{\mathrm{EWMA}}(\gD_t;k),
\]
where $\gM_t^{\mathrm{EWMA}}(\gD_t;k)$ is the same evaluation criterion as in Section~\ref{subsec:Opt-OCTF}. The optimized conformal $p$-value is then constructed using the selected score function based on $\hat{\mu}_{\hat{k}_t}^{(t)}$. Since each candidate score function is symmetric in $\gC_{0t}\cup\{t\}$, and the selection criterion depends only on permutation-invariant quantities, the selected score function remains symmetric as well. Therefore, the same argument as in Corollary~\ref{the:FDR-Opt} yields finite-sample mFDR control.

The detailed pseudo-code is given in Algorithm~\ref{alg:optOCTF-SD}.

\begin{algorithm}[h!]
    \small
    \captionsetup{font=small}
    \caption{Optimized OCTF with online-updated score functions}
    \label{alg:optOCTF-SD}
\begin{algorithmic}[1]
    \REQUIRE Initial data $\gD_1=\{(\X_i,Y_i)\}_{i=-N}^0$, target region $\mathcal{A}$, FDR target level $\alpha\in(0,1)$, $K$ candidate model setups, evaluation criterion $\mathcal{M}$, parameter $s_0$, parameter sequence $\{\gamma_t\}$, stopping time $T$.
    \FOR{$t=1,\dots,T$}
        \STATE Observe test covariate $\X_t$.
        \FOR{$k=1,\dots,K$}
            \STATE Fit the one-step-ahead model $\hat{\mu}_k^{(t)}$ by distinguishing
            $\{(\X_i,Y_i)\}_{i\in\gC_{1t}}$ from
            $\{(\X_i,Y_i)\}_{i\in\gC_{0t}}\cup\{(\X_t,0)\}$.
            \STATE Construct the score function $V(\cdot;\hat{\mu}_k^{(t)})$.
        \ENDFOR
        \STATE Select the predictive model by
        \[
        \hat{k}_t=\arg\min_{k\in[K]}\gM_t^{\mathrm{EWMA}}(\gD_t;k).
        \]
        \STATE Construct the optimized conformal $p$-value
        \[
        p_t^{\mathrm{opt}}
        =
        \frac{
        \sum_{i\in\gC_{0t}}
        \mathbb{I}\!\left\{
        V(\X_i;\hat{\mu}_{\hat{k}_t}^{(t)})
        <
        V(\X_t;\hat{\mu}_{\hat{k}_t}^{(t)})
        \right\}+\xi_t\left[1+\sum_{i\in\gC_{0t}}
        \mathbb{I}\!\left\{
        V(\X_i;\hat{\mu}_{\hat{k}_t}^{(t)})
        =
        V(\X_t;\hat{\mu}_{\hat{k}_t}^{(t)})
        \right\}\right]
        }{1+|\gC_{0t}|}.
        \]
        \STATE Update $\alpha_t=\alpha_t^{\mathrm{LFS}}$ in \eqref{eq:alpha-LORDF-conservative} or $\alpha_t=\alpha_t^{\mathrm{SFS}}$ in \eqref{eq:alpha-SAFFRONF-conservative}.
        \STATE Make a decision $\delta_t=\mathbb{I}\{p_t^{\mathrm{opt}}\le \alpha_t\}$.
        \STATE Observe the revealed feedback $Y_t$ (and hence $\theta_t$).
        \STATE Update the calibration sets $\gC_{0t}$ and $\gC_{1t}$.
    \ENDFOR
    \ENSURE Rejection set $\mathcal{R}_{\mathrm{opt}}=\{i:\delta_i=1,\ \delta_i\in\hdelta^T\}$.
\end{algorithmic}
\end{algorithm}

Although online updating of score functions improves data utilization and adaptivity, it is substantially more expensive than selecting from a fixed pool of pre-trained models. To preserve the required symmetry, each candidate model typically needs to be refit from scratch after adding the provisional null sample $(\X_t,0)$. For this reason, in the main text we focus on score selection among pre-trained score functions, and view online updating as a valid but computationally heavier extension.

\paragraph{Experimental Evaluations}
\label{sec:online_update_exp}
We illustrate this through an additional experiment.
In particular, we compare our approach with an alternative procedure that retrains the score function (predictive model) at every time step using all previously observed data. We adopt the same data-generating process as Scenario IV with sine pattern shift in Section \ref{subsec:simu-opt-OCTF}.  We set the target level $\alpha=0.2$.

Consider three strategies for constructing the predictive score function:

\begin{itemize}
    \item \textbf{Fixed score function.} A score function based on Random Forest is trained once using historical data and then kept fixed throughout the online testing process.
    
    \item \textbf{Score selection (ours).} We pre-train $K=3$ candidate score functions (Neural network, SVM, and Random Forest) using historical data, which induce $K$ candidate score functions. During the online testing process, the algorithm dynamically selects one candidate at each time step according to the feedback-based EWMA criterion described in Section~\ref{subsec:Opt-OCTF}.
    
    \item \textbf{Online retraining/updating.} At each time step $t$, the score function based on Random Forest is retrained using all previously observed data. To maintain the validity of conformal $p$-values, samples with feedback $\theta_t=0$ are incorporated into the null calibration set in the same way as other null observations.
\end{itemize}

For each strategy, we evaluate the empirical FDR and statistical power across the testing horizon. In addition, we record the total running time required to process a sequence of $T=300$ test samples, which measures the computational cost of the corresponding score construction method.

Table~\ref{tab:online_update} reports the empirical FDR and power for the three strategies. All methods control the empirical FDR at the target level. The online retraining (score updating) strategy achieves slightly higher power than the fixed model, but its improvement over the proposed score-selection approach is marginal. From a computational perspective, however, the difference is substantial. Online retraining requires fitting a predictive model at every time step, leading to a computational cost that grows linearly with the length of the testing sequence. In contrast, the proposed method only evaluates $K=3$ pre-trained candidate score functions and updates the EWMA statistics, which incurs negligible additional overhead.

\begin{table}[t]
\centering
\caption{Comparison of score construction strategies for SF.}
\label{tab:online_update}
\begin{tabular}{lccc}
\toprule
Strategy & FDR & Power & Runtime (s) \\
\midrule
Fixed score function & 0.162\scriptsize{(0.011)} & 0.829\scriptsize{(0.035)} & 0.621\scriptsize{(0.018)} \\
Score selection (ours) & 0.161\scriptsize{(0.009)} & 0.837\scriptsize{(0.031)} & 0.644\scriptsize{(0.028)} \\
Online retraining & 0.179\scriptsize{(0.010)} & 0.866\scriptsize{(0.026)} & 122.000\scriptsize{(1.650)} \\
\bottomrule
\end{tabular}
\end{table}



}

{\color{black}\subsection{OCTF under null distribution shifts}\label{appen:nullshift}
In this subsection, we examine how shifts in the null distribution affect OCTF and discuss possible remedies.

\paragraph{Performance evaluation under null distribution shift}
To assess robustness under non-stationarity, we consider a binary classification setting similar to Scenario IV in Section~\ref{subsec:simu-OCTF}, but introduce a time-varying drift in the null distribution. The initial training and calibration sets are sampled from the distribution at time \(t=0\), while the test stream evolves over time:

\begin{itemize}
    \item Null (\(Y_t=0\)): \(\mathbf X_t \sim \mathcal N(\bm\mu_t,\mathbf I_4)\), where
    \[
    \bm\mu_t=\Bigl(2-\Delta\frac{t}{T},0,0,0\Bigr)^\top,
    \]
    and \(\Delta\in\mathbb{R}\) controls the drift magnitude. 
    \item Alternative (\(Y_t=1\)): \(\mathbf X_t \sim \mathcal N(\bm\mu_1,\mathbf I_4)\) with fixed mean
    \[
    \bm\mu_1=(0,0,0,0)^\top.
    \]
\end{itemize}

We set \(T=1000\) and \(\pi_1=0.5\). To adapt to distribution shift, a natural idea is to use a sliding calibration window. We therefore compare:
\begin{itemize}
    \item \textbf{Growing window:} \(\gC_{0t}=\{i<t:\theta_i=0\}\).
    \item \textbf{Sliding window:} \(\gC_{0t}=\{i:t-W\le i<t,\theta_i=0\}\), with \(W=200\).
\end{itemize}
Under exchangeability, the growing window preserves the mutual independence of null online conformal \(p\)-values in Proposition~\ref{prop-inde-p}, whereas the sliding window introduces local dependence among consecutive conformal \(p\)-values.

We consider two drift directions, \(\Delta\in\{-2,2\}\). The results are shown in Figures~\ref{fig:cla-nullshift1}--\ref{fig:cla-nullshift2}. Our main findings are:

 \begin{figure}[htbp!]
		\centering
		\includegraphics[width=0.7\textwidth]{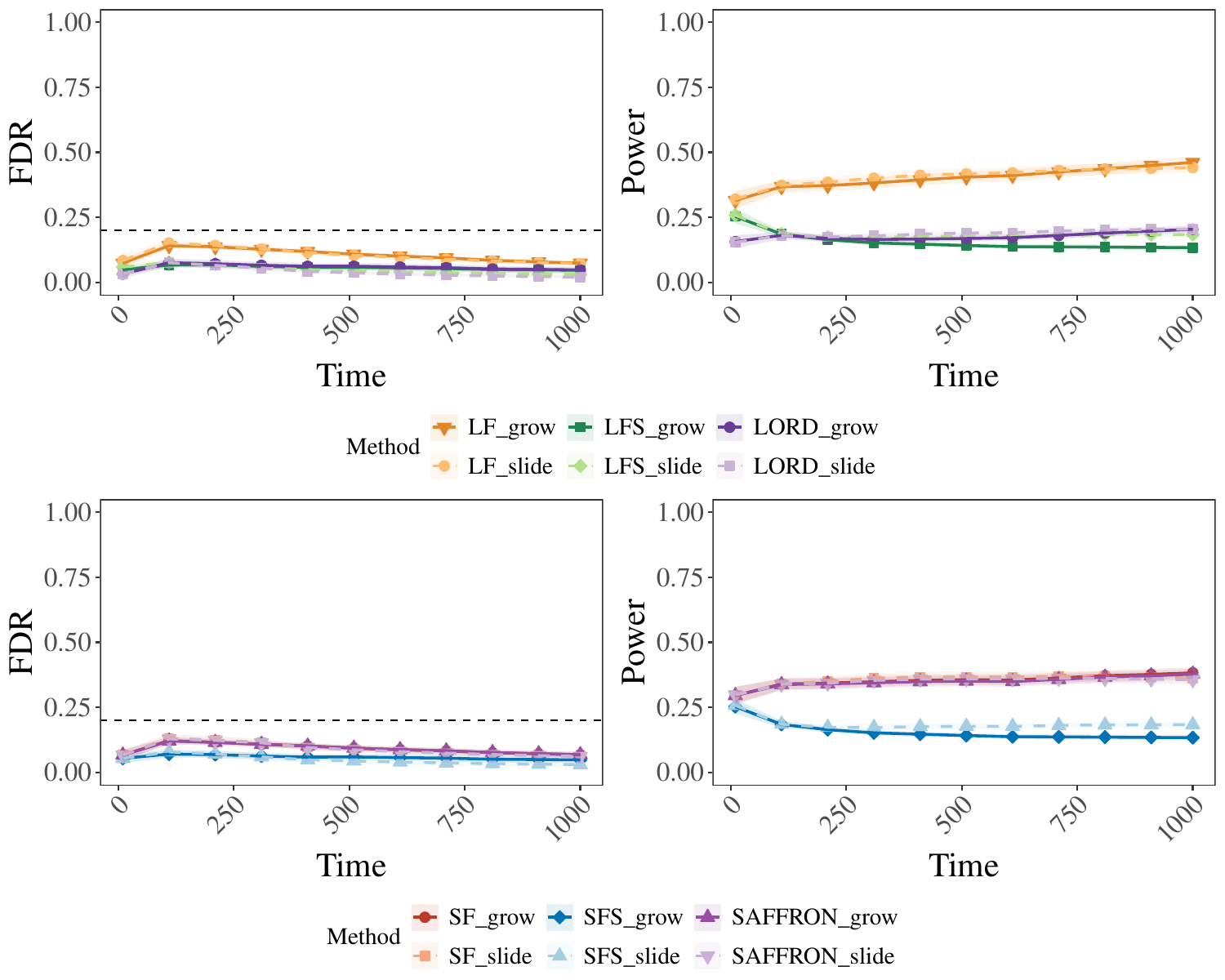}
		\caption{\small Results under null distribution shift:  Line charts of FDR and Power across time $t$; The black dashed lines denote the target FDR level $\alpha=0.2$. We set the drift parameter $\Delta=-2$ so that the null shifts away from the non-null.}
  \label{fig:cla-nullshift1}
	\end{figure}

 \begin{figure}[htbp!]
		\centering
		\includegraphics[width=0.7\textwidth]{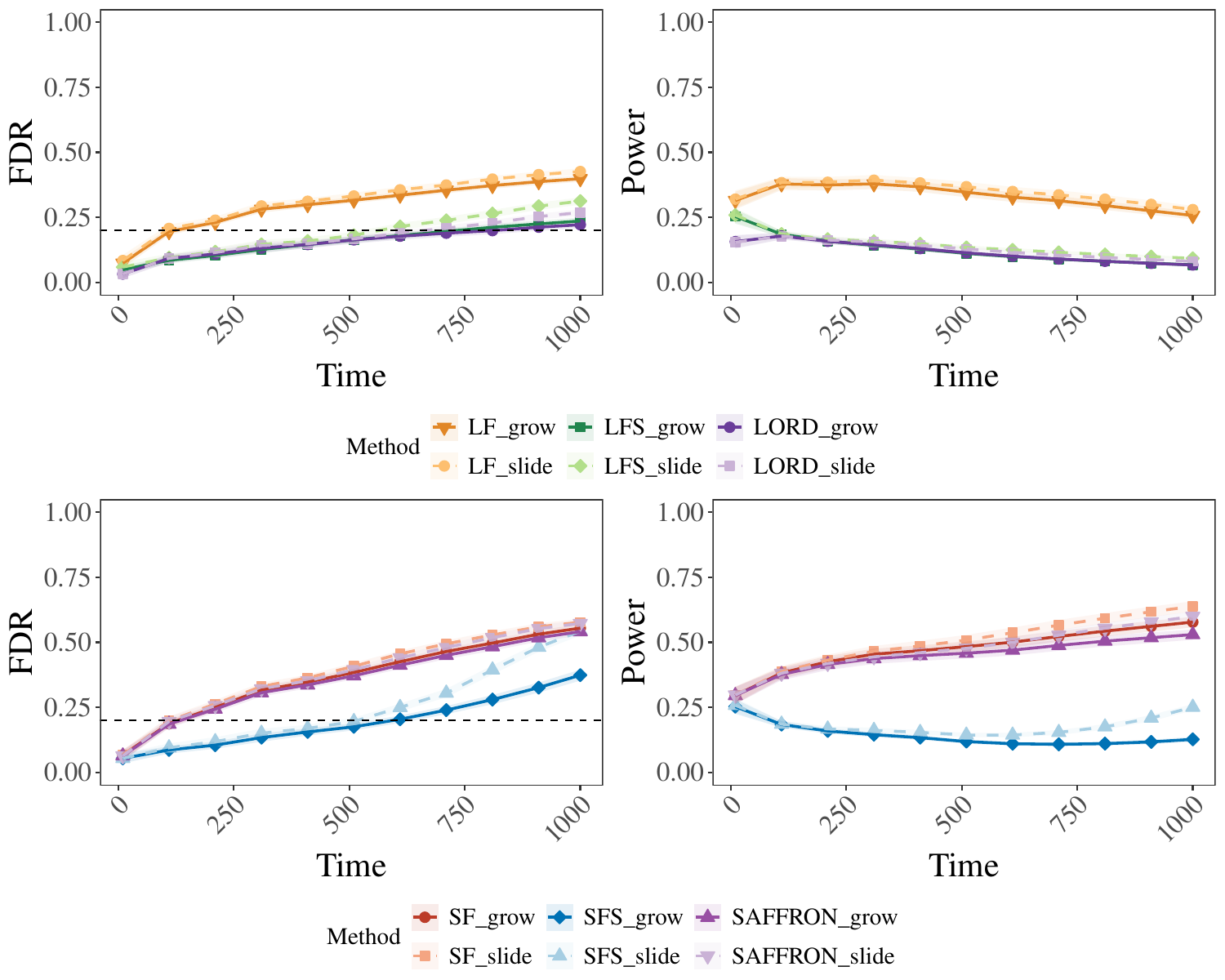}
		\caption{\small Results under null distribution shift:  Line charts of FDR and Power across time $t$; The black dashed lines denote the target FDR level $\alpha=0.2$. We set the drift magnitude $\Delta=2$ so that the null drifts towards the alternative/non-null region.}
  \label{fig:cla-nullshift2}
	\end{figure}

\begin{itemize}
    \item \textbf{Dual Effects of Drift:} Null drift does not necessarily invalidate OCTF. When the null moves away from the alternative, the procedure remains valid but becomes conservative. When the null moves toward the alternative, FDR inflation may occur.
    \item \textbf{Sliding windows may hurt validity.} Although sliding windows are intended to improve adaptation, in our experiments they exacerbate FDR inflation relative to the growing-window baseline.
    
    We attribute this to the loss of independence: the sliding window creates strong local dependence among nearby conformal \(p\)-values, and standard OCTF rules applied to such dependent \(p\)-values can inflate FDR. This suggests that simply discarding old data is not enough; the gain in adaptation may be outweighed by the loss of error-rate control. In contrast, the dependence-corrected \(\mathrm{LF}_{\mathrm{dep}}\) variant restores control, albeit at a substantial power cost.
\end{itemize}

\paragraph{Online weighted conformal $p$-values for handling null covariate shift}
These observations suggest that handling null distribution shift requires more principled corrections, such as \emph{online weighted conformal \(p\)-values}. However, obtaining exactly valid online \(p\)-values becomes highly nontrivial when the null distribution changes over time.

To illustrate this difficulty, consider a stylized setting in which all observations up to time \(t\) are null, but follow different distributions:
\[
V_s \sim F_s, \quad s=1,\ldots,t,
\]
with $F_s \neq F_{s'}$ in general. Let $\Omega_t$ denote the set of all permutations of $\{1,\ldots,t\}$. Given an unordered realization $\{V_1,\ldots,V_t\}=\{v_1,\ldots,v_t\}$, an oracle posterior over permutations can be defined as
\[
w(\{v_i\}_{i=1}^t;\sigma)
\;\propto\;
\prod_{s=1}^t f_s\!\bigl(v_{\sigma(s)}\bigr),
\qquad \sigma\in\Omega_t,
\]
where $f_s$ is the null density of $F_s$ and the weights are normalized so that
$\sum_{\sigma\in\Omega_t} w(\{v_i\}_{i=1}^t;\sigma)=1$. Intuitively, $w(\{v_i\};\sigma)$ quantifies how plausible it is that the time labels $1,\ldots,t$ are assigned to the observed $\{v_i\}_{i=1}^t$ according to $\sigma$.

Using this oracle weighting, an online conformal $p$-value could be constructed by averaging the randomized rank over permutations:
\[
p_t^{w}=
\sum_{\sigma\in\Omega_t}
w(\{V_i\}_{i=1}^t;\sigma)
\left(
\mathbb{I}\{V_{\sigma(t)}<V_t\}
+
\xi_t \mathbb{I}\{V_{\sigma(t)}=V_t\}
\right).
\]
In principle, such oracle weighting can restore super-uniformity and mutual independence by integrating over the uncertainty induced by non-identical null distributions. In offline covariate-shift settings, this weighting may simplify substantially \citep{jin2025model}; for example, if only the test point has a different null distribution while the calibration points remain i.i.d., the weight depends only on the identity of the test point. In the online setting, however, when all \(F_1,\dots,F_t\) may differ, the weight generally depends on the entire permutation, leading to factorial computational complexity. Moreover, estimating the full sequence of drifting null distributions is itself difficult, making practical implementations unstable.

\paragraph{An empirical investigation based on weighted conformal $p$-values}

A practical workaround is to maintain a fixed calibration set drawn from a stable null distribution and compute weighted conformal $p$-values relative to this reference set, similar to the offline covariate-shift approach in \citet{prinster2025watch}, although a rigorous theoretical guarantee remains open. We conduct a real-data application to verify the performance of adopting online weighted conformal \(p\)-values, where oracle weights are used.

We use the airfoil dataset \citep{misc_airfoil_self-noise_291} from the UCI Machine Learning Repository. The task is to identify samples with high sound pressure. The dataset contains \(n=1503\) observations, where the response \(Y\) is the scaled sound pressure and the covariates \(\X\in\mathbb R^5\) include log frequency, angle of attack, chord length, free-stream velocity, and suction-side log displacement thickness. We test
\[
\mathbb H_{0t}: Y_t\in (-\infty,c),
\]
where \(c\) is the \((1-\pi_1)\)-quantile of \(Y\), and \(\pi_1\in\{0.2,\dots,0.8\}\) controls the non-null proportion.

To simulate covariate shift, we construct a shifted test set \(D_{\mathrm{shift}}\) by resampling \(25\%\) of the points from \(D_{\mathrm{test}}\) with probabilities proportional to
\[
\omega(x)=\exp(x^\top\beta),
\qquad
\beta=(2,0,0,0,-3).
\]

Following \citet{prinster2025watch}, we use a fixed calibration set and compare weighted conformal \(p\)-values with their unweighted counterparts. Figure~\ref{fig:airfoil_prop} reports the results for weighted methods (wSFS, wLFS, wSF, and wLF) and unweighted methods (SFS, LFS, SF, and LF). The weighted safe procedures wSFS and wLFS consistently maintain valid empirical FDR control under covariate shift, whereas SF, LF, and SFS suffer from FDR inflation due to their reliance on unweighted conformal \(p\)-values.

\begin{figure}[h!]
		\centering
		\includegraphics[width=0.8\textwidth]{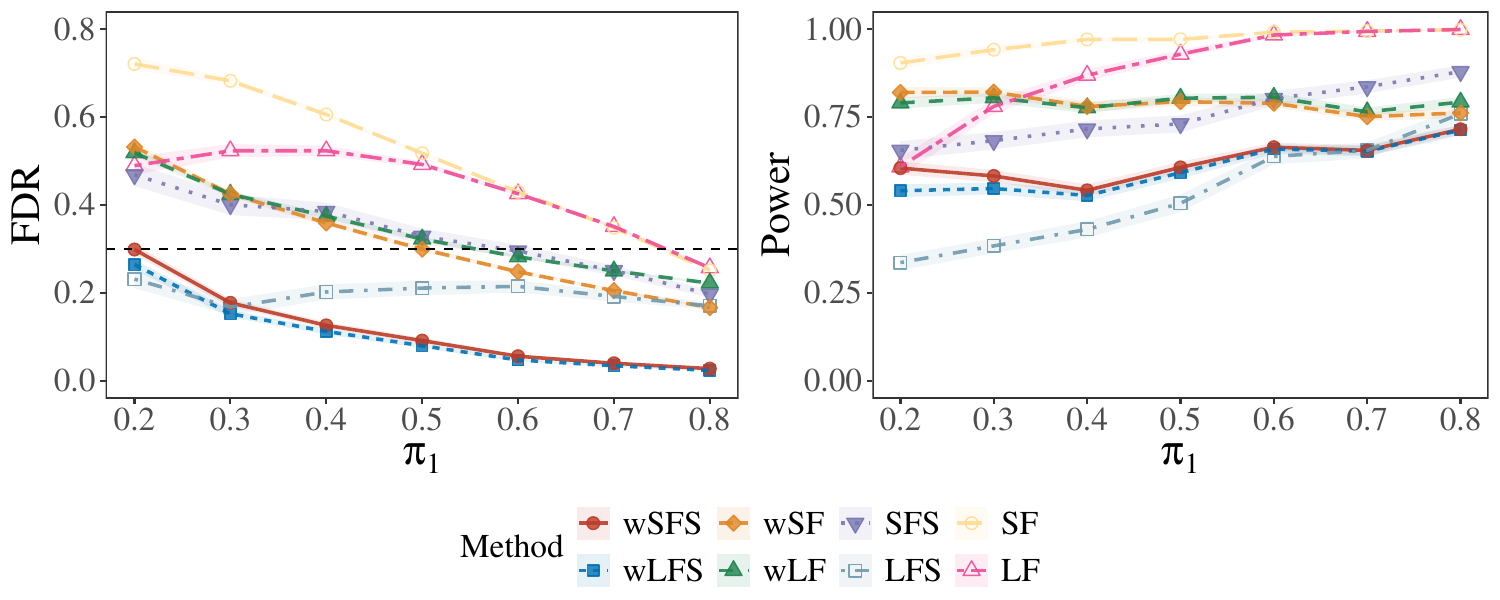}
		\caption{\small Results for airfoil data: the values of $\FDR(T)$  and $\operatorname{Power}(T)$ for all methods. The training algorithm is random forest. The black dashed lines denote the FDR level $\alpha=0.3$. The non-null proportion is $\pi_1\in\{0.2,\dots,0.8\}$. }
  \label{fig:airfoil_prop}
	\end{figure}
}

\section{Applications on Real-time LLM Alignment}\label{appen:LLM}

In this section, we introduce the potential application of our proposed OCTF procedure on the task of real-time LLM alignment. For example, in medical report generation tasks, we may need to sequentially select radiology images \( t \in \{1, \dots, T\} \) for which the generated reports align with expert standards. Similarly, in question-answering tasks, our goal is to identify the generated answer that best matches the true reference answer in an online fashion. Specifically, let \( f: \mathcal{X} \to \mathcal{Y} \) be a pre-trained foundation model that maps a prompt to an output. A holdout set \( \mathcal{D} = (\X_i, E_i)_{i=-n}^0 \) is available, where \( \X_i \in \mathcal{X} \) represents an input prompt, and \( E_i \in \mathcal{E} \) serves as a reference for assessing alignment. The alignment score function \( \mathcal{A}: \mathcal{Y} \times \mathcal{E} \to \mathbb{R} \) maps the generated output \( f(\X) \) and reference \( E \) to an alignment score \( A = \mathcal{A}(f(\X), E) \). For example, \( A \) may represent the similarity score between the machine-generated report \( f(\X) \) and a human expert report \( E \). The test data \( \{\X_{t}\}_{t=1}^T \) arrive sequentially, forming the following online multiple hypothesis testing problem at time \( t \):
\[
\mathbb{H}_{0t}: A_{t} \leq c \quad \text{versus} \quad \mathbb{H}_{1t}: A_{t} > c,
\]
where \( c \in \mathbb{R} \) is a pre-specified threshold. After making a decision \( \delta_t \) at time \( t \), the corresponding human expert report is revealed either immediately or with a delay of \( d \) time steps. 

Our goal is to control the online $\operatorname{FDR}$:
\begin{equation}
    \operatorname{FDR}(t)=\E\left[\frac{\sum_{j\leq t}\mathbb{I}\{A_{j}\leq c, \delta_j=1\}}{1\vee \sum_{j\leq t}\delta_j}\right]\leq \alpha.
\end{equation}

Following the conformal alignment framework \cite{guiconformal}, we randomly split \( \mathcal{D} \) into two subsets: a training set \( \mathcal{D}_\gT \) and an initial calibration set \( \mathcal{D}_{\gC} \). Using \( \mathcal{D}_\gT\), we train an alignment predictor $g(\X)$ to estimate the alignment score based on features of LLM outputs \( \X_{t} \) and  compute the predicted alignment scores $\widehat{A}_t=g(\X_t)$ for each $t\in[T]$. Then applying OCTF with \( \mathcal{D}_{\gC_{0t}} \) to select new images whose generated reports are aligned with expert standards ensures finite-sample online mFDR control according to Theorem \ref{the:FDR-OCTF}, where the calibration set $\gC_{0t}$ is updated online as the feedback is revealed. The conformal $p$-value for $t\in[T]$ is:
\begin{equation}\label{eq:conf_p_LLM}
    p_t = \frac{1 + \sum_{i \in \gC_{0t}} \I\{\widehat{A}_i\geq \widehat{A}_{t}\}}{1 + |\gC_{0t}|}.
\end{equation}  

The real-time LLM alignment procedure is summarized in Algorithm~\ref{alg:LLM-alg}, which, similar to Theorem~\ref{the:FDR-OCTF}, guarantees finite-sample mFDR control under the same assumptions.

\begin{algorithm}[h!]
\caption{Real-time LLM conformal Alignment with feedback}
\label{alg:LLM-alg}
\begin{algorithmic}[1]
\REQUIRE Pre-trained foundation model $f$;  alignment score function $\mathcal{A}$; 
reference dataset $\mathcal{D} = (\X_i,E_i)_{i=-n}^0$; 
algorithm for fitting alignment predictor $\mathcal{G}$; 
alignment level $c$; target FDR level $\alpha$.
\STATE Compute the alignment score $A_i = \mathcal{A}(f(\X_i),E_i)$,  $\forall i \in \mathcal{D}$.\;
\STATE Randomly split $\mathcal{D}$ into two disjoint sets: the training set $\mathcal{D}_{\gT}$ and 
the calibration set $\mathcal{D}_{\gC}$.
\STATE Fit the alignment score predictor with $\mathcal{D}_{\gT}$: 
$g \leftarrow \mathcal{G}(\mathcal{D}_{\gT})$.\; 

 \STATE Initialize $\gC_{0t}=\{i\in\gC:A_i\leq c\}$
\FOR{$t \in [T]$}
\STATE Observe test data $\X_{n+t}$
\STATE Compute the predicted alignment score: $\hat A_i \leftarrow g(\X_i)$, 
 $\forall i\in \gC$ and $\hat A_t \leftarrow g(\X_t)$.\;
\STATE Compute the conformal p-value $p_t$ according to Equation~\eqref{eq:conf_p_LLM}.\;
 \STATE  Update $\alpha_t=\alpha_t^{\text{LFS}}$ in Equation (\ref{eq:alpha-LORDF-conservative}) (or $\alpha_t=\alpha_t^{\text{SFS}}$ in (\ref{eq:alpha-SAFFRONF-conservative})) 
\STATE Obtain the revealed feedback $\theta_{t}$.\;
\STATE Update the calibration dataset $\gC_{0t}$.
\ENDFOR

\ENSURE The selection set $\mathcal{R}=\{t:\delta_t=1,t\in[T]\}$.
\end{algorithmic}
\end{algorithm}

\section{Additional Experimental Results}\label{appen:add-experiment}

In this section, we provide additional experimental results to further demonstrate the superior performance of our proposed algorithms, with a focus on online conformal testing.  

{\color{black}Specifically, Appendix \ref{add:numericalvalue} reports numerical results for the experiments in Sections \ref{sec:simu} and \ref{sec:real-data}; Appendix \ref{add:experement-reg} reports additional results for a regression example (Scenario~V); Appendix \ref{add:experement-training-alg} reports results under different training algorithms for \(\hat{\mu}\); Appendix \ref{add:experement-score-selection} reports additional results after score selection; Appendix \ref{add:experement-mFDR} reports empirical mFDR results; and Appendix \ref{appen_subsec:simu_auxi_modelsel} reports results on auxiliary calibration for score selection; Appendix \ref{app:safe-baselines} reports the results for comparison  with alternative finite-sample-valid baselines in online conformal testing.} 


{\color{black}\subsection{Additional Numerical Results for Sections~\ref{sec:simu} and~\ref{sec:real-data}}
\label{add:numericalvalue}

This section provides numerical summaries complementing the simulation experiments in Section~\ref{sec:simu} and the real-data analysis in Section~\ref{sec:real-data}. Tables~\ref{tab:sim_results}--
\ref{tab:sim_results_classification} report the FDR and Power values underlying Figures~\ref{fig:GAIF_Gaussian}--
\ref{fig:cla_stop_A_vary_prop}, respectively. Table~\ref{tab:stopping_time_results} reports the results for Scenario {IV} with sine pattern shift at the stopping time ($T=1000$). Table~\ref{tab:appendix_safe_fdr_power} further reports the performance of the safe variants \text{Opt-SFS} and \text{Opt-LFS} on the real-data tasks of Section~\ref{sec:real-data}.
}

\begin{table}[!htbp]
\centering
\caption{FDR and Power of competing methods under Scenario I and Scenario II. Results are averaged over 500 simulation runs. Standard errors are given in parentheses.}
\label{tab:sim_results}
\resizebox{\textwidth}{!}{%
\begin{tabular}{ll rrrrrrrr}
\toprule
& & \multicolumn{8}{c}{$\pi_1$} \\
\cmidrule(lr){3-10}
Scenario & Method & 0.1 & 0.2 & 0.3 & 0.4 & 0.5 & 0.6 & 0.7 & 0.8 \\
\midrule
\multirow{10}{*}{I}
& \multicolumn{9}{l}{\textit{FDR}} \\
& LOND    & .023\tiny{(.005)} & .012\tiny{(.004)} & .008\tiny{(.003)} & .006\tiny{(.001)} & .005\tiny{(.001)} & .005\tiny{(.001)} & .003\tiny{(.001)} & .001\tiny{(.000)} \\
& LORD++  & .021\tiny{(.006)} & .011\tiny{(.003)} & .010\tiny{(.002)} & .013\tiny{(.001)} & .010\tiny{(.001)} & .009\tiny{(.001)} & .007\tiny{(.001)} & .005\tiny{(.000)} \\
& SAFFRON & .043\tiny{(.007)} & .053\tiny{(.005)} & .063\tiny{(.003)} & .073\tiny{(.002)} & .077\tiny{(.001)} & .072\tiny{(.001)} & .063\tiny{(.001)} & .053\tiny{(.001)} \\
& LF      & .063\tiny{(.009)} & .068\tiny{(.007)} & .071\tiny{(.004)} & .086\tiny{(.003)} & .092\tiny{(.001)} & .091\tiny{(.001)} & .090\tiny{(.001)} & .086\tiny{(.000)} \\
& SF      & .065\tiny{(.009)} & .069\tiny{(.007)} & .075\tiny{(.004)} & .086\tiny{(.002)} & .091\tiny{(.001)} & .091\tiny{(.001)} & .088\tiny{(.001)} & .080\tiny{(.001)} \\
\cmidrule(lr){2-10}
& \multicolumn{9}{l}{\textit{Power}} \\
& LOND    & .009\tiny{(.001)} & .012\tiny{(.000)} & .014\tiny{(.000)} & .016\tiny{(.000)} & .018\tiny{(.000)} & .020\tiny{(.000)} & .022\tiny{(.000)} & .024\tiny{(.000)} \\
& LORD++  & .007\tiny{(.001)} & .013\tiny{(.001)} & .021\tiny{(.001)} & .031\tiny{(.001)} & .046\tiny{(.001)} & .059\tiny{(.001)} & .074\tiny{(.001)} & .090\tiny{(.002)} \\
& SAFFRON & .010\tiny{(.001)} & .028\tiny{(.002)} & .080\tiny{(.003)} & .169\tiny{(.004)} & .278\tiny{(.003)} & .381\tiny{(.003)} & .482\tiny{(.003)} & .584\tiny{(.002)} \\
& LF      & .009\tiny{(.001)} & .023\tiny{(.001)} & .076\tiny{(.003)} & \textbf{.184}\tiny{(.004)} & \textbf{.335}\tiny{(.003)} & \textbf{.477}\tiny{(.002)} & \textbf{.620}\tiny{(.002)} & \textbf{.771}\tiny{(.001)} \\
& SF      & \textbf{.011}\tiny{(.001)} & \textbf{.031}\tiny{(.002)} & \textbf{.096}\tiny{(.003)} & .210\tiny{(.004)} & .335\tiny{(.003)} & .465\tiny{(.002)} & .594\tiny{(.002)} & .738\tiny{(.002)} \\
\midrule
\multirow{10}{*}{II}
& \multicolumn{9}{l}{\textit{FDR}} \\
& LOND    & .014\tiny{(.004)} & .012\tiny{(.003)} & .008\tiny{(.002)} & .005\tiny{(.001)} & .002\tiny{(.001)} & .002\tiny{(.001)} & .002\tiny{(.001)} & .003\tiny{(.001)} \\
& LORD++  & .013\tiny{(.004)} & .012\tiny{(.003)} & .012\tiny{(.002)} & .013\tiny{(.001)} & .010\tiny{(.001)} & .009\tiny{(.001)} & .006\tiny{(.001)} & .005\tiny{(.000)} \\
& SAFFRON & .025\tiny{(.005)} & .042\tiny{(.004)} & .063\tiny{(.004)} & .074\tiny{(.002)} & .080\tiny{(.001)} & .084\tiny{(.001)} & .079\tiny{(.001)} & .072\tiny{(.001)} \\
& LF      & .073\tiny{(.010)} & .067\tiny{(.008)} & .077\tiny{(.006)} & .084\tiny{(.002)} & .089\tiny{(.001)} & .092\tiny{(.001)} & .090\tiny{(.001)} & .087\tiny{(.000)} \\
& SF      & .049\tiny{(.007)} & .061\tiny{(.005)} & .074\tiny{(.003)} & .081\tiny{(.002)} & .086\tiny{(.001)} & .085\tiny{(.001)} & .083\tiny{(.001)} & .080\tiny{(.001)} \\
\cmidrule(lr){2-10}
& \multicolumn{9}{l}{\textit{Power}} \\
& LOND    & .011\tiny{(.001)} & .012\tiny{(.000)} & .014\tiny{(.000)} & .015\tiny{(.000)} & .016\tiny{(.000)} & .018\tiny{(.000)} & .020\tiny{(.000)} & .022\tiny{(.000)} \\
& LORD++  & .009\tiny{(.001)} & .013\tiny{(.001)} & .019\tiny{(.001)} & .028\tiny{(.001)} & .036\tiny{(.001)} & .049\tiny{(.001)} & .059\tiny{(.001)} & .070\tiny{(.001)} \\
& SAFFRON & .008\tiny{(.001)} & .024\tiny{(.001)} & .064\tiny{(.003)} & .152\tiny{(.004)} & .280\tiny{(.005)} & .440\tiny{(.005)} & .616\tiny{(.004)} & .785\tiny{(.002)} \\
& LF      & .010\tiny{(.001)} & .022\tiny{(.001)} & .061\tiny{(.002)} & .172\tiny{(.004)} & .328\tiny{(.004)} & \textbf{.517}\tiny{(.003)} & \textbf{.713}\tiny{(.002)} & \textbf{.875}\tiny{(.001)} \\
& SF      & \textbf{.017}\tiny{(.001)} & \textbf{.042}\tiny{(.002)} & \textbf{.104}\tiny{(.003)} & \textbf{.212}\tiny{(.003)} & \textbf{.338}\tiny{(.003)} & .477\tiny{(.003)} & .640\tiny{(.003)} & .800\tiny{(.002)} \\
\bottomrule
\end{tabular}%
}
\end{table}

\begin{table}[!htbp]
\centering
\caption{FDR and Power of competing methods under dependent setting (Scenario III). Results are averaged over 500 simulation runs. Standard errors are given in parentheses.}
\label{tab:sim_results_dep}
\resizebox{\textwidth}{!}{%
\begin{tabular}{ll rrrrrrrr}
\toprule
& & \multicolumn{8}{c}{$\pi_1$} \\
\cmidrule(lr){3-10}
& Method & 0.1 & 0.2 & 0.3 & 0.4 & 0.5 & 0.6 & 0.7 & 0.8 \\
\midrule
\multicolumn{10}{l}{\textit{FDR}} \\
& LOND        & .010\tiny{(.003)} & .007\tiny{(.002)} & .006\tiny{(.002)} & .004\tiny{(.001)} & .002\tiny{(.001)} & .002\tiny{(.000)} & .002\tiny{(.001)} & .001\tiny{(.000)} \\
& LORD++      & .034\tiny{(.006)} & .033\tiny{(.005)} & .030\tiny{(.003)} & .025\tiny{(.003)} & .013\tiny{(.002)} & .015\tiny{(.002)} & .013\tiny{(.001)} & .007\tiny{(.001)} \\
& LORDdep     & .023\tiny{(.004)} & .022\tiny{(.003)} & .020\tiny{(.003)} & .018\tiny{(.002)} & .009\tiny{(.001)} & .009\tiny{(.001)} & .009\tiny{(.001)} & .004\tiny{(.001)} \\
& SAFFRON     & .270\tiny{(.014)} & .312\tiny{(.011)} & .280\tiny{(.008)} & .236\tiny{(.006)} & .191\tiny{(.005)} & .160\tiny{(.003)} & .117\tiny{(.003)} & .085\tiny{(.002)} \\
& SAFFRONdep  & .036\tiny{(.005)} & .041\tiny{(.004)} & .038\tiny{(.003)} & .037\tiny{(.003)} & .025\tiny{(.002)} & .025\tiny{(.002)} & .019\tiny{(.001)} & .012\tiny{(.001)} \\
& LF          & .156\tiny{(.011)} & .177\tiny{(.009)} & .155\tiny{(.006)} & .126\tiny{(.005)} & .105\tiny{(.004)} & .097\tiny{(.002)} & .083\tiny{(.002)} & .071\tiny{(.001)} \\
& LFdep       & .030\tiny{(.005)} & .029\tiny{(.004)} & .026\tiny{(.003)} & .024\tiny{(.002)} & .013\tiny{(.001)} & .014\tiny{(.001)} & .012\tiny{(.001)} & .006\tiny{(.001)} \\
& SF          & .270\tiny{(.014)} & .312\tiny{(.011)} & .281\tiny{(.008)} & .236\tiny{(.006)} & .191\tiny{(.005)} & .160\tiny{(.003)} & .117\tiny{(.003)} & .085\tiny{(.002)} \\
& SFdep       & .031\tiny{(.005)} & .035\tiny{(.004)} & .032\tiny{(.003)} & .033\tiny{(.003)} & .022\tiny{(.002)} & .023\tiny{(.002)} & .019\tiny{(.001)} & .014\tiny{(.001)} \\
\midrule
\multicolumn{10}{l}{\textit{Power}} \\
& LOND        & .126\tiny{(.006)} & .155\tiny{(.006)} & .186\tiny{(.006)} & .209\tiny{(.007)} & .210\tiny{(.007)} & .244\tiny{(.007)} & .248\tiny{(.007)} & .255\tiny{(.008)} \\
& LORD++      & .085\tiny{(.006)} & .128\tiny{(.008)} & .193\tiny{(.009)} & .225\tiny{(.010)} & .237\tiny{(.010)} & .295\tiny{(.011)} & .311\tiny{(.011)} & .310\tiny{(.012)} \\
& LORDdep     & .079\tiny{(.006)} & .117\tiny{(.007)} & .176\tiny{(.009)} & .205\tiny{(.010)} & .217\tiny{(.010)} & .274\tiny{(.011)} & .288\tiny{(.011)} & .285\tiny{(.012)} \\
& SAFFRON     & .282\tiny{(.012)} & .416\tiny{(.012)} & .515\tiny{(.012)} & .588\tiny{(.012)} & .648\tiny{(.011)} & .748\tiny{(.009)} & .788\tiny{(.009)} & .857\tiny{(.007)} \\
& SAFFRONdep  & .150\tiny{(.008)} & .236\tiny{(.009)} & .323\tiny{(.010)} & .387\tiny{(.012)} & .438\tiny{(.011)} & .519\tiny{(.011)} & .542\tiny{(.011)} & .580\tiny{(.010)} \\
& LF          & .188\tiny{(.009)} & .347\tiny{(.011)} & .490\tiny{(.011)} & \textbf{.604}\tiny{(.011)} & \textbf{.698}\tiny{(.010)} & \textbf{.808}\tiny{(.007)} & \textbf{.858}\tiny{(.006)} & \textbf{.905}\tiny{(.005)} \\
& LFdep       & .145\tiny{(.007)} & .234\tiny{(.009)} & .308\tiny{(.009)} & .372\tiny{(.010)} & .417\tiny{(.010)} & .478\tiny{(.009)} & .490\tiny{(.010)} & .520\tiny{(.009)} \\
& SF          & \textbf{.283}\tiny{(.012)} & \textbf{.416}\tiny{(.012)} & \textbf{.515}\tiny{(.012)} & .589\tiny{(.012)} & .650\tiny{(.011)} & .748\tiny{(.009)} & .791\tiny{(.008)} & .862\tiny{(.007)} \\
& SFdep       & .182\tiny{(.009)} & .270\tiny{(.010)} & .351\tiny{(.010)} & .412\tiny{(.011)} & .458\tiny{(.011)} & .542\tiny{(.010)} & .572\tiny{(.010)} & .628\tiny{(.009)} \\
\bottomrule
\end{tabular}%
}
\end{table}
\begin{table}[!htbp]
\centering
\caption{FDR and Power of competing methods under bandit and instant feedback setting. Results are averaged over 500 simulation runs. Standard errors are given in parentheses.}
\label{tab:sim_results_bi}
\resizebox{\textwidth}{!}{%
\begin{tabular}{l rrrrrrrr}
\toprule
& \multicolumn{8}{c}{$\pi_1$} \\
\cmidrule(lr){2-9}
Method & 0.1 & 0.2 & 0.3 & 0.4 & 0.5 & 0.6 & 0.7 & 0.8 \\
\midrule
\multicolumn{9}{l}{\textit{FDR}} \\
LOND    & .014\tiny{(.004)} & .012\tiny{(.003)} & .009\tiny{(.002)} & .008\tiny{(.002)} & .004\tiny{(.001)} & .003\tiny{(.001)} & .003\tiny{(.001)} & .002\tiny{(.000)} \\
LORD++  & .004\tiny{(.002)} & .002\tiny{(.001)} & .006\tiny{(.001)} & .006\tiny{(.001)} & .008\tiny{(.001)} & .006\tiny{(.001)} & .004\tiny{(.000)} & .003\tiny{(.000)} \\
SAFFRON & .049\tiny{(.008)} & .050\tiny{(.006)} & .065\tiny{(.004)} & .071\tiny{(.002)} & .074\tiny{(.001)} & .071\tiny{(.001)} & .064\tiny{(.001)} & .053\tiny{(.000)} \\
LF-BI   & .045\tiny{(.007)} & .058\tiny{(.007)} & .070\tiny{(.005)} & .084\tiny{(.003)} & .089\tiny{(.002)} & .091\tiny{(.001)} & .090\tiny{(.001)} & .086\tiny{(.000)} \\
SF-BI   & .049\tiny{(.008)} & .050\tiny{(.006)} & .065\tiny{(.004)} & .071\tiny{(.002)} & .074\tiny{(.001)} & .071\tiny{(.001)} & .064\tiny{(.001)} & .053\tiny{(.000)} \\
\midrule
\multicolumn{9}{l}{\textit{Power}} \\
LOND    & .009\tiny{(.000)} & .013\tiny{(.000)} & .014\tiny{(.000)} & .016\tiny{(.000)} & .018\tiny{(.000)} & .021\tiny{(.000)} & .022\tiny{(.000)} & .026\tiny{(.001)} \\
LORD++  & .004\tiny{(.000)} & .007\tiny{(.001)} & .012\tiny{(.001)} & .020\tiny{(.001)} & .028\tiny{(.001)} & .039\tiny{(.002)} & .050\tiny{(.002)} & .066\tiny{(.002)} \\
SAFFRON & .011\tiny{(.001)} & .028\tiny{(.002)} & .082\tiny{(.003)} & .171\tiny{(.004)} & .271\tiny{(.004)} & .377\tiny{(.003)} & .483\tiny{(.002)} & .588\tiny{(.002)} \\
LF-BI   & .010\tiny{(.001)} & .024\tiny{(.001)} & .074\tiny{(.003)} & .183\tiny{(.004)} & .313\tiny{(.004)} & .462\tiny{(.004)} & .606\tiny{(.003)} & .765\tiny{(.002)} \\
SF-BI   & .011\tiny{(.001)} & .028\tiny{(.002)} & .082\tiny{(.003)} & .171\tiny{(.004)} & .271\tiny{(.004)} & .377\tiny{(.003)} & .483\tiny{(.002)} & .588\tiny{(.002)} \\
\bottomrule
\end{tabular}%
}
\end{table}

\begin{table}[!htbp]
\centering
\caption{FDR and Power of competing methods under full and delayed feedback setting with delay $d \in \{0, 10, 100\}$ (Scenario~I). Results are averaged over 500 simulation runs. Standard errors are given in parentheses.}
\label{tab:sim_results_fd_scenario1}
\resizebox{\textwidth}{!}{%
\begin{tabular}{ll rrrrrrrr}
\toprule
& & \multicolumn{8}{c}{$\pi_1$} \\
\cmidrule(lr){3-10}
& Method & 0.1 & 0.2 & 0.3 & 0.4 & 0.5 & 0.6 & 0.7 & 0.8 \\
\midrule
\multicolumn{10}{l}{\textit{FDR}} \\
& LOND         & .021\tiny{(.005)} & .014\tiny{(.003)} & .008\tiny{(.003)} & .006\tiny{(.001)} & .003\tiny{(.001)} & .002\tiny{(.001)} & .002\tiny{(.000)} & .001\tiny{(.000)} \\
& LORD++       & .001\tiny{(.001)} & .004\tiny{(.001)} & .003\tiny{(.001)} & .005\tiny{(.001)} & .006\tiny{(.001)} & .006\tiny{(.001)} & .005\tiny{(.001)} & .004\tiny{(.000)} \\
& SAFFRON      & .044\tiny{(.007)} & .052\tiny{(.005)} & .062\tiny{(.004)} & .072\tiny{(.002)} & .074\tiny{(.001)} & .071\tiny{(.001)} & .065\tiny{(.001)} & .053\tiny{(.000)} \\
\cmidrule(lr){2-10}
& LF-FD $d=0$   & .047\tiny{(.008)} & .068\tiny{(.007)} & .066\tiny{(.004)} & .085\tiny{(.003)} & .090\tiny{(.001)} & .091\tiny{(.001)} & .091\tiny{(.001)} & .087\tiny{(.000)} \\
& LF-FD $d=10$  & .044\tiny{(.008)} & .057\tiny{(.007)} & .048\tiny{(.005)} & .046\tiny{(.003)} & .045\tiny{(.001)} & .039\tiny{(.001)} & .032\tiny{(.001)} & .022\tiny{(.000)} \\
& LF-FD $d=100$ & .041\tiny{(.007)} & .055\tiny{(.007)} & .046\tiny{(.005)} & .042\tiny{(.003)} & .040\tiny{(.001)} & .035\tiny{(.001)} & .028\tiny{(.001)} & .020\tiny{(.000)} \\
\cmidrule(lr){2-10}
& SF-FD $d=0$   & .048\tiny{(.007)} & .056\tiny{(.005)} & .068\tiny{(.004)} & .084\tiny{(.002)} & .088\tiny{(.001)} & .090\tiny{(.001)} & .090\tiny{(.001)} & .081\tiny{(.001)} \\
& SF-FD $d=10$  & .045\tiny{(.007)} & .055\tiny{(.005)} & .065\tiny{(.004)} & .076\tiny{(.002)} & .079\tiny{(.001)} & .078\tiny{(.001)} & .074\tiny{(.001)} & .064\tiny{(.001)} \\
& SF-FD $d=100$ & .044\tiny{(.007)} & .054\tiny{(.005)} & .063\tiny{(.004)} & .073\tiny{(.002)} & .075\tiny{(.001)} & .072\tiny{(.001)} & .066\tiny{(.001)} & .054\tiny{(.000)} \\
\midrule
\multicolumn{10}{l}{\textit{Power}} \\
& LOND         & .009\tiny{(.000)} & .012\tiny{(.000)} & .013\tiny{(.000)} & .015\tiny{(.000)} & .017\tiny{(.000)} & .020\tiny{(.000)} & .022\tiny{(.000)} & .025\tiny{(.000)} \\
& LORD++       & .004\tiny{(.000)} & .007\tiny{(.001)} & .010\tiny{(.001)} & .018\tiny{(.001)} & .029\tiny{(.001)} & .039\tiny{(.002)} & .052\tiny{(.002)} & .071\tiny{(.002)} \\
& SAFFRON      & .010\tiny{(.001)} & .031\tiny{(.002)} & .073\tiny{(.003)} & .166\tiny{(.004)} & .275\tiny{(.004)} & .382\tiny{(.003)} & .482\tiny{(.002)} & .588\tiny{(.002)} \\
\cmidrule(lr){2-10}
& LF-FD $d=0$   & .008\tiny{(.001)} & .026\tiny{(.002)} & .070\tiny{(.003)} & .182\tiny{(.004)} & .328\tiny{(.004)} & .476\tiny{(.002)} & .620\tiny{(.002)} & .775\tiny{(.001)} \\
& LF-FD $d=10$  & .007\tiny{(.001)} & .019\tiny{(.001)} & .036\tiny{(.002)} & .080\tiny{(.003)} & .146\tiny{(.003)} & .212\tiny{(.003)} & .269\tiny{(.003)} & .326\tiny{(.002)} \\
& LF-FD $d=100$ & .007\tiny{(.001)} & .017\tiny{(.001)} & .029\tiny{(.002)} & .059\tiny{(.002)} & .114\tiny{(.003)} & .167\tiny{(.003)} & .215\tiny{(.003)} & .266\tiny{(.003)} \\
\cmidrule(lr){2-10}
& SF-FD $d=0$   & .011\tiny{(.001)} & .035\tiny{(.002)} & .087\tiny{(.003)} & \textbf{.202}\tiny{(.004)} & \textbf{.331}\tiny{(.003)} & \textbf{.462}\tiny{(.002)} & \textbf{.596}\tiny{(.002)} & \textbf{.741}\tiny{(.002)} \\
& SF-FD $d=10$  & .011\tiny{(.001)} & .033\tiny{(.002)} & .079\tiny{(.003)} & .181\tiny{(.004)} & .296\tiny{(.004)} & .413\tiny{(.002)} & .528\tiny{(.002)} & .654\tiny{(.002)} \\
& SF-FD $d=100$ & .010\tiny{(.001)} & \textbf{.032}\tiny{(.002)} & .075\tiny{(.003)} & .169\tiny{(.004)} & .279\tiny{(.004)} & .387\tiny{(.003)} & .489\tiny{(.002)} & .599\tiny{(.002)} \\
\bottomrule
\end{tabular}%
}
\end{table}

\begin{table}[!htbp]
\centering
\caption{FDR and Power of competing methods in the classification setting. Results are averaged over 500 simulation runs. Standard errors are given in parentheses.}
\label{tab:sim_results_classification}
\resizebox{\textwidth}{!}{%
\begin{tabular}{l rrrrrrrr}
\toprule
& \multicolumn{8}{c}{$\pi_1$} \\
\cmidrule(lr){2-9}
Method & 0.1 & 0.2 & 0.3 & 0.4 & 0.5 & 0.6 & 0.7 & 0.8 \\
\midrule
\multicolumn{9}{l}{\textit{FDR}} \\
LOND    & .034\tiny{(.007)} & .023\tiny{(.004)} & .027\tiny{(.005)} & .019\tiny{(.004)} & .020\tiny{(.005)} & .015\tiny{(.004)} & .011\tiny{(.003)} & .012\tiny{(.004)} \\
LORD++  & .011\tiny{(.004)} & .008\tiny{(.002)} & .005\tiny{(.001)} & .007\tiny{(.002)} & .005\tiny{(.002)} & .002\tiny{(.001)} & .004\tiny{(.002)} & .004\tiny{(.002)} \\
SAFFRON & .090\tiny{(.007)} & .138\tiny{(.006)} & .137\tiny{(.004)} & .148\tiny{(.004)} & .147\tiny{(.003)} & .142\tiny{(.003)} & .126\tiny{(.004)} & .084\tiny{(.003)} \\
LF      & .132\tiny{(.011)} & .151\tiny{(.008)} & .152\tiny{(.006)} & .168\tiny{(.004)} & .169\tiny{(.004)} & .172\tiny{(.003)} & .175\tiny{(.004)} & .147\tiny{(.003)} \\
SF      & .117\tiny{(.009)} & .153\tiny{(.007)} & .151\tiny{(.005)} & .165\tiny{(.004)} & .162\tiny{(.003)} & .154\tiny{(.002)} & .137\tiny{(.004)} & .094\tiny{(.003)} \\
LFS     & .111\tiny{(.012)} & .083\tiny{(.009)} & .052\tiny{(.007)} & .044\tiny{(.005)} & .041\tiny{(.005)} & .033\tiny{(.004)} & .036\tiny{(.006)} & .031\tiny{(.004)} \\
SFS     & .088\tiny{(.010)} & .071\tiny{(.008)} & .038\tiny{(.005)} & .031\tiny{(.004)} & .033\tiny{(.004)} & .024\tiny{(.003)} & .032\tiny{(.005)} & .024\tiny{(.004)} \\
\midrule
\multicolumn{9}{l}{\textit{Power}} \\
LOND    & .014\tiny{(.001)} & .018\tiny{(.001)} & .017\tiny{(.001)} & .013\tiny{(.001)} & .012\tiny{(.001)} & .007\tiny{(.001)} & .003\tiny{(.000)} & .001\tiny{(.000)} \\
LORD++  & .007\tiny{(.002)} & .020\tiny{(.003)} & .035\tiny{(.005)} & .028\tiny{(.004)} & .023\tiny{(.004)} & .016\tiny{(.003)} & .006\tiny{(.002)} & .000\tiny{(.000)} \\
SAFFRON & .077\tiny{(.006)} & .340\tiny{(.013)} & .533\tiny{(.015)} & .660\tiny{(.015)} & .735\tiny{(.015)} & .779\tiny{(.015)} & .794\tiny{(.016)} & .729\tiny{(.018)} \\
LF      & .049\tiny{(.005)} & .294\tiny{(.012)} & .510\tiny{(.015)} & .719\tiny{(.013)} & .792\tiny{(.014)} & .861\tiny{(.012)} & .881\tiny{(.013)} & .860\tiny{(.014)} \\
SF      & .085\tiny{(.007)} & .372\tiny{(.013)} & .560\tiny{(.014)} & .717\tiny{(.013)} & .784\tiny{(.013)} & .847\tiny{(.012)} & .856\tiny{(.013)} & .820\tiny{(.015)} \\
LFS     & .011\tiny{(.000)} & .011\tiny{(.000)} & .011\tiny{(.000)} & .011\tiny{(.000)} & .011\tiny{(.000)} & .010\tiny{(.000)} & .011\tiny{(.000)} & .009\tiny{(.000)} \\
SFS     & .013\tiny{(.001)} & .014\tiny{(.000)} & .013\tiny{(.000)} & .012\tiny{(.000)} & .012\tiny{(.000)} & .011\tiny{(.000)} & .011\tiny{(.000)} & .009\tiny{(.000)} \\
\bottomrule
\end{tabular}%
}
\end{table}

\begin{table}[!htbp]
\centering
\caption{Results for Scenario {IV} (sine pattern shifts): Performance at the stopping time ($T=1000$). The target level $\alpha=0.05$.}
\label{tab:stopping_time_results}
{\small
\begin{tabular}{llcc}
\toprule
Framework & Method & FDR & Power \\
\midrule
\multirow{5}{*}{Adaptive GAIF}
& SAFFRON & .038\tiny{(.004)} & .280\tiny{(.021)} \\
& Ran-SF  & .037\tiny{(.003)} & .289\tiny{(.021)} \\
& Opt-SF  & \textbf{.046}\tiny{(.007)} & \textbf{.421}\tiny{(.023)} \\
& Ran-SFS & .005\tiny{(.003)} & .007\tiny{(.000)} \\
& Opt-SFS & .019\tiny{(.008)} & .008\tiny{(.001)} \\
\midrule
\multirow{5}{*}{GAIF}
& LORD++  & .007\tiny{(.005)} & .002\tiny{(.001)} \\
& Ran-LF  & .048\tiny{(.007)} & .297\tiny{(.021)} \\
& Opt-LF  & \textbf{.047}\tiny{(.009)} & \textbf{.395}\tiny{(.024)} \\
& Ran-LFS & .015\tiny{(.008)} & .006\tiny{(.000)} \\
& Opt-LFS & .021\tiny{(.009)} & .006\tiny{(.000)} \\
\bottomrule
\end{tabular}
}
\end{table}

\begin{table}[!ht]
\centering
\setlength{\heavyrulewidth}{0.5pt}
\setlength{\lightrulewidth}{0.3pt}
\caption{\small $\operatorname{FDR}(T)$ and $\operatorname{Power}(T)$ for the safe variants (\text{Opt-SFS} and \text{Opt-LFS}) across four datasets (Candidate, Diabetes, Income, Airfoil). The target FDR level is $\alpha = 0.3$.}
\label{tab:appendix_safe_fdr_power}
\resizebox{\textwidth}{!}{
{\tiny
\begin{tabular}{lcccccccc}
\toprule
\textbf{Method} 
& \multicolumn{2}{c}{\textbf{Task 1}} 
& \multicolumn{2}{c}{\textbf{Task 2}} 
& \multicolumn{2}{c}{\textbf{Task 3}} 
& \multicolumn{2}{c}{\textbf{Task 4}} \\
\cmidrule(lr){2-3} \cmidrule(lr){4-5} \cmidrule(lr){6-7} \cmidrule(lr){8-9}
& FDR & Power & FDR & Power & FDR & Power & FDR & Power \\
\midrule
Opt-SFS & .071\tiny{(.014)} & .002\tiny{(.000)} & .090\tiny{(.017)} & .002\tiny{(.000)} & .077\tiny{(.014)} & .008\tiny{(.000)} & .002\tiny{(.001)} & .007\tiny{(.003)} \\
Opt-LFS & .070\tiny{(.014)} & .002\tiny{(.000)} & .088\tiny{(.016)} & .002\tiny{(.000)} & .105\tiny{(.017)} & .007\tiny{(.000)} & .001\tiny{(.001)} & .005\tiny{(.001)} \\
\bottomrule
\end{tabular}
}
}
\end{table}

{\color{black}\subsection{Results on synthetic data for a regression task}\label{add:experement-reg}}
The corresponding data generation process for {\color{black} the regression task} is detailed below: 
\begin{itemize}
    \item \textbf{Scenario {V} (Regression example)}: $Y=-0.5X_1^2+\exp{X_2}+(X_3+X_4)^2+\varepsilon,$ with $\X\sim \mathcal{N}_4({\bf 0},\mathbf{I}_4)$ and $\varepsilon\sim \mathcal{N}(0,2)$. The target region is $\mathcal{A}=[c,\infty)$, where $c$ is the $1-\pi_1$ quantile of $Y$. 
\end{itemize}

In terms of the non-conformity score function, denote $W_t=\widehat{\mu}_t(\X_t)$, in classification settings, we set $V(W_t)=1-W_t$.  In regression settings, if $\mathcal{A}=[b,+\infty)$, we can use $V(W_j)=b-W_j$. If $\mathcal{A}=(-\infty,a]\cup[b,+\infty)$, then we can choose $V(W_t)=\max\{W_t-a,b-W_t\}$. 

 The results for the regression example (Scenario {V}) using a fixed training algorithm (random forest) are shown in Figure \ref{fig:reg_stop_B_vary_prop}. The performance trends are similar to those observed in the classification case (Scenario {IV}).

  \begin{figure}[H]
    \centering
    \includegraphics[width=0.8\textwidth]{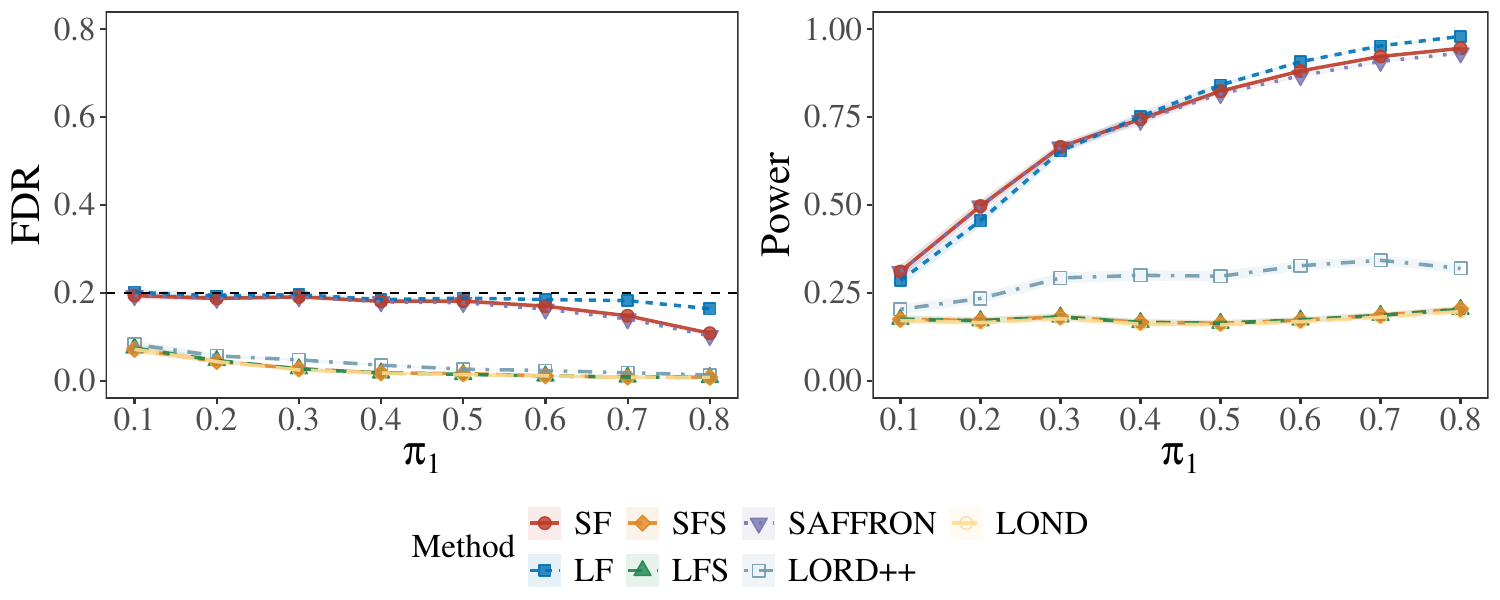}
    \caption{\small Results for Scenario {V}: values of $\FDR(T)$ and $\operatorname{Power}(T)$ at stopping time $T$ across different non-null proportions $\pi_1$. The black dashed line denotes the FDR level $\alpha=0.2$.}
    \label{fig:reg_stop_B_vary_prop}
  \end{figure}

{\color{black}\subsection{Additional experiments for different training algorithms}\label{add:experement-training-alg}} 

  The results for Scenarios {IV} and {\color{black}V} under different training algorithms—RF, SVM, and NN—with varying initial calibration sizes are presented in Figure \ref{fig:stop_vary_ncal_scenarioA}-\ref{fig:stop_vary_ncal_scenarioB}. Thanks to the online updating of the calibration dataset, even a small initial calibration size does not significantly impact performance. While the choice of predictive model $\widehat{\mu}$ does affect power, all methods maintain valid FDR control. Notably, our SF and LF methods consistently outperform the baselines across all score functions, benefiting from the distribution-free and model-agnostic nature of online conformal \(p\)-values. The variation in performance across different algorithms further underscores the importance of careful score selection in practice.

  \begin{figure}[htbp!]
    \centering
    \includegraphics[width=0.8\textwidth]{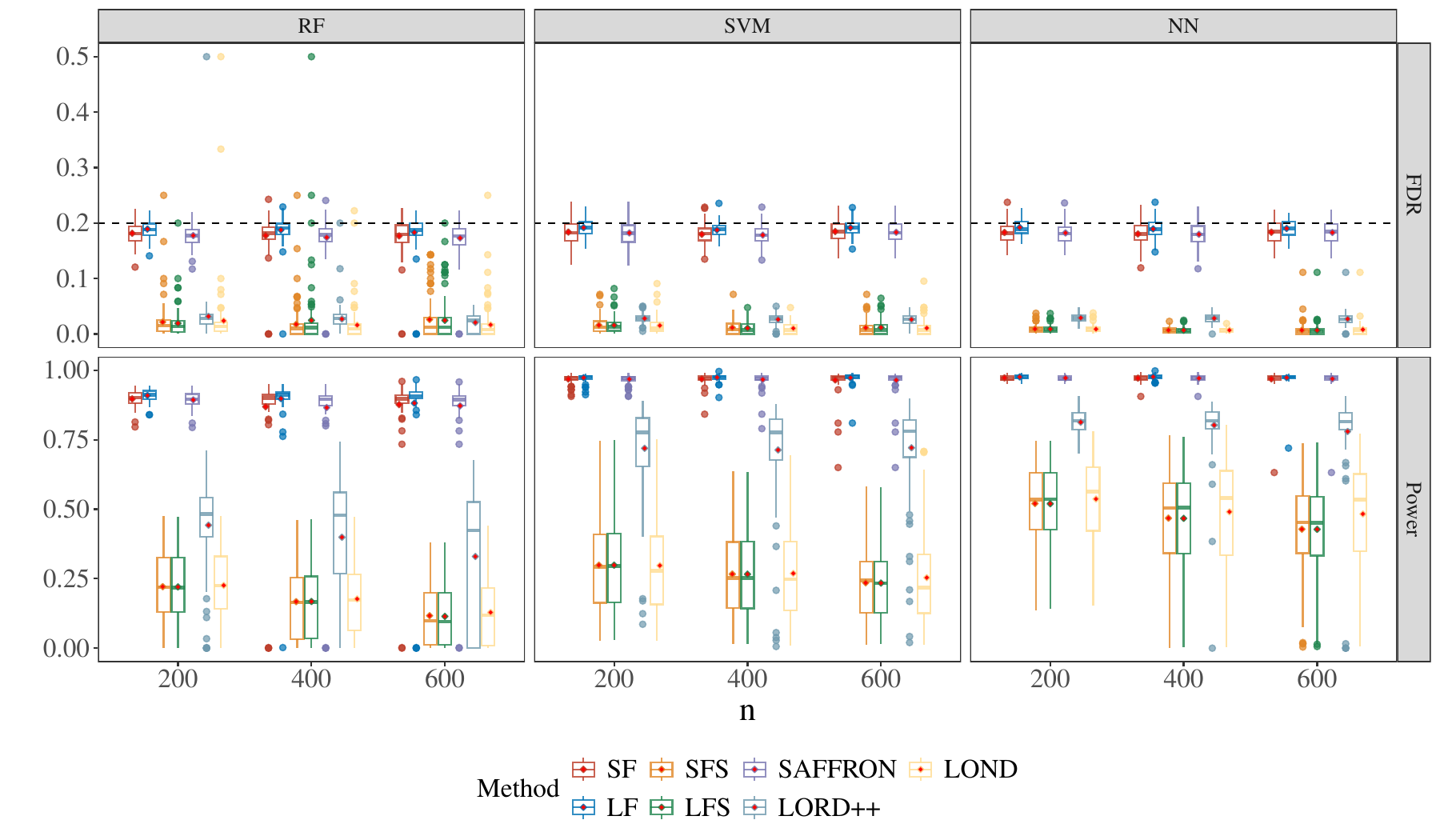}
    \caption{\small Results for Scenario {IV}: $\FDR(T)$ and $\operatorname{Power}(T)$ vs. initial calibration size $n$ ($\pi_1 = 0.5$, $\alpha = 0.2$).}
    \label{fig:stop_vary_ncal_scenarioA}
  \end{figure}

  \begin{figure}[htbp!]
    \centering
    \includegraphics[width=0.8\textwidth]{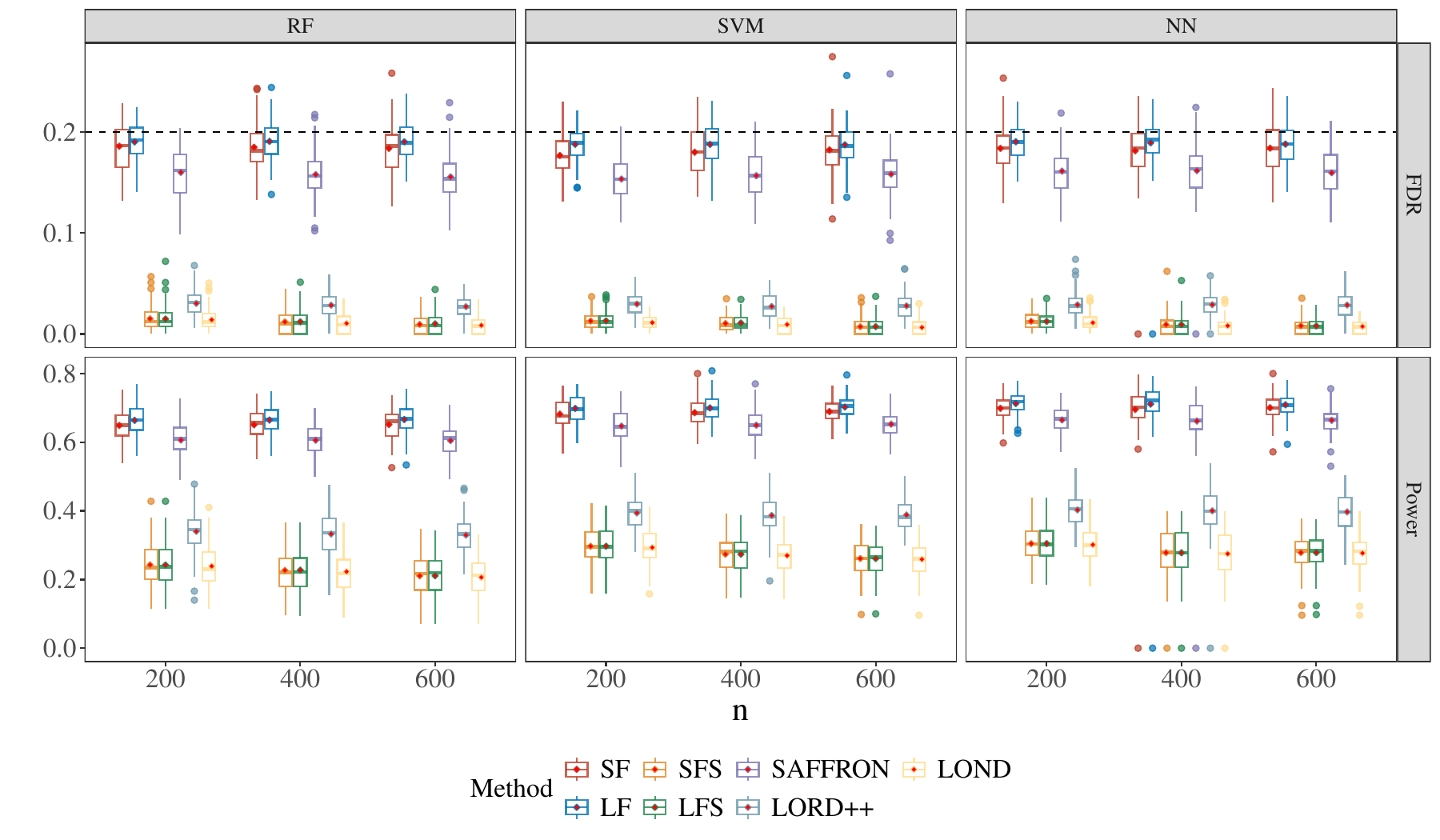}
    \caption{\small Results for Scenario {V}: $\FDR(T)$ and $\operatorname{Power}(T)$ vs. initial calibration size $n$ ($\pi_1 = 0.5$, $\alpha = 0.2$).}
    \label{fig:stop_vary_ncal_scenarioB}
  \end{figure}



{\color{black}\subsection{Additional experiments on score selection}\label{add:experement-score-selection}} 

 The results with score selection for Scenarios {IV} and {V} are shown below Figure \ref{fig:cla_stop_A_vary_prop_opt}-Figure \ref{fig:reg_stop_B_vary_prop_opt}. The target level is $\alpha=0.05$. In both settings, the performance gap between the Opt methods and their randomly selected counterparts is also pronounced.

   \begin{figure}[htbp!]
		\centering
		\includegraphics[width=0.7\textwidth]{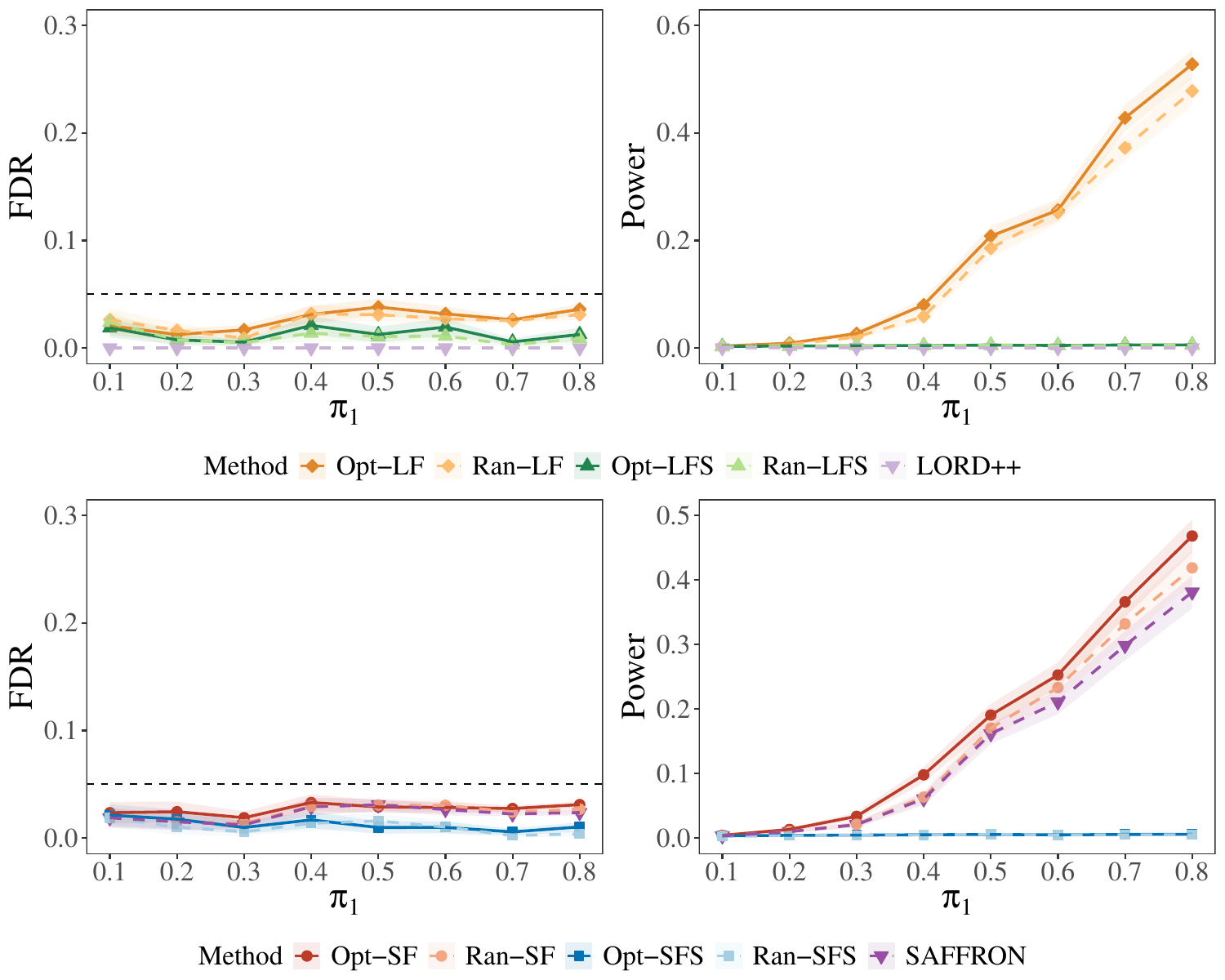}
		\caption{\small Results for Scenario {IV}: the values of $\FDR(T)$  and $\operatorname{Power}(T)$ at stopping time $T$ across different non-null proportions $\pi_1$. The black dashed lines denote the FDR level $\alpha=0.05$.}
  \label{fig:cla_stop_A_vary_prop_opt}
	\end{figure}

  \begin{figure}[htbp!]
    \centering
    \includegraphics[width=0.7\textwidth]{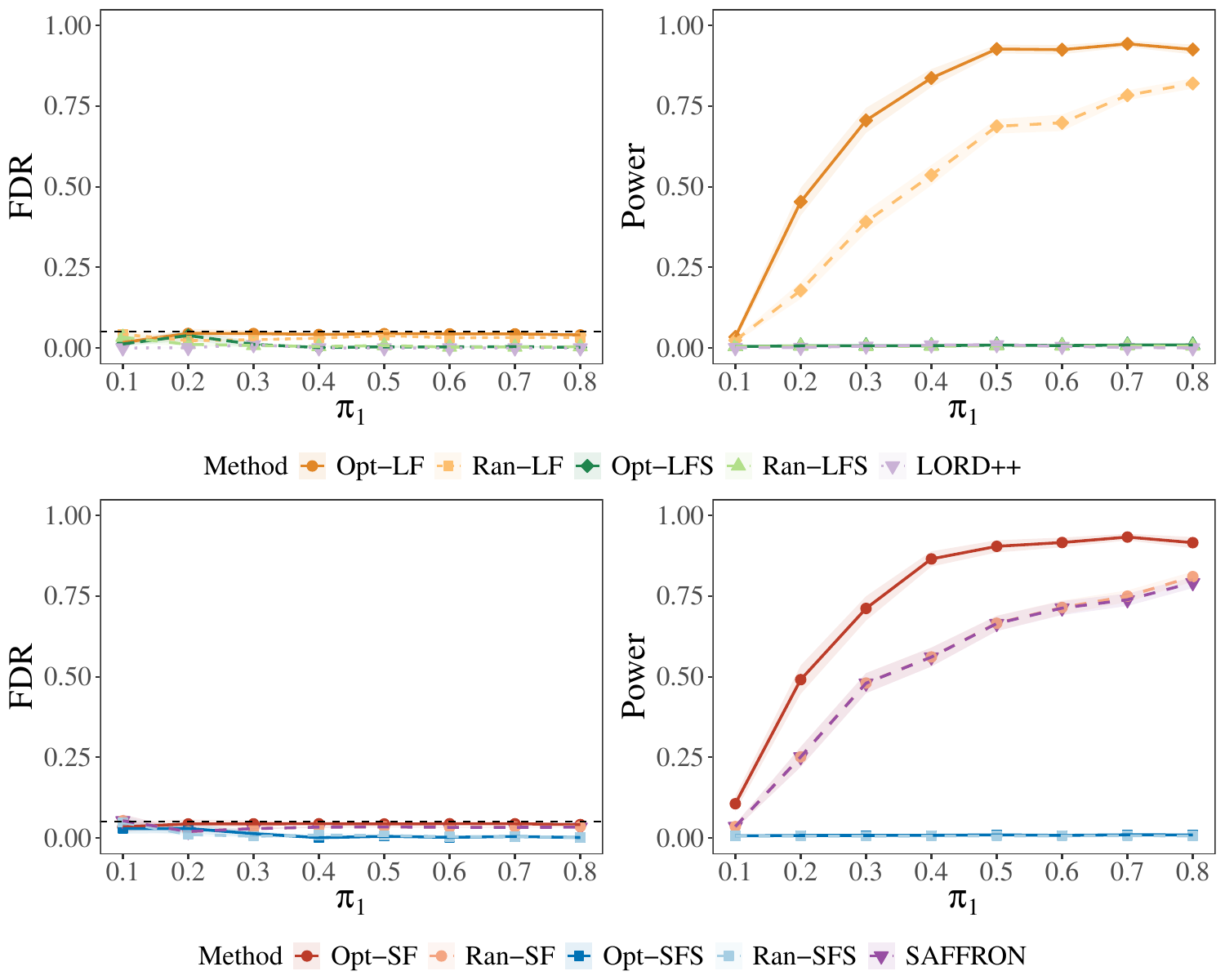}
    \caption{\small  Results for Scenario {V}: the values of $\FDR(T)$  and $\operatorname{Power}(T)$ at stopping time $T$ across different non-null proportions $\pi_1$. The black dashed lines denote the FDR level $\alpha=0.05$.}
    \label{fig:reg_stop_B_vary_prop_opt}
  \end{figure}


{\color{black}\subsection{Additional experiments results on empirical mFDR}\label{add:experement-mFDR}
 }

   To illustrate the similarity between mFDR and FDR, we present results under various settings below in Figure \ref{fig:GAIF_dep_mFDR}-Figure \ref{fig:cla-reg-mFDR}. We estimate mFDR by computing the ratio of the average number of false discoveries and the average total number of discoveries. Empirical mFDR closely tracks empirical FDR, and both are well controlled by our proposed methods.

  \begin{figure}[htbp!]
		\centering
		\includegraphics[width=0.8\textwidth]{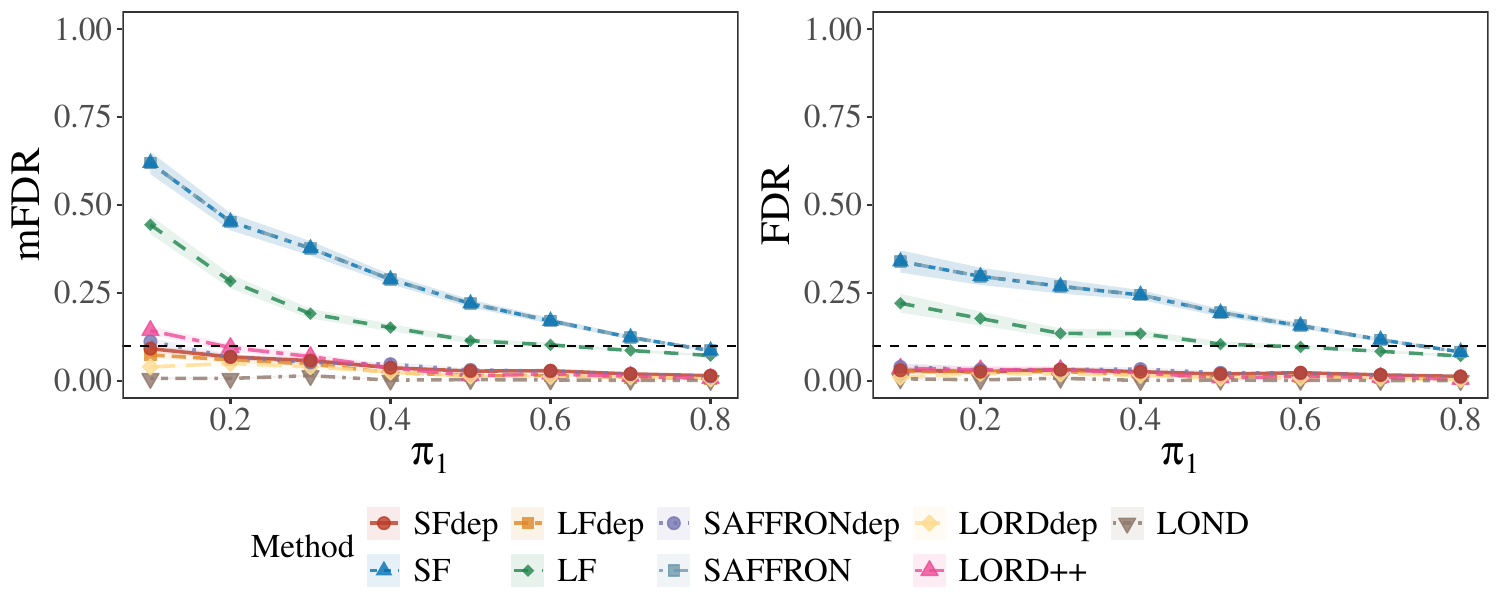}
		\caption{\small Results for Scenario III (local dependence): Line charts of mFDR and FDR at stopping time with varying non-null proportion $\pi_1$ from $0.1$ to $0.8$.  The black dashed lines denote the target FDR level $\alpha=0.1$.}
  \label{fig:GAIF_dep_mFDR}
	\end{figure}

 \begin{figure}[htbp!]
		\centering
		\includegraphics[width=0.8\textwidth]{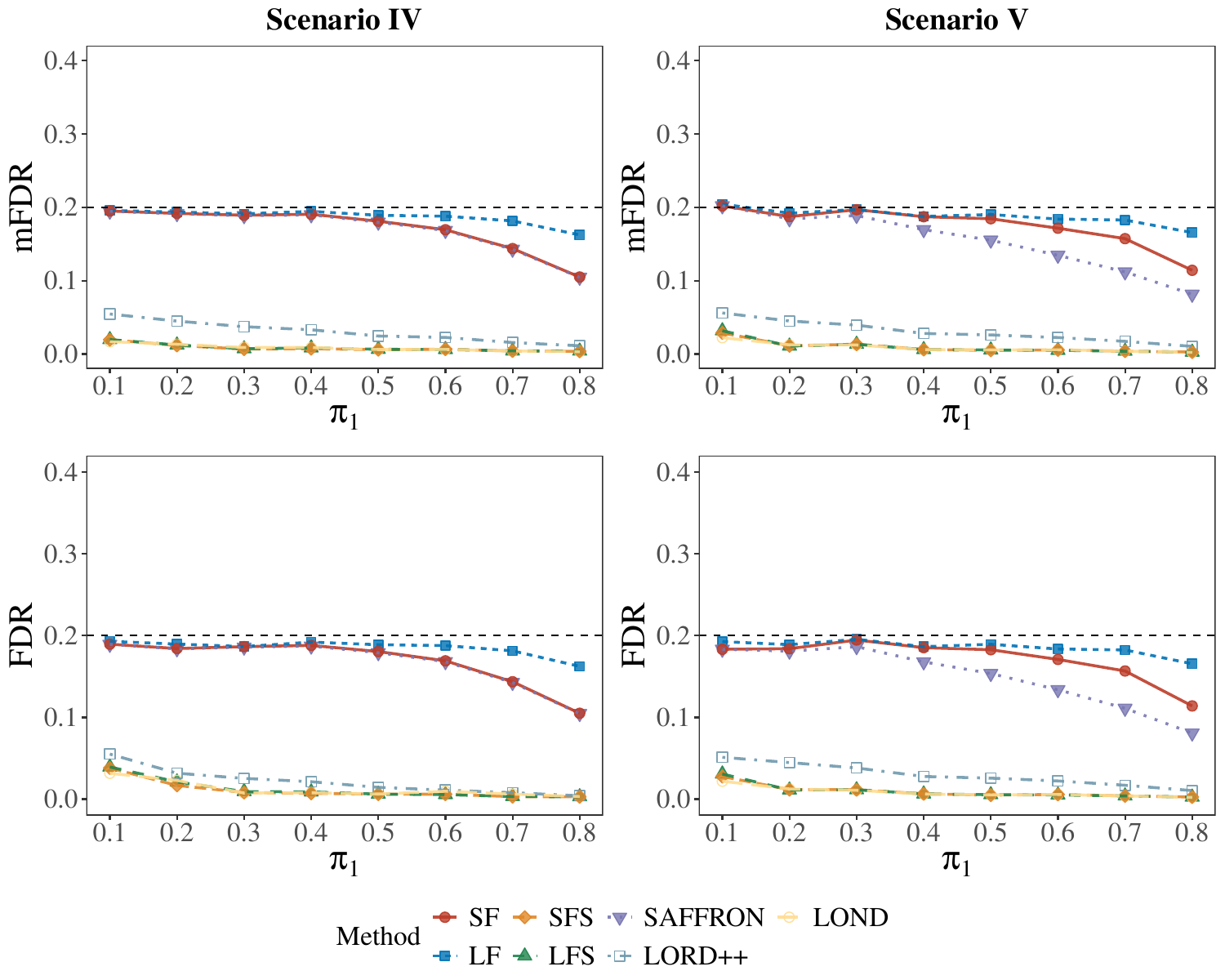}
		\caption{\small Results for Scenario {\color{black}IV} and Scenario {\color{black}V} :  Line charts of mFDR and FDR at stopping time with varying non-null proportion $\pi_1$ from $0.1$ to $0.8$ after $500$ replications; The black dashed lines denote the target FDR level $\alpha=0.2$.}
  \label{fig:cla-reg-mFDR}
	\end{figure}

{\color{black}


\subsection{Additional experiments on auxiliary calibration for score selection}\label{appen_subsec:simu_auxi_modelsel}

In our EWMA-based model-selection procedure, we construct auxiliary non-null p-values in \eqref{eq:aux_nonnull_p_value} based on $\gC_{0t}\cup\{t\}$, where the current index $t$ is included. The main purpose of this design is to preserve \emph{finite-sample validity under data-adaptive score selection}. In particular, including $t$ ensures the permutation symmetry of $\gC_{0t}\cup\{t\}$ required by our validity argument.

From a power perspective, one may worry that, when $t$ is non-null and $t$ is small, including the current point could slightly perturb the auxiliary calibration distribution. In practice, however, this effect is limited in our setting. Our procedure starts from a historical calibration dataset $\gC$ of size $n$, so even at early times the inclusion of a single point has only a minor impact, since the effective null calibration size is already $n+t$.

To further reduce the possible adverse effect when \(t\) is non-null, we also consider a simple robustified version based on a truncated calibration set:
\begin{equation}
\tilde{p}_{k,j}
=
\frac{\sum_{s \in \gC_{t}^{\mathrm{trun},k}} \mathbb{I}\{V(\mathbf{X}_s; k) \leq V(\mathbf{X}_j; k)\}}
{|\gC_{t}^{\mathrm{trun},k}|},
\quad j \in \gC_{1t},
\label{eq:tru-cp}
\end{equation}
where
\[
\gC_{t}^{\mathrm{trun},k}
=
\big(\gC_{0t}\cup\{t\}\big)
\setminus
\left\{
\arg\min_{i \in \gC_{0t}\cup\{t\}} V(\mathbf{X}_i;k)
\right\}.
\]
This operation remains permutation-invariant on \( \gC_{0t}\cup\{t\} \), so the symmetry needed for validity is preserved. Intuitively, when \(t\) is non-null, its score \(V(\mathbf{X}_t;k)\) is more likely to be unusually small and hence removed by the truncation, reducing its impact.

To investigate the practical effect of including $t$, we compare the following four variants:
\begin{itemize}
    \item \textbf{Opt-SFS:} Our proposed feedback-enhanced approach using EWMA of past non-null auxiliary $p$-values based on calibration set $\gC_{0t}\cup\{t\}$ 
    \item \textbf{OptEx-SFS:} Similar to Opt-SFS but excluding $t$ from the calibration data, i.e. constructing  auxiliary $p$-values based on $\gC_{0t}$.
    \item \textbf{OptTr-SFS:} Optimized selection using the truncated calibration set $\gC_{t}^{\mathrm{trun},k}$ defined in Eq.~(\ref{eq:tru-cp}).
    \item \textbf{Ran-SFS:} A baseline that randomly selects a model at each time step to quantify the gain from feedback-driven selection.
\end{itemize}

 We generated data streams of length $T=1000$ characterized by smoothly varying distribution shifts according to Scenario IV in Section \ref{subsec:simu-opt-OCTF}. Conformal $p$-values were computed online using the above proposed EWMA-based selection strategies, and rejections were determined via standard SFS procedures in the main text. For each method, we tracked the empirical online FDR and empirical power across all time points across 200 replications.

\begin{figure}[h!]
    \centering
    \includegraphics[width=0.8\textwidth]{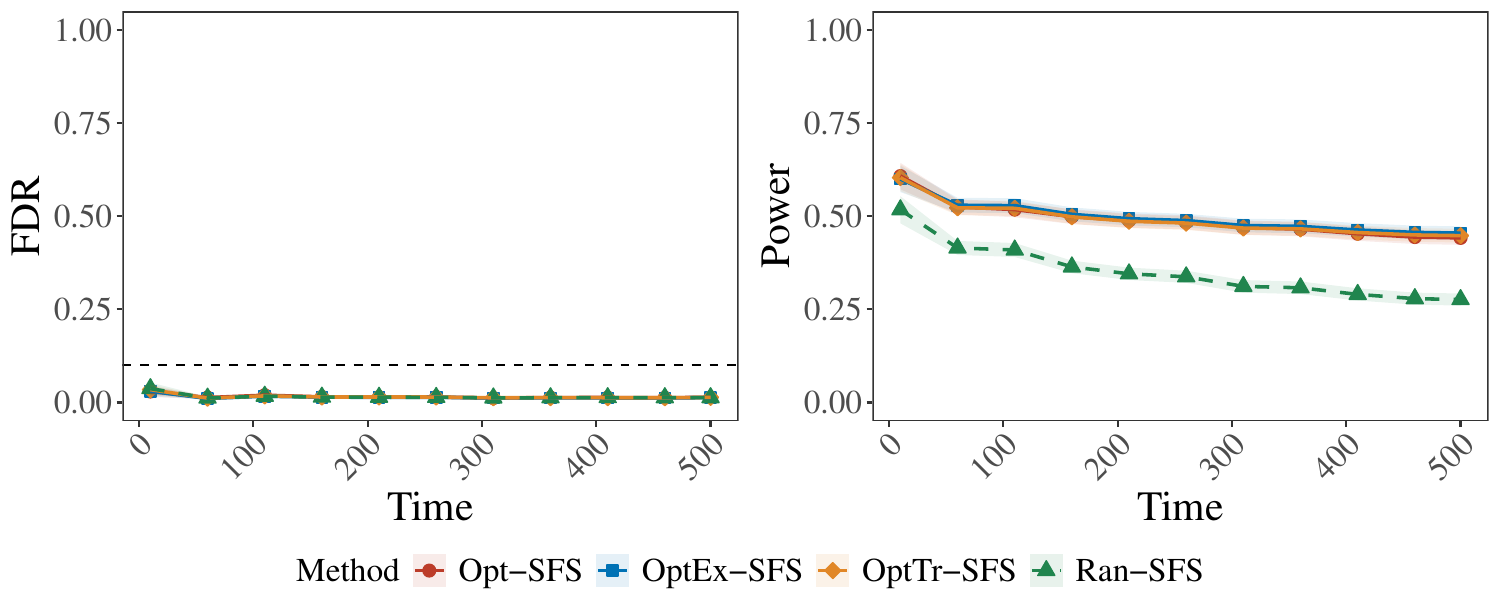}
    \caption{\small Online FDR and power for various Opt-SFS variants versus the Ran-SFS baseline over time.}
    \label{fig:cla-Opt-SFS}
\end{figure}

As shown in Figure~\ref{fig:cla-Opt-SFS}, all three optimized methods (Opt-SFS, OptEx-SFS, and OptTr-SFS) maintain empirical FDR control below the nominal level ($\alpha=0.1$) while consistently achieving higher power than the random baseline. 

A comparison between Opt-SFS and OptEx-SFS reveals that at early time points ($t$ small), differences in power among these strategies are relatively minor. As more observations accumulate, the differences among the three Opt methods diminish. 
These results empirically verify that the inclusion of $t$ has negligible influence on the power performance.
}

{\color{black}\subsection{Comparison with alternative finite-sample-valid baselines}
\label{app:safe-baselines}


In this section, we compare LFS and SFS against two natural alternatives that are also
finite-sample valid in online conformal testing.

\paragraph{
Fixed-calibration conformal $p$-values and LOND.} 
The first alternative keeps the calibration set fixed throughout the entire test stream, exactly as in standard split-conformal
inference: for each test point $t$ with conformity score $V_t$,
\begin{equation}
    p_t = \frac{ \sum_{i \in \gC_{0}} \I\{V_{i} < V_{t}\}+\xi_t[1+\sum_{i\in\gC_0}\I\{V_i=V_t\}]}{1 + |\gC_{0}|},
    \label{eq:classic-conformal-p}
\end{equation}  
where $\gC_0$ is a hold-out calibration dataset. Under the standard exchangeability assumption,  the null conformal \(p\)-values
constructed from a shared fixed calibration set \(\gC_0\) are PRDS, 
which is sufficient for the LOND rule
$\alpha_t^{\mathrm{LOND}} := \alpha\gamma_t\bigl(R^{\mathrm{LOND}}{(t-1)}+1\bigr)$
to control the FDR \citep{Zrnic2021asynchronous}. 

\paragraph{Online conformal $e$-values and e-LOND.}
The second alternative retains our online calibration update: each confirmed null is added to the calibration pool as soon as it is observed. It then uses $e$-values in place of $p$-values and applies e-LOND rule of \citep{xu2024online}. Because e-LOND controls the FDR under arbitrary, possibly unknown dependence among $e$-values, it accommodates the dependence induced by an evolving shared calibration set. 
We use
the rank-based conformal $e$-value
\begin{equation}
e_t \;=\; \frac{n_t+1}{k_t}\cdot \mathbf 1\bigl\{\mathrm{rank}_t \le k_t\bigr\},
\qquad
k_t = \max\bigl\{1,\ \lfloor \rho\,(n_t+1)\rfloor\bigr\},
\label{eq:online-conformal-e}
\end{equation}
where $n_t$ is the calibration set size immediately before test point $t$ is
processed, $\mathrm{rank}_t$ is the (randomized, tie-broken) rank of the test
score among the $n_t$ calibration scores and itself, and $\rho\in(0,1)$ is a
fixed proportion (we use $\rho=0.05$). This is a generalization of the
single-test conformal $e$-value of \citet{vovk2021values} and the
construction underlying \citet{bashari2023derandomized}: since $k_t$ depends
only on the calibration history $n_t$ and not on the test score itself, it is
predictable, and exchangeability of the test score with the calibration pool
under the null gives $\mathbb E[e_t]=1$. Discoveries are then made using
\begin{equation}
\alpha_t^{\mathrm{e\text{-}LOND}} := \alpha\gamma_t\bigl(R^{\mathrm{e\text{-}LOND}}{(t-1)}+1\bigr),
\qquad
\text{reject } \mathbb{H}_{0t} \iff e_t \ge 1/\alpha_t^{\mathrm{e\text{-}LOND}}.
\label{eq:e-lond-rule}
\end{equation}

\paragraph{Experimental setup.}
We employ the binary classification setting of Scenario~IV with $\boldsymbol{\mu}_1=(5,0,0,0)^{\top}$ and $\boldsymbol{\mu}_2=(0,0,-4,-4)^{\top}$, sweeping the non-null proportion $\pi_1 \in \{0.1,0.2,\dots,0.8\}$. FDR and power are averaged over $500$ independent replications at nominal level $\alpha=0.1$, with an initial null calibration set of size $n_{\text{cal}}=50$. The training algorithm is support vector machine.
We compare four procedures:
LOND applied to classic conformal $p$-values~\eqref{eq:classic-conformal-p},
e-LOND applied to online conformal $e$-values~\eqref{eq:online-conformal-e}
via~\eqref{eq:e-lond-rule}, and our proposed LFS and SFS applied to online
conformal $p$-values.

\paragraph{Results.}
Figure~\ref{fig:safe-baselines} reports empirical FDR and power across different $\pi_1$
for all four methods. LOND with a fixed calibration set and e-LOND both control
the empirical FDR at the nominal level across all values of $\pi_1$, confirming
their finite-sample validity. While all four methods maintain valid empirical error
control, LFS and SFS consistently achieve power gains relative to LOND. In contrast, e-LOND remains highly conservative. This indicates that incorporating feedback through the proposed
safe strategy in OCTF can improve detection performance while preserving finite-sample
guarantees.

\begin{figure}[t]
\centering
\includegraphics[width=0.8\textwidth]{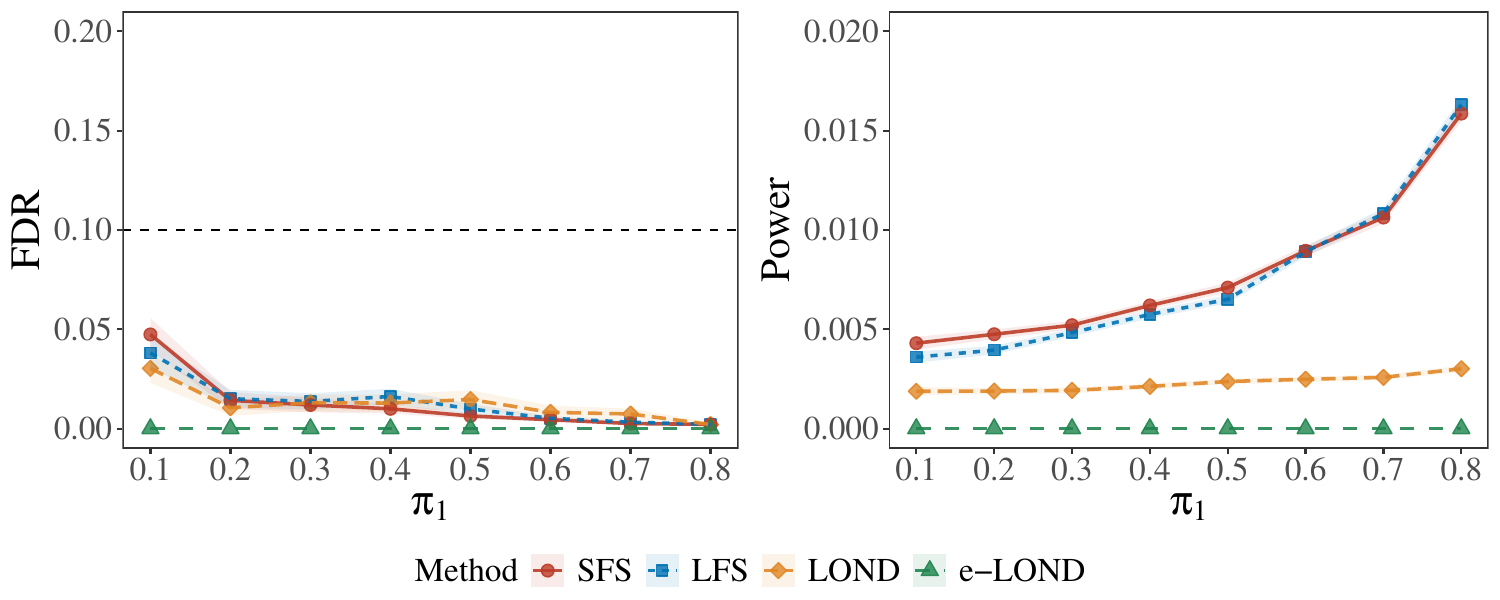}
\caption{FDR (left) and power (right) across different non-null
proportion $\pi_1$, for LOND with classic (fixed) conformal $p$-values,
e-LOND with online conformal $e$-values, and our proposed LFS/SFS with
online conformal $p$-values. The dashed line marks the nominal level
$\alpha=0.1$.}
\label{fig:safe-baselines}
\end{figure}
}

\section{Technical Details} \label{appen:proofs} 

{\color{black}

\subsection{Auxiliary lemmas and proof of lemmas}

Our proof of Theorem~\ref{the:FDR_GAIF_ind} and Theorem~\ref{the:FDR_Ada_GAIF_ind} relies on the following lemmas, which are modified versions of Lemma~1 in \cite{ramdas2017online} and Lemma~1 in \cite{ramdas2018saffron}, respectively. The modifications arise because, in our setting, the feedback information {\color{black}$\{\theta_j\}_{j\in\mathcal{I}_t}$} is available at time~$t$, and the test levels $\alpha_t$ in GAIF and Adaptive GAIF procedures depend on both past rejections $\delta_j$ and feedback {\color{black}$\{\theta_j\}_{j\in\mathcal{I}_t}$}.

For simplicity, denote $\delta_{1:t}=(\delta_1,\dots,\delta_t)$. Given a sequence $p_1,p_2,\dots$ of independent $p$-values (i.e., the null $p$-values are independent of all other $p$-values), we define a filtration via the sigma-field of decisions and feedback $\mathcal{F}_{t-1}=\sigma( \delta_{1:t-1};\{\theta_j\}_{j\in\mathcal{I}_t})$, where $\delta_t=\mathbb{I}\{p_t\leq f_t(\delta_{1:t-1};\{\theta_j\}_{j\in\mathcal{I}_t})\}$ for some coordinate-wise non-decreasing function $f_t:\{0,1\}^{t-1+|\mathcal{I}_t|}\to \mathbb{R}$. With this set-up, we have the following guarantees in Lemma \ref{lem:super-uniformity}:


\begin{lemma}\label{lem:super-uniformity}
   Let $g:\{0,1\}^{T}\to\mathbb{R}$ be any coordinate-wise non-decreasing function such that $g(\boldsymbol{x})>0$ for any vector $\boldsymbol{x}\neq (0,\dots,0)$.  Then for any index $t\leq T$ such that $H_t\in\mathcal{H}_0$, we have
    \[\E\left[\frac{\mathbb{I}\{p_t\leq f_t(\delta_{1:t-1};\{\theta_j\}_{j\in\mathcal{I}_t})\}}{g(\delta_{1:T})}\mid \mathcal{F}_{t-1}\right]\leq \E\left[\frac{f_t(\delta_{1:t-1};\{\theta_j\}_{j\in\mathcal{I}_t})}{g(\delta_{1:T})}\mid \gF_{t-1}\right].\]
\end{lemma}

\begin{proof}
    {\color{black}
Denote ${p}_{1:T}=(p_1,\dots,p_T)$, and $\tilde{{p}}_{1:T}^{-t}=(\tilde{p}_1,\dots,\tilde{p}_T)$, where 
\begin{equation}
    \tilde{p}_i=\begin{cases}
        0\quad \text{if}\; i=t \\
        p_i\quad \text{if}\; i\neq t,
    \end{cases} \nonumber 
\end{equation}
i.e., the $t$-th component of $\tilde{{p}}_{1:T}^{-t}$ is zero and the other components of $\tilde{{p}}_{1:T}^{-t}$ equals to that of ${p}_{1:T}$. For all $i$, define $\tilde{\delta}_i=\mathbb{I}\{\tilde{{p}}_i\leq f_i(\tilde{\delta}_{1:i-1},\{\theta_j\}_{j\in\mathcal{I}_i})\}$. And let the decision vectors using ${p}_{1:T}$ and $\tilde{{p}}_{1:T}^{-t}$ be ${\delta}_{1:T}=(\delta_1,\dots,\delta_T)$ and $\tilde{{\delta}}_{1:T}^{-t}=(\tilde{\delta}_1,\dots,\tilde{\delta}_T)$. By construction, we have
\[\begin{cases}
    \tilde{\delta}_i=\delta_i \quad \text{for all}\; i<t\\
    \tilde{\delta}_i\geq\delta_i \quad \text{for all}\; i\geq t,\\
\end{cases}\]
from which we obtain $f_i(\delta_{1:i-1};\{\theta_j\}_{j\in\mathcal{I}_i})=f_i(\tilde{\delta}_{1:i-1};\{\theta_j\}_{j\in\mathcal{I}_i})$ for all $i\leq t$. Also, by noting that $\tilde{\delta}_t=1$ by construction and by definition of $g(\cdot)$, we have $g(\tilde{{\delta}}_{1:T}^{-t})>0$. Hence, on the event $\{p_t\leq f_t(\delta_{1:t-1};\{\theta_j\}_{j\in\mathcal{I}_t})\}$, we have $\delta_t=\tilde{\delta}_t=1$ and hence also ${\delta}_{1:T}=\tilde{{\delta}}_{1:T}^{-t}$, which allows us to conclude that
\[\frac{\mathbb{I}\{p_t\leq f_t(\delta_{1:t-1};\{\theta_j\}_{j\in\mathcal{I}_t})\}}{g({\delta}_{1:T})}=\frac{\mathbb{I}\{p_t\leq f_t(\delta_{1:t-1};\{\theta_j\}_{j\in\mathcal{I}_t})\}}{g(\tilde{{\delta}}_{1:T}^{-t})}.\]

Since $\tilde{{\delta}}_{1:T}^{-t}$ is independent of $p_t$, we take conditional expectations to obtain
\begin{align}
    \E\left[\frac{\mathbb{I}\{p_t\leq f_t(\delta_{1:t-1};\{\theta_j\}_{j\in\mathcal{I}_t})\}}{g({\delta}_{1:T})}\mid \gF_{t-1}\right]&=\E\left[\frac{\mathbb{I}\{p_t\leq f_t(\delta_{1:t-1};\{\theta_j\}_{j\in\mathcal{I}_t})\}}{g(\tilde{{\delta}}_{1:T}^{-t})}\mid \gF_{t-1}\right] \nonumber \\
    ~&\Eqmark{i}\leq \E\left[\frac{ f_t(\delta_{1:t-1};\{\theta_j\}_{j\in\mathcal{I}_t})}{g(\tilde{{\delta}}_{1:T}^{-t})}\mid \gF_{t-1}\right] \nonumber \\
    ~&\Eqmark{ii}\leq \E\left[\frac{ f_t(\delta_{1:t-1};\{\theta_j\}_{j\in\mathcal{I}_t})}{g({{\delta}}_{1:T})}\mid \gF_{t-1}\right], \nonumber
\end{align}
where inequality (i) follows by taking expectation only with respect to $p_t$ by the conditional super-uniformity property (\ref{csuag}); and inequality (ii) follows because $g({\delta}_{1:T})\leq g(\tilde{{\delta}}_{1:T}^{-t})$ since $\delta_i\leq \tilde{\delta}_i$ for all $i$ by monotonicity of the online FDR rule. This concludes the proof of the Lemma \ref{lem:super-uniformity}.

}
\end{proof}

Denote $C_{1:t}:=(C_1,\dots,C_t)$, where $C_t=\mathbb{I}\{p_t\leq \lambda\}$, with a fixed constant $\lambda\in(0,1)$. Furthermore, assume $\alpha_t=h_t(\delta_{1:t-1},C_{1:t-1}, \{\theta_{j}\}_{j\in\mathcal{I}_t})$, for some coordinate-wise non-decreasing function $h_t:\{0,1\}^{2(t-1)+|\mathcal{I}_t|}\to[0,\lambda]$. Define the sigma-fields $\mathcal{J}_{t-1}=\sigma(\delta_{1:t-1};C_{1:t-1};\{\theta_j\}_{j\in\mathcal{I}_t})$.
For independent $p$-values, we have the following guarantees.

\begin{lemma}\label{lem:super-uniformity-saffron}
    Let $g:\{0,1\}^T\to\mathbb{R}$ be any coordinate-wise non-decreasing function. 
    Then, for any index $t\leq T$ such that $H_t\in\mathcal{H}_0$, we have
    \begin{align}
       &\E\left[\frac{h_t(\delta_{1:t-1},C_{1:t-1},\{\theta_j\}_{j\in\mathcal{I}_t})\mathbb{I}\{p_t>\lambda\}}{(1-\lambda)g(\delta_{1:T})}\mid \mathcal{J}_{t-1}\right] \nonumber \\
       ~&\geq\E\left[\frac{h_t(\delta_{1:t-1},C_{1:t-1},\{\theta_j\}_{j\in\mathcal{I}_t})}{g(\delta_{1:T})}\mid \mathcal{J}_{t-1}\right] \nonumber \\
        ~&\geq\E\left[\frac{\mathbb{I}\{p_t\leq h_t(\delta_{1:t-1},C_{1:t-1},\{\theta_j\}_{j\in\mathcal{I}_t})\}}{g(\delta_{1:T})}\mid \mathcal{J}_{t-1}\right] \nonumber
    \end{align}
\end{lemma}

\begin{proof}
    {\color{black}
    The second inequality is a consequence of Lemma \ref{lem:super-uniformity}, so we only prove the first inequality.
Denote ${p}_{1:T}=(p_1,\dots,p_T)$, and $\bar{{p}}^{t\to 1}_{1:T}=(\bar{p}_1,\dots,\bar{p}_T)$, where 
\begin{equation}
    \bar{p}_i=\begin{cases}
        1\quad \text{if}\; i=t \\
        p_i\quad \text{if}\; i\neq t,
    \end{cases} \nonumber 
\end{equation}
i.e., the $t$-th component of $\bar{{p}}^{t\to 1}_{1:T}$ is set to one and the other components of $\bar{{p}}^{t\to 1}_{1:T}$ equals to that of ${p}_{1:T}$. Define $\bar{C}_i=\mathbb{I}\{\bar{p}_i\leq \lambda\}$ and $\bar{\delta}_i=\mathbb{I}\{\bar{p}_i\leq h_i(\bar{\delta}_{1:i-1},\bar{C}_{1:i-1},\{\theta_j\}_{j\in\mathcal{I}_i})\}$ respectively. Let $\delta_{1:T}=(\delta_1,\dots,\delta_T)$ and $\bar{\delta}^{t\to 1}_{1:T}=(\bar{\delta}_1,\dots,\bar{\delta}_T)$ denote the decision vectors using $p_{1:T}$ and $\bar{p}_{1:T}$, respectively. Similarly, $C_{1:T}=(C_1,\dots,C_T)$ and $\bar{C}_{1:T}=(\bar{C}_1,\dots,\bar{C}_T)$. By construction, we have
\begin{enumerate}
    \item $\bar{\delta}_i=\delta_i$ and $\bar{C}_i=C_i$ for all $i\leq t$, hence $h_i(\delta_{1:i-1},C_{1:i-1},\{\theta_j\}_{j\in\mathcal{I}_i})=h_i(\bar{\delta}_{1:i-1},\bar{C}_{1:i-1},\{\theta_j\}_{j\in\mathcal{I}_i})$ for all $i\leq t$.
    \item $\bar{\delta}_t=\bar{C}_t=0$, and hence $\bar{\delta}_i\leq \delta_i$ for all $i\geq t$, due to monotonicity of the function $h_i$.
\end{enumerate}
Hence, on the event $\{p_t>\lambda\}$, we have $\delta_t=\bar{\delta}_t=0$ and $C_t=\bar{C}_t=0$, and hence also $\delta_{1:T}=\bar{\delta}_{1:T}^{t\to 1}$. Therefore, we obtain
\[\frac{h_t(\delta_{1:t-1},C_{1:t-1},\{\theta_j\}_{j\in\mathcal{I}_t})\mathbb{I}\{p_t>\lambda\}}{(1-\lambda)g(\delta_{1:T})}=\frac{h_t(\delta_{1:t-1},C_{1:t-1},\{\theta_j\}_{j\in\mathcal{I}_t})\mathbb{I}\{p_t>\lambda\}}{(1-\lambda)g(\bar{\delta}_{1:T})}. \]
Since $\bar{\delta}_{1:T}^{t\to 1}$ is independent of $p_t$, we  may take conditional expectations to obtain:
 \begin{align}
       &\E\left[\frac{h_t(\delta_{1:t-1},C_{1:t-1},\{\theta_j\}_{j\in\mathcal{I}_t})\mathbb{I}\{p_t>\lambda\}}{(1-\lambda)g(\delta_{1:T})}\mid \mathcal{J}_{t-1}\right] \nonumber \\
       &=\E\left[\frac{h_t(\delta_{1:t-1},C_{1:t-1},\{\theta_j\}_{j\in\mathcal{I}_t})\mathbb{I}\{p_t>\lambda\}}{(1-\lambda)g(\bar{\delta}^{t\to 1}_{1:T})}\mid \mathcal{J}_{t-1}\right] \nonumber \\
       ~&\Eqmark{i}\geq\E\left[\frac{h_t(\delta_{1:t-1},C_{1:t-1},\{\theta_j\}_{j\in\mathcal{I}_t})}{g(\bar{\delta}^{t\to 1}_{1:T})}\mid \mathcal{J}_{t-1}\right] \nonumber \\
        ~&\Eqmark{ii}\geq\E\left[\frac{h_t(\delta_{1:t-1},C_{1:t-1},\{\theta_j\}_{j\in\mathcal{I}_t})}{g(\delta_{1:T})}\mid \mathcal{J}_{t-1}\right], \nonumber 
    \end{align}
    where inequality (i) follows by taking an expectation only with respect to $p_t$ by invoking the conditional super-uniformity property (\ref{csuag-SF}); and inequality (ii) follows because $g(\delta_{1:T})\geq g(\bar{\delta}^{t\to 1}_{1:T})$ since $\delta_i\geq \bar{\delta}_i$ for all $i$ by monotonicity of the online FDR rule. This concludes the proof of Lemma \ref{lem:super-uniformity-saffron}.
}
\end{proof}

\begin{lemma}
\label{lem:decondition_event}
Let $\mathcal{F}$ be a sigma-field and let $\mathcal{E},\mathcal{A}\in\mathcal{F}$ be the events and $\mathcal{E}\perp\mathcal{A}$.
Let $Z$ be a random variable and $b\in\mathbb{R}$. 
Assume that
 $\Pr(\mathcal{E})\ge 1-\delta'$ and  on $\mathcal{E}$, $\Pr(Z\le b\mid \mathcal{F})\ge 1-\delta$. 
Then
\[
\Pr(Z\le b\mid \mathcal{A})\ge 1-\delta-\delta'.
\]
\end{lemma}

\begin{proof}
Since $\mathcal{E}\in\mathcal{F}$,
\begin{align*}
    &\Pr(Z\le b\mid\mathcal{A})\ \ge\ \Pr(Z\le b,\mathcal{E}\mid\mathcal{A})\\
=& \E\!\left[\mathbb{I}_{\mathcal{E}}\Pr(Z\le b\mid \mathcal{F})\mid \mathcal{A}\right]
\ge (1-\delta)\Pr(\mathcal{E}\mid\mathcal{A})
\ge 1-\delta-\delta'.
\end{align*}
The last inequality comes from the independence of $\mathcal{E}$ and $\mathcal{A}$ such that $\Pr(\mathcal{E}\mid\mathcal{A})=\Pr(\mathcal{E})\geq1-\delta'$.
\end{proof}

\subsection{Proof of Theorem \ref{the:FDR_GAIF_ind} }\label{proof:them1}

\begin{proof}

 {\it Proof of Theorem \ref{the:FDR_GAIF_ind}} (a). (mFDR control of GAIF procedures.)
For the GAIF procedures,
\begin{eqnarray}
    \E[V(t)]&=&\E\left[\sum_{j\leq t,j\in\mathcal{H}_0}\mathbb{I}\{p_j\leq \alpha_j\}\right]=\sum_{j\leq t}\E\left[(1-\theta_j)\mathbb{I}\{p_j\leq \alpha_j\}\right] \nonumber \\
    &\Eqmark{i}{=}& \sum_{j\leq t}\E\left[\E\left[(1-\theta_j)\mathbb{I}\{p_j\leq \alpha_j\}\mid \gF_{j-1}\right]\right]   \nonumber \\
    &\Eqmark{ii}{\leq}& \E\left[\sum_{j\leq t}\alpha_j(1-\theta_j)\right] \nonumber \\
    &\Eqmark{iii}{\leq}& \E\left[\sum_{j\in\mathcal{I}_t}\alpha_j(1-\theta_j)+\sum_{j\in\bar{\mathcal{I}}_t}\alpha_j\right] \nonumber \\
    &\Eqmark{iv}{\leq}& \alpha \E\left[1\vee \sum_{j\leq t}\delta_j\right], \nonumber 
\end{eqnarray}
where equality (i) is derived by conditioning on $\mathcal{F}_{j-1}$ and applying the law of iterated expectations. Inequality (ii) is a consequence of the conditional super-uniformity property stated in (\ref{csuag}). Subsequent inequalities (iii) follows from the fact that $1-\theta_j\leq1$ for all $j\in\bar{\mathcal{I}}_t$ and (iv) follows from the definition of GAIF procedures. Therefore, we obtain the conclusion that $\mFDR(t)\leq \alpha$ for all GAIF procedures.

{\it Proof of Theorem \ref{the:FDR_GAIF_ind}} (b). (FDR control of GAIF procedures  procedures.)
We first prove online FDR control for GAIF procedures. Under the independence and the monotonicity assumptions, we have
\begin{eqnarray}
    \operatorname{FDR}(t)&=&\E\left[\frac{\sum_{j\leq t,j\in \mathcal{H}_0}\mathbb{I}\{p_j\leq \alpha_j\}}{1\vee\sum_{j\leq t}\delta_j}\right] \nonumber \\
    ~&=& \sum_{j\leq t,j\in\mathcal{H}_0}\E\left[\frac{\mathbb{I}\{p_j\leq \alpha_j\}}{1\vee\sum_{j\leq t}\delta_j}\right] \nonumber \\
    ~&\Eqmark{i}=&\sum_{j\leq t,j\in\mathcal{H}_0}\E\left[\E\left[\frac{\mathbb{I}\{p_j\leq \alpha_j\}}{1\vee\sum_{j\leq t}\delta_j}\mid \gF_{j-1}\right]\right] \nonumber \\
    ~&\Eqmark{ii}\leq& \sum_{j\leq t,j\in\mathcal{H}_0}\E\left[\frac{\alpha_j}{1\vee\sum_{j\leq t}\delta_j}\right] \nonumber \\
     ~&=& \E\left[\frac{\sum_{j\leq t}(1-\theta_j)\alpha_j}{1\vee\sum_{j\leq t}\delta_j}\right] \nonumber \\
     ~&\Eqmark{iii}\leq& \E\left[\frac{\sum_{j\in\mathcal{I}_t}(1-\theta_j)\,\alpha_j
         + \sum_{j\in\bar{\mathcal{I}}_t}\alpha_j
   }{1 \vee \sum_{j = 1}^t \delta_j}\right] \\
~&=&\E\left[\widehat{\operatorname{FDP}}_{\operatorname{GAIF}}(t)\right] \nonumber \\
    ~&\Eqmark{iv}\leq& \alpha, \nonumber 
\end{eqnarray}
where the equality (i) follows from the law of iterated expectations by conditioning on $\mathcal{F}_{j-1}$ and inequality (ii) results from Lemma \ref{lem:super-uniformity} by setting $g(\delta_{1:t})=(1\vee \sum_{j\leq t}\delta_j)$,  and the inequality (iii) holds because $1-\theta_j\leq 1$ for all $j\in\bar{\mathcal{I}}_t$, the inequality (iv) follows from definition of GAIF, which completes the proof of FDR control for monotone GAIF procedures.

\end{proof}

}

{\color{black}\subsection{Proof of Theorem \ref{the:FDR_Ada_GAIF_ind}\label{proof:them1-Ada-GAIF}}

\begin{proof}

 {\it Proof of Theorem \ref{the:FDR_Ada_GAIF_ind}} (a). (mFDR control of Adaptive GAIF procedures.)

Now we prove mFDR control for the Adaptive GAIF procedures. 
{\color{black}For any $c\in(0,1)$ that is predictable with respect to 
$\mathcal J_{t-1}$, the conditional super-uniformity of null p-values implies that, 
for $t\in\mathcal H_0$,
\begin{equation}\label{eq:condi-super-SF-lambda}
    \E\left[\frac{\mathbb{I}\{p_t>c\}}{1-c}\mid \mathcal{J}_{t-1}\right]\geq 1 \geq \E\left[\frac{\mathbb{I}\{p_t\leq c\}}{c}\mid \mathcal{J}_{t-1}\right].
\end{equation}

}

Note that for any time $t\in[T]$, we have
\begin{eqnarray}
    \E[V(t)]&=&\E\left[\sum_{j\leq t,j\in\mathcal{H}_0}\mathbb{I}\{p_j\leq \alpha_j\}\right]=\sum_{j\leq t}\E\left[(1-\theta_j)\mathbb{I}\{p_j\leq \alpha_j\}\right] \nonumber \\
    &\Eqmark{i}=& \sum_{j\leq t}\E\left[\E\left[(1-\theta_j)\mathbb{I}\{p_j\leq \alpha_j\}\mid \mathcal{J}_{j-1}\right]\right]   \nonumber \\
    &\Eqmark{ii}\leq& \E\left[\sum_{j\leq t}\alpha_j(1-\theta_j)\right] \nonumber \\
     &\Eqmark{iii}\leq& \E\left[\sum_{j\leq t}\alpha_j(1-\theta_j)\frac{\mathbb{I}\{p_j>\lambda\}}{(1-\lambda)}\right] \nonumber \\
    &\Eqmark{iv}\leq& \E\left[\sum_{j\in\mathcal{I}_t}\alpha_j(1-\theta_j)\frac{\mathbb{I}\{p_j>\lambda\}}{(1-\lambda)}+\sum_{j\in\bar{\mathcal{I}}_t}\alpha_j \frac{\mathbb{I}\{p_j>\lambda\}}{(1-\lambda)}\right] \nonumber \\
    &\Eqmark{v}\leq& \alpha \E\left[1\vee \sum_{j\leq t}\delta_j\right],
\end{eqnarray}
where the equality (i) follows from the law of iterated expectations by conditioning on $\mathcal{J}_{j-1}$ and inequality (ii) applying the conditional super-uniformity property in (\ref{csuag-SF}), and the inequality (iii) also follows by the law of iterated expectations by conditioning on $\mathcal{J}_{j-1}$ and then applying the conditional super-uniformity property in \eqref{eq:condi-super-SF-lambda} with $c=\lambda$, and the inequality (iv) holds since the fact that $1-\theta_j\leq 1$ for all $j\in\bar{\mathcal{I}}_t$, and the inequality (v) follows from the construction such that $\widehat{\FDP}_{\text{Adaptive-GAIF}}\leq \alpha$. Therefore, we obtain the conclusion that $\operatorname{mFDR}(t)\leq \alpha$ for Adaptive GAIF procedures.

{\it Proof of Theorem \ref{the:FDR_Ada_GAIF_ind}} (b). (FDR control of Adaptive GAIF procedures.)

We next prove the online FDR control for Adaptive GAIF procedures. Under the independence and the monotonicity assumptions, we have
\begin{eqnarray}
  \operatorname{FDR}(t)&=&\E\left[\frac{\sum_{j\leq t,j\in \mathcal{H}_0}\mathbb{I}\{p_j\leq \alpha_j\}}{1\vee\sum_{j\leq t}\delta_j}\right] \nonumber \\
    ~&=& \sum_{j\leq t,j\in\mathcal{H}_0}\E\left[\frac{\mathbb{I}\{p_j\leq \alpha_j\}}{1\vee\sum_{j\leq t}\delta_j}\right] \nonumber \\
    ~&\Eqmark{i}=&\sum_{j\leq t,j\in\mathcal{H}_0}\E\left[\E\left[\frac{\mathbb{I}\{p_j\leq \alpha_j\}}{1\vee\sum_{j\leq t}\delta_j}\mid \mathcal{J}_{j-1}\right]\right] \nonumber \\
    ~&\Eqmark{ii}\leq& \sum_{j\leq t,j\in\mathcal{H}_0}\E\left[\frac{\alpha_j}{1\vee\sum_{j\leq t}\delta_j}\right] \nonumber \\
    ~&\Eqmark{iii}\leq& \sum_{j\leq t,j\in\mathcal{H}_0}\E\left[\frac{\alpha_j}{1\vee\sum_{j\leq t}\delta_j}\cdot\frac{\mathbb{I}\{p_j>\lambda\}}{(1-\lambda)}\right] \nonumber \\
    ~&=&\E\left[\frac{\sum_{j=1}^{t}\alpha_j(1-\theta_j)\frac{\mathbb{I}\{p_j>\lambda\}}{(1-\lambda)}}{1\vee\sum_{j\leq t}\delta_j}\right] \nonumber \\
    ~&\Eqmark{iv}\leq&\E\left[\frac{\sum_{j\in\mathcal{I}_t}\alpha_j(1-\theta_j)\frac{\mathbb{I}\{p_j>\lambda\}}{(1-\lambda)}+\sum_{j\in\bar{\mathcal{I}}_t}\alpha_j\frac{\mathbb{I}\{p_j>\lambda\}}{(1-\lambda)}}{1\vee\sum_{j\leq t}\delta_j}\right] \nonumber \\
    ~&=& \E\left[\widehat{\operatorname{FDP}}_{\operatorname{Ada-GAIF}}(t)\right] \nonumber \\
    ~&\Eqmark{v}\leq& \alpha, \nonumber 
\end{eqnarray}
where (i) and (ii) follows form the law of iterated expectations by conditioning on $\mathcal{J}_{j-1}$ and applying Lemma \ref{lem:super-uniformity-saffron} by setting $g(\delta_{1:t})=(1\vee \sum_{j\leq t}\delta_j)$,  the inequality (iii) also follows from the law of iterated expectations by conditioning on $\mathcal{J}_{j-1}$ and applying Lemma \ref{lem:super-uniformity-saffron} by setting $g(\delta_{1:t})=(1\vee \sum_{j\leq t}\delta_j)$, the inequality (iv) holds since $1-\theta_j\leq 1$ for all $j\in\bar{\mathcal{I}}_t$, and the  inequality (v) follows from the definition of Adaptive GAIF,  which completes the proof of online FDR control for monotone Adaptive GAIF.

\end{proof}
}

\subsection{Proof of Theorem \ref{the:mFDR-control-dep}}
\begin{proof}
{\color{black}
Define locally conditional super-uniformity as follows: if the null hypothesis $H_t$ is true, then for all $\alpha_t\in[0,1]$,
\[\Pr\left(p_t\leq \alpha_t\mid \gF_{\text{dep}}^{-\gX^t}\right)\leq \alpha_t,\]
where $\gX^t:=\{t-L_t,\dots,t-1\}$, and $\gF_{\text{dep}}^{-\gX^t}:=\sigma(\delta_{1:t-L_t-1};\{\theta_j\}_{j\in\mathcal{I}_{t-L_t-1},1\le j\le t-L_t-1})$. 

This condition is immediately true by local dependence \citep{Zrnic2021asynchronous}. 
    Note that for any time $t\in[T]$, we have
\begin{eqnarray}
    \E[V(t)]&=&\E\left[\sum_{j\leq t,j\in\mathcal{H}_0}\mathbb{I}\{p_j\leq \alpha_j\}\right]=\sum_{j\leq t}\E\left[(1-\theta_j)\mathbb{I}\{p_j\leq \alpha_j\}\right] \nonumber \\
    &=& \sum_{j\leq t}\E\left[\E\left[(1-\theta_j)\mathbb{I}\{p_j\leq \alpha_j\}\mid \gF_{\text{dep}}^{-\gX^j}\right]\right]   \nonumber \\
    &\leq& \sum_{j\leq t} \E[(1-\theta_j)\alpha_j]=\E\left[\sum_{j\leq t}\alpha_j(1-\theta_j)\right] \nonumber \\
    &\leq& \E\left[\sum_{j\in\mathcal{I}_t}\alpha_j(1-\theta_j)+\sum_{j\in\bar{\mathcal{I}}_t}\alpha_j\right] \nonumber \\
    &\leq& \alpha \E\left[1\vee \sum_{j\leq t,j\notin\{t-L_t,\dots,t-1\}}\delta_j\right],\ \nonumber \\
    &\leq& \alpha\E\left[1\vee \sum_{j\leq t}\delta_j\right],
\end{eqnarray}

where the first inequality follows from the law of iterated expectations by conditioning on $\mathcal{F}_{\text{dep}}^{-\gX^t}$ and then applying the conditional super-uniformity property and by noticing that the measurability of $\alpha_j$ with respect to $\mathcal{F}_{\text{dep}}^{-\gX^t}$, and the second inequality follows by the fact that $1-\theta_t\leq 1$, and the third inequality follows from the construction such that $\widehat{\FDP}_{\text{GAIF}_{\text{dep}}}\leq \alpha$. Therefore, we obtain the conclusion that $\operatorname{mFDR}(t)\leq \alpha$ for $\text{GAIF}_{\text{dep}}$. 

Define $\mathcal{J}_{\text{dep}}^{-\gX^t}:=\sigma(\delta_{1:t-L_t-1};C_{1:t-L_t-1};\{\theta_j\}_{j\in\mathcal{I}_{t-L_t-1},1\le j\le t-L_t-1})$. Under local dependence, we have the conditional super-uniformity for null $p$-values:
{\color{black}
\begin{equation}
    \E\left[\frac{\mathbb{I}\{p_t>c\}}{1-c}\mid \mathcal{J}_{\text{dep}}^{-\gX_t}\right]\geq 1 \geq \E\left[\frac{\mathbb{I}\{p_t\leq c\}}{c}\mid \mathcal{J}_{\text{dep}}^{-\gX_t}\right],\nonumber
\end{equation}
where $c\in(0,1)$ is predictable with respect to $\mathcal{J}_{\text{dep}}^{-\gX_t}$.
}
The mFDR control can be obtained for $\text{Ada-GAIF}_{\text{dep}}$ as follows. 

\begin{eqnarray}
    \E[V(t)]&=&\E\left[\sum_{j\leq t,j\in\mathcal{H}_0}\mathbb{I}\{p_j\leq \alpha_j\}\right]=\sum_{j\leq t}\E\left[(1-\theta_j)\mathbb{I}\{p_j\leq \alpha_j\}\right] \nonumber \\
    &\Eqmark{i}=& \sum_{j\leq t}\E\left[\E\left[(1-\theta_j)\mathbb{I}\{p_j\leq \alpha_j\}\mid \mathcal{J}_{\text{dep}}^{-\gX^j}\right]\right]   \nonumber \\
    &\Eqmark{ii}\leq& \sum_{j\leq t} \E[(1-\theta_j)\alpha_j]=\E\left[\sum_{j\leq t}\alpha_j(1-\theta_j)\right] \nonumber \\
    &\Eqmark{iii}\leq& \E\left[\sum_{j\in\mathcal{I}_t}\alpha_j(1-\theta_j)\frac{\mathbb{I}\{p_j>\lambda\}}{(1-\lambda)}+\sum_{j\in\bar{\mathcal{I}}_t}\alpha_j\frac{\mathbb{I}\{p_j>\lambda\}}{(1-\lambda)}\right] \nonumber \\
    &\Eqmark{iv}\leq& \E\left[\sum_{j < t-L_t} \frac{\alpha_j \mathbb{I}\{p_j>\lambda\}}{1-\lambda} \mathcal{W}_j + \sum_{t-L_t \le j \le t} \frac{\alpha_j}{1-\lambda} \mathcal{W}_j\right] \nonumber \\
     &\Eqmark{v}\leq& \alpha \E\left[1\vee \sum_{j\leq t,j\notin\{t-L_t,\dots,t-1\}}\delta_j\right]\ \nonumber \\
    &\Eqmark{vi}\leq& \alpha \E\left[1\vee \sum_{j\leq t}\delta_j\right],
\end{eqnarray}
where $\mathcal{W}_j = (1-\theta_j)\mathbb{I}\{j \in \mathcal{I}_t\} + \mathbb{I}\{j \in \bar{\mathcal{I}}_t\}$ and the equality (i) follows from the law of iterated expectations by conditioning on $\mathcal{J}_{\rm dep}^{-\gX^j}$, and the inequality (ii) applying the conditional super-uniformity, and the inequality (iii) also follows by the law of iterated expectations by conditioning on $\mathcal{J}_{\rm dep}^{-\gX^t}$ and then applying the conditional super-uniformity. The inequality (iv) follows from the fact that $\mathbb{I}\{p_j>\lambda\}\leq 1$ for all $j\in\{t-L_t,\dots,t\}$ and $1-\theta_t\leq 1$. The inequality (v) follows from the construction such that $\widehat{\FDP}_{\text{Ada-GAIF}_{\text{dep}}}\leq \alpha$, which concludes the proof.
}
\end{proof}

\subsection{Proof of Proposition \ref{prop-inde-p} }\label{appen_sub:online_cp}
We first restate Proposition \ref{prop-inde-p} as follows and finish the proof. The proof of Proposition \ref{prop:online_conf_val_exch} essentially follows the argument of Theorem 8.2 in \cite{angelopoulos2024theoretical}, with an extension to our online conformal testing setting.
\begin{proposition}[Validity and Mutual Independence of Online Conformal $p$-values under Exchangeability and Symmetric Scores]
\label{prop:online_conf_val_exch}
Suppose at each time $t$, the score function $V(\cdot;\gD_t)$ is constructed through the current data $\gD_t=\big((\X_i,Y_i):-n\leq i\leq t\big)$. The p-value of each time $t$ is constructed as
\begin{equation}\label{eq:conf_p-appen}
    p_t=\frac{\sum_{i\in{\gC}_{0t}} \I\{V(\X_i;\gD_t)<V(\X_t;\gD_t)\}+\xi_t\cdot (1+\sum_{i\in \gC_{0t}} \mathbb{I}\{V(\X_i;\gD_t)=V(\X_t;\gD_t)\})}{1+|\gC_{0t}|}.
\end{equation}
Suppose Assumption \ref{assump:conformal_setting} holds and the score function $V(\cdot;\gD_t)$ is symmetric to $\{(\X_i,Y_i):-n\leq i\leq t, \theta_i=0\}$. Then under the null hypothesis,
\begin{enumerate}
    \item Each null $p$-value $p_t$ with $\theta_t=0$ is marginally uniformly distributed on $[0,1]$.
    \item The null $p$-values $\{p_t:t\in\mathbb{N},\theta_t=0\}$ are mutually independent.
\end{enumerate}
\end{proposition}




\begin{proof}
We adapt the standard argument for the validity of conformal $p$-values to our online setting with exchangeable data. For each time $t$, define
\[\color{black}
\Phi_t = \Biggl(\{(\X_i,Y_i) : i\in \gC_{0t}\cup\{t\}\},\, (\theta_i: i< t),\, \bigl((\X_i,Y_i) : i\in \gC_{1t}\bigr)\Biggr),
\]
where $\Phi_t$ contains the unordered set of conformity scores for indices in $\gC_{0t}\cup\{t\}$, the true state $\theta_i$, and the data corresponding to indices not used in the calibration set (denoted here by $\gC_{1t}$). To prove the mutual independence, consider any time indices $t\leq T$ and let $x_t,x_{t+1},\dots,x_T\in [0,1]$ be arbitrary. {\color{black}For ease of notation, we first consider the case in which all p-values $\{p_s\}_{s=t}^T$ are associated with null hypotheses. The argument extends directly to the more general setting in which some non-null hypotheses appear between them.} 

Since $\{p_k\}_{k=t+1}^T$ are determined by $\Phi_t\cup\{(\X_i,Y_i)\}_{i=t+1}^{T}$ and $p_t$ is independent of $\{(\X_i,Y_i)\}_{i=t+1}^{T}$, it holds that 
\begin{align*}
    &\Pr\bigl(p_t\leq x_t,\, p_{t+1}\leq x_{t+1},\,\dots,\, p_T\leq x_T\bigr)\\[1mm]
    =&\; \mathbb{E}\Bigl[\mathbb{I}\{p_t\leq x_t\}\mathbb{I}\{p_{t+1}\leq x_{t+1}\}\cdots \mathbb{I}\{p_T\leq x_T\}\Bigr]\\[1mm]
    =&\; \mathbb{E}\Biggl[\mathbb{E}\Bigl[\mathbb{I}\{p_t\leq x_t\}\,\Big|\,\textcolor{black}{\Phi_t\cup\{(\X_i,Y_i)\}_{i=t+1}^{T}}\Bigr]\cdot \mathbb{I}\{p_{t+1}\leq x_{t+1}\}\cdots \mathbb{I}\{p_T\leq x_T\}\Biggr]\\[1mm]
    =&\; \mathbb{E}\Biggl[\mathbb{E}\Bigl[\mathbb{I}\{p_t\leq x_t\}\,\Big|\,\Phi_t\Bigr]\cdot \mathbb{I}\{p_{t+1}\leq x_{t+1}\}\cdots \mathbb{I}\{p_T\leq x_T\}\Biggr].
\end{align*}
The key observation is that, by exchangeability, the conditional distribution of $p_t$ given $\Phi_t$ is uniform on $[0,1]$, so that
\[\mathbb{E}\Bigl[\mathbb{I}\{p_t\leq x_t\}\,\Big|\,\Phi_t\Bigr] = x_t.
\]
Thus,
\[
\Pr\bigl(p_t\leq x_t,\, p_{t+1}\leq x_{t+1},\,\dots,\, p_T\leq x_T\bigr)
= x_t\, \mathbb{E}\Bigl[\mathbb{I}\{p_{t+1}\leq x_{t+1}\}\cdots \mathbb{I}\{p_T\leq x_T\}\Bigr].
\]
Next, we note the following two key facts:

\medskip
\textbf{(a) Uniformity of single p-value $p_t$:}\\[1mm]
Fix any $t\in[T]$. Let $\Omega_t$ be the set of all {permutations} of $\gC_0\cup[T]$ that fix indices outside of $\gC_{0t}\cup\{t\}$. Note that given $\Phi_t$, the only randomness for $p_t$ is the order of $\{(\X_i,Y_i) : i\in \gC_{0t}\cup\{t\}\}$ and $\xi_t$.

For any $\sigma\in\Omega_t$ and given $\{\X_i : i\in \gC_{0t}\cup\{t\}\}=\{\x_i : i\in \gC_{0t}\cup\{t\}\}$ as the realizations, define 
\begin{align*}
    &p_{t}((\gD_T)_\sigma)\\
    =&\frac{\sum_{i\in{\gC}_{0t}\cup\{t\}} \I\{V(\x_{\sigma(i)};(\gD_t)_\sigma)<V(\x_{\sigma(t)};(\gD_t)_\sigma)\}+\xi_t\cdot \sum_{i\in \gC_{0t}\cup\{t\}} \mathbb{I}\{V(\x_{\sigma(i)};(\gD_t)_\sigma)=V(\x_{\sigma(t)};(\gD_t)_\sigma)\}}{1+|\gC_{0t}|}\\
    \Eqmark{i}{=}&\frac{\sum_{i\in{\gC}_{0t}\cup\{t\}} \I\{V(\x_{\sigma(i)};\gD_t)<V(\x_{\sigma(t)};\gD_t)\}+\xi_t\cdot \sum_{i\in \gC_{0t}\cup\{t\}} \mathbb{I}\{V(\x_{\sigma(i)};\gD_t)=V(\x_{\sigma(t)};\gD_t)\}}{1+|\gC_{0t}|}\\
    \Eqmark{ii}{=}&\frac{\sum_{i\in{\gC}_{0t}\cup\{t\}} \I\{V(\x_{i};\gD_t)<V(\x_{\sigma(t)};\gD_t)\}+\xi_t\cdot \sum_{i\in \gC_{0t}\cup\{t\}} \mathbb{I}\{V(\x_{i};\gD_t)=V(\x_{\sigma(t)};\gD_t)\}}{1+|\gC_{0t}|}.
\end{align*}
Equality (i) holds since $V(\cdot;\sigma(D_{t}))=V(\cdot;D_{t})$ from the symmetry of score function $V$. Equality (ii) is true by the property that $\sigma$ only permutes indices in $\gC_{0t}\cup\{t\}$ such that $\{V(\X_i;D_{t}):i\in\gC_{0t}\}=\{V(\X_{\sigma(i)};D_{t}):i\in\gC_{0t}\}$.

Denote $Q_\alpha(S_i:i\in\mathcal{I})$ as the $\alpha$-th quantile of the set $\{S_i\}_{i\in\gI}$. Let $\{X_i:i\in \gC_{0t}\cup\{t\}\}=\{x_i:i\in \gC_{0t}\cup\{t\}\}$ as a set of realizations and 
$$q=Q_\alpha(V(x_i;\gD_t):i\in\gC_{0t}\cup\{t\}).$$

For any $\sigma\in\Omega_t$ such that $V(\x_{\sigma(t)};\gD_t)>q$, we have
$$p_{t}((\gD_T)_\sigma)\geq \frac{\sum_{i\in{\gC}_{0t}\cup\{t\}} \I\{V(\x_{i};\gD_t)<V(\x_{\sigma(t)};\gD_t)\}}{1+|\gC_{0t}|}{\geq} \frac{N_=+N_-}{1+|\gC_{0t}|}\geq \alpha;$$
Here
$$N_-=\sum_{i\in\gC_{0t}\cup\{t\}}\mathbb{I}\{(V(x_i;\gD_t)<q\},\quad N_{=}=\sum_{i\in\gC_{0t}\cup\{t\}}\mathbb{I}\{(V(x_i;\gD_t)=q\}.$$
The last inequality holds by the property of quantile function such that $N_=+N_-\geq 1+\alpha(|\gC_{0t}+1|)$.

For any $\sigma\in\Omega_t$ such that $V(\x_{\sigma(t)};\gD_t)<q$, we have
$$p_{t}((\gD_T)_\sigma)\leq \frac{\sum_{i\in{\gC}_{0t}\cup\{t\}} \I\{V(\x_{i};\gD_t)\leq V(\x_{\sigma(t)};\gD_t)\}}{1+|\gC_{0t}|}\leq \frac{N_-}{1+|\gC_{0t}|}\leq \alpha,$$
which is from the property that $N_-\leq\alpha(|\gC_{0t}+1|)$.

And for any $\sigma\in\Omega_t$ such that $V(\x_{\sigma(t)};\gD_t)= q$, we have 
\begin{align}
    p_{t}((\gD_T)_\sigma) 
    =&\frac{\sum_{i\in{\gC}_{0t}\cup\{t\}} \I\{V(\x_{i};\gD_t) <V(\x_{\sigma(t)};\gD_t)\}+\xi_t\cdot \sum_{i\in \gC_{0t}\cup\{t\}} \mathbb{I}\{V(\x_{i};\gD_t)=V(\x_{\sigma(t)};\gD_t)\}}{1+|\gC_{0t}|} \nonumber \\
    \leq& \frac{N_-+\xi_t N_=}{1+|\gC_{0t}|}. \nonumber 
\end{align}

Hence,
\begin{align*}
    \Pr(p_t\leq\alpha\mid\Phi_t)&=\frac{1}{(|\gC_{0t}|+1)!}\sum_{\sigma\in\Omega_t}\Pr(p_t((\gD_T)_\sigma)\leq\alpha)\\
    &\Eqmark{i}{=}\frac{1}{(|\gC_{0t}|+1)!}\sum_{\sigma\in\Omega_t}\left(\mathbb{I}\{V(\x_{\sigma(t)};\gD_t)<q\}+\mathbb{I}\{V(\x_{\sigma(t)};\gD_t)=q\}\frac{\alpha(|\gC_{0t}|+1)-N_-}{N_=}\right)\\
    &\Eqmark{ii}{=}\frac{1}{|\gC_{0t}|+1}\sum_{i\in\gC_{0t}\cup\{j\}}\left(\mathbb{I}\{V(\x_{i};\gD_t)<q\}+\mathbb{I}\{V(\x_{i};\gD_t)=q\}\frac{\alpha(|\gC_{0t}|+1)-N_-}{N_=}\right)\\
    &\Eqmark{iii}{=}\frac{1}{|\gC_{0t}|+1}\left(N_-+N_=\frac{\alpha(|\gC_{0t}|+1)-N_-}{N_=}\right)=\alpha,
\end{align*}
where equality (i) holds since $\xi_t$ is uniformly distributed and independent of everything else, making $$\E[\mathbb{I}\{V(\x_{\sigma(t)};\gD_t)=q,p_t((\gD_T)_\sigma)\leq\alpha\}\mid (\gD_T)_\sigma]=\mathbb{I}\{V(\x_{\sigma(t)};\gD_t)=q\}\frac{\alpha(|\gC_{0t}|+1)-N_-}{N_=}. $$
 Equality (ii) is from the fact that $\sum_{\sigma\in\Omega_t,\sigma(t)=i}\mathbb{I}\{V(\x_{\sigma(t)};\gD_t)<q\}= |\gC_{0t}|\mathbb{I}\{V(\x_{i};\gD_t)<q\}$ and $\sum_{\sigma\in\Omega_t,\sigma(t)=i}\mathbb{I}\{V(\x_{\sigma(t)};\gD_t)=q\}= |\gC_{0t}|\mathbb{I}\{V(\x_{i};\gD_t)=q\}$. And equality (iii) is direct by the definition of $N_-$ and $N_=$. Marginalizing over the $\Phi_t$ implies $\Pr(p_t\leq \alpha)=\alpha$ for all $t\in[T]$.

\textbf{(b) Independence of future $p$-values from $\Phi_t$:}\\[1mm]
Define the data set $\gD_T=((\X_{-n+1},Y_{-n+1}),\dots,(\X_T,Y_T))$. By exchangeability of the data, it holds that $\gD_T\overset{d}{=}(\gD_T)_{\sigma}$, where $(\gD_T)_{\sigma}$ is obtained from $\gD_T$ by permuting the data points according to $\sigma$. A key observation is that for any $\sigma\in\Omega_t$, defining $\sigma$ as above we have 
\[p_{t'}(\gD_T)=p_{t'}((\gD_T)_{\sigma})\]
for all $t'\in\{t+1,\dots,T\}$. Intuitively, this indicates that permuting the data according to $\sigma$ does not change $p$-values after time $t$.

This is because $p_{t'}((\gD_T)_{\sigma})$ is
\begin{eqnarray}
     && \frac{ \sum_{i \in \gC_{0t'}} \mathbb{I}\{V(\X_{\sigma(i)};\sigma(\gD_{t'})) < V(\X_{\sigma(t')};\sigma(\gD_{t'}))+\xi_{t'}(1+\sum_{i\in\gC_{0t'}}\mathbb{I}\{V(\X_{\sigma(i)};\sigma(\gD_{t'})) = V(\X_{\sigma(t')};\sigma(\gD_{t'}))\})}{1 + |\gC_{0t'}|} \nonumber \\
     &\Eqmark{i}{=}& \frac{ \sum_{i \in \gC_{0t'}} \mathbb{I}\{V(\X_{\sigma(i)};\gD_{t'}) < V(\X_{t'};\gD_{t'})\}+\xi_{t'}(1+\sum_{i\in\gC_{0t'}}\mathbb{I}\{V(\X_{\sigma(i)};\gD_{t'}) = V(\X_{t'};\gD_{t'})\})}{1 + |\gC_{0t'}|}\; \nonumber \\
     &\Eqmark{ii}{=}& \frac{ \sum_{i \in \gC_{0t'}} \mathbb{I}\{V(\X_i;\gD_{t'})< V(\X_{t'};\gD_{t'})\}+\xi_{t'}(1+\sum_{i\in\gC_{0t'}}\mathbb{I}\{V(\X_i;\gD_{t'}) = V(\X_{t'};\gD_{t'})\})}{1 + |\gC_{0t'}|}\; \nonumber \\
     &=&p_{t'}(\gD_T), \nonumber 
\end{eqnarray}
where equality (i) holds since $\sigma(t')=t'$ by definition and $V(\cdot;\sigma(D_{t'}))=V(\cdot;D_{t'})$ from the symmetry of score function $V$. And equality (ii) is true as $\{V(\X_i;D_{t'}):i\in\gC_{0t'}\}=\{V(\X_{\sigma(i)};D_{t'}):i\in\gC_{0t'}\}$.

Therefore, 
\begin{align*}
    &\Pr(p_t\leq x_t,\cdots,p_T\leq \textcolor{black}{x_T})\\
    =&\frac{1}{(|\gC_{0t}|+1)!}\E\left[\sum_{\sigma\in\Omega_t} \mathbb{I}\{p_t((\gD_T)_\sigma)\leq x_t,\cdots,p_T((\gD_T)_\sigma)\leq \textcolor{black}{x_T}\}\right]\\
    =&\frac{1}{(|\gC_{0t}|+1)!}\E\left[\sum_{\sigma\in\Omega_t} \mathbb{I}\{p_t(\gD_T)\leq x_t,\cdots,p_T(\gD_T)\leq \textcolor{black}{x_T}\}\right]\\
    =&\frac{1}{(|\gC_{0t}|+1)!}\E\left[\left(\sum_{\sigma\in\Omega_t} \mathbb{I}\{p_t(\gD_T)\leq x_t\}\right)\mathbb{I}\{p_{t+1}(\gD_T)\cdots,p_T(\gD_T)\leq \textcolor{black}{x_T}\}\right]\\
    =&\E\left[\Pr\left(p_t\leq x_t\mid\textcolor{black}{\Phi_t\cup\{(\X_i,Y_i)\}_{i=t+1}^{T}}\right)\mathbb{I}\{p_{t+1}(\gD_T)\cdots,p_T(\gD_T)\leq \textcolor{black}{x_T}\}\right]\\
    =&x_t\Pr(p_{t+1}(\gD_T)\cdots,p_T(\gD_T)\leq \textcolor{black}{x_T}).
\end{align*}

Repeating the above argument recursively for $s=t,t+1,\ldots,T$, and conditioning at step $s$ on
$\Phi_{s}\cup\{(\X_i,Y_i)\}_{i=s}^{T}$, we obtain
\[
\Pr\bigl(p_t\leq x_t,\, p_{t+1}\leq x_{t+1},\dots, p_T\leq x_T\bigr)
= x_t\, x_{t+1}\, \cdots\, x_T.
\]
Since the joint cumulative distribution function factors as the product of the marginals, it follows that the sequence $\{p_t,\dots,p_T\}$ is mutually independent, with each $p_t$ marginally distributed as $\mathrm{Uniform}(0,1)$.

\medskip
Thus, we have demonstrated that under the exchangeability assumption and the online updating scheme, the online conformal $p$-values defined in \eqref{eq:conf_p-appen} are mutually independent. 
\end{proof}



\subsection{Proof of Theorem \ref{the:FDR-OCTF}: finite sample mFDR control}\label{proof:mFDR-OCTF}
\begin{proof}
    \textbf{mFDR control for Algorithm \ref{alg:OCTF-ind-conservative}.} 
We first prove the results for LFS. 
Denote $\Psi_t=\Big( (p_k:k\in\gC_{0t}), (\theta_k:k< t),\big((X_k,Y_k):k\in\gC_{1t}\big)\Big)$. 
We need to verify two facts, for any time $t\in[T]$, 

(i):  $\alpha_t^{\text{LFS}}$ is fixed given $\Psi_t$. Since $\Psi_t$ contains all non-null information available up to time \(t\), we can regard the p-values \(\{p_i\}_{i\leq t}\) as being constructed without using any non-null data; the general case then follows directly by conditioning on $\Psi_t$ as non-null data are all fixed. Note that $\alpha_t^{\text{LFS}}$ is fixed given all past null decisions $(\delta_i:1\leq i< t)$ and a null decision $\delta_i$ is decided by $( (p_k,\alpha_k):\theta_k=0,1\leq k\leq i)$. By iterated discussion, $\alpha_t$ is decided by null p-values $(p_k:k\in\gC_{0t})$, thereby determined fully by $\Psi_t$.

(ii): $p_t$ is super-uniform given $\Psi_t$ and $\theta_t=0$. This is direct as long as $p_t$ is independent of past null p-values given non-null data, which is verified by Proposition \ref{prop-inde-p}.

Then we have
 \textcolor{black}{\begin{eqnarray}
    \E[V(t)]&=&\E\left[\sum_{j\leq t}(1-\theta_j)\mathbb{I}\{p_j\leq \alpha_j\}\right] \nonumber \\
    &{=}& \sum_{j\leq t}\E\left[(1-\theta_j)\E\left[\mathbb{I}\{p_j\leq \alpha_j\}\mid \Psi_j,\theta_j=0\right]\right]   \nonumber \\
    &\Eqmark{i}{\leq}& \E\left[\sum_{j\leq t}\alpha_j(1-\theta_j)\right] \nonumber \\
    &\Eqmark{ii}{\leq}& \E\left[\sum_{j=1}^{t-1}\alpha_j(1-\theta_j)+\alpha_t\right] \nonumber \\
   &\Eqmark{iii}{\leq}&  \alpha \E\left[1\vee\sum_{j\leq t}\delta_j(1-\theta_j)\right]   \nonumber \\
    &\Eqmark{iv}{\leq}& \alpha \cdot \E[1\vee R(t)],
\end{eqnarray}
}
where the inequality (i) follows from the above two facts, thereby $\Pr(p_j\leq \alpha_j\mid \Psi_j,\theta_j=0)\leq \alpha_j$. And the inequality (ii) follows by the fact that $1-\theta_t\leq 1$. For inequality (iii) follows from the construction of $\alpha_t^{\text{LFS}}$, and the inequality (iv) holds since $1-\theta_j\leq 1$ for all $j\leq t$. Therefore, we conclude that $\mFDR(t)\leq \alpha$ for LFS.

Similar results for SFS can also be proved as follows:
\begin{eqnarray}
    \E[V(t)]&=&\E\left[\sum_{j\leq t}(1-\theta_j)\mathbb{I}\{p_j\leq \alpha_j\}\right] \nonumber \\
    &{=}& \sum_{j\leq t}\E\left[(1-\theta_j)\E\left[\mathbb{I}\{p_j\leq \alpha_j\}\mid \Psi_j,\theta_j=0\right]\right]   \nonumber \\
    &\leq& \E\left[\sum_{j\leq t}\alpha_j(1-\theta_j)\right] \nonumber \\
     &\leq& \E\left[\sum_{j=1}^{t-1}\alpha_j(1-\theta_j)\frac{\mathbb{I}\{p_j>\lambda\}}{(1-\lambda)}+\alpha_t(1-\theta_t) \frac{\mathbb{I}\{p_t>\lambda\}}{(1-\lambda)}.\right] \nonumber \\
    &\leq& \E\left[\sum_{j=1}^{t-1}\alpha_j(1-\theta_j)\frac{\mathbb{I}\{p_j>\lambda\}}{(1-\lambda)}+\alpha_t \frac{\mathbb{I}\{p_t>\lambda\}}{(1-\lambda)}\right] \nonumber \\
    &\leq& \alpha \E\left[1\vee\sum_{j\leq t}\delta_j(1-\theta_j)\right]  \nonumber \\
    &\leq& \alpha\cdot \E[1\vee R(t)] . \nonumber 
\end{eqnarray}
Therefore, we obtain the conclusion that $\operatorname{mFDR}(t)\leq \alpha$ for the proposed OCTF procedures in Algorithm \ref{alg:OCTF-ind-conservative} with $\alpha_t=\alpha_t^{\text{LFS}}$ or $\alpha_t=\alpha_t^{\text{SFS}}$.
\end{proof}

\subsection{Proof of Corollary \ref{the:FDR-Opt}}
\begin{proof}
    It suffices to verify that the optimized score function \( V(\cdot;\hat{k}_t) \) is symmetric with respect to \( \gC_{0t} \cup \{t\} \), so that Proposition~\ref{prop:online_conf_val_exch} and Theorem~\ref{the:FDR-OCTF} can be directly applied to establish the validity of mFDR control at the target level for Algorithm Opt-OCTF.

For any permutation $\sigma\in\Omega_t$ that only permutes the indices in $\gC_{0t}\cup\{t\}$, we have $$\hat{k}_t^\sigma=\underset{k \in [K]}{\arg\min}\;\gM^{\mathrm{EWMA}}_t((\gD_T)_\sigma;k)=\hat{k}_t.$$
To see why, the auxiliary p-value for $j\in\gC_{1t}$ after permutation $\sigma$ is
\begin{align*}
(\tilde{p}_{k,j})_\sigma&=\frac{\sum_{s\in\gC_{0t}\cup\{t\}}\mathbb{I}\{V(\mathbf{X}_{\sigma(s)};k)\leq V(\mathbf{X}_j;k)\}}{1+|\gC_{0t}|}\\
&=\frac{\sum_{s\in\gC_{0t}\cup\{t\}}\mathbb{I}\{V(\mathbf{X}_{s};k)\leq V(\mathbf{X}_j;k)\}}{1+|\gC_{0t}|}=\tilde{p}_{k,j}.
\end{align*}
This means $\{\tilde{p}_{k,j}\}_{j\in\gC_{1t}}$ is permutation invariant to $\sigma$. Applying this, we have
\begin{equation*}
\gM^{\mathrm{EWMA}}_t((\gD_t)_\sigma;k)=\frac{\sum_{j=1}^{t-1} \rho^{t-1-j}\, (\tilde{p}_{k,j})_\sigma\theta_j\,}{\sum_{j=1}^{t-1}\rho^{t-1-j}\cdot \theta_j}=\frac{\sum_{j=1}^{t-1} \rho^{t-1-j}\, \tilde{p}_{k,j}\theta_j\,}{\sum_{j=1}^{t-1}\rho^{t-1-j}\cdot \theta_j}=\gM^{\mathrm{EWMA}}_t(\gD_t;k),
\end{equation*}
which keeps invariant to the permutation $\sigma$. Combining together, we have $\hat{k}_t$ is symmetric to $\gC_{0t}\cup\{t\}$. And it indicates that $V(\cdot;\hat{k}_t)$ is symmetric to $\gC_{0t}\cup\{t\}$.
\end{proof}

\subsection{Proof of Theorem \ref{the:FDR_e-GAIF}}
\begin{proof}
{\color{black}We first prove the results for $e$-GAIF. Since $\delta_j=\mathbb{I}\{e_j\geq 1/\alpha_j\}$, by definition, we have
\begin{eqnarray}
    \operatorname{FDR}(t)&=&\E\left[\frac{\sum_{j\leq t}\mathbb{I}\{e_j\geq 1/\alpha_j\}(1-\theta_j)}{1\vee \sum_{j\leq t}\delta_j}\right] \nonumber \\
    ~&\leq&\E\left[\sum_{j\leq t}\frac{(1-\theta_j)\mathbb{I}\{1/e_j\leq \alpha_j\}}{R({j-1})+ 1}\right] \nonumber \\
    ~&\leq& \E\left[\sum_{j\leq t}\frac{(1-\theta_j)e_j\alpha_j}{R({j-1})+ 1}\right] \nonumber \\
     ~&=& \E\left[\sum_{j\leq t}\frac{(1-\theta_j)\E[e_j\mid \gF_{j-1}]\alpha_j}{R({j-1})+ 1}\right] \nonumber \\
   ~&\leq& \E\left[\sum_{j\leq t}\frac{(1-\theta_j)\alpha_j}{R({j-1})+1}\right] \nonumber \\
    ~&\leq&\E\left[\sum_{j\in\mathcal{I}_t}\frac{\alpha_j(1-\theta_j)}{1+R({j-1})}+\sum_{j\in\bar{\mathcal{I}}_t}\frac{\alpha_j}{1+R({j-1})}\right] \nonumber \\
    ~&=&\E\left[\widehat{\FDP}_{\text{e-GAIF}} (t)\right]\leq \alpha, \nonumber
\end{eqnarray}

where the first inequality holds since $R({j-1})+1\leq (R(t)\vee 1)$ for every $j\in\{j\leq t:\delta_j=1\}$ by definition, the second inequality holds since $\mathbb{I}\{e_j\alpha_j\geq 1\}\leq e_j\alpha_j$, the third inequality uses the law of iterated expectations by conditioning on $\mathcal{F}_{j-1}$ and then applies the property of $e$-values, and the fourth inequality holds since $1-\theta_t\leq 1$, and the last inequality follows from the construction of $e$-GAIF, which completes the proof of FDR control for $e$-GAIF.

We now proceed to establish the FDR control guarantee for the $e$-Ada-GAIF procedure. Specifically, we show that:
\begin{eqnarray}
    \operatorname{FDR}(t) &=& \E\left[\frac{\sum_{j\leq t}\mathbb{I}\{e_j\geq 1/\alpha_j\}(1-\theta_j)}{1\vee \sum_{j\leq t}\delta_j}\right] \nonumber \\
    &\leq& \E\left[\sum_{j\leq t}\frac{(1-\theta_j)\mathbb{I}\{1/e_j\leq \alpha_j\}}{R({j-1})+ 1}\right] \nonumber \\
    &\leq& \E\left[\sum_{j\leq t}\frac{(1-\theta_j)e_j\alpha_j}{R({j-1})+1}\right] \nonumber \\
    &=& \E\left[\sum_{j\leq t}\frac{(1-\theta_j)\E[e_j\mid \mathcal{J}_{j-1}]\alpha_j}{R({j-1})+1}\right] \nonumber \\
    &\leq& \E\left[\sum_{j\leq t}\frac{(1-\theta_j)\alpha_j}{R({j-1})+1}\right] \nonumber \\
    &\leq& \sum_{j\leq t}\E\left[\frac{(1-\theta_j)\alpha_j}{R({j-1})+1} \cdot \frac{\E\left[\mathbb{I}\{e_j<1/\lambda\}\mid \mathcal{J}_{j-1}\right]}{1-\lambda}\right] \nonumber \\
    &=& \sum_{j\leq t}\E\left[\frac{(1-\theta_j)\alpha_j}{R({j-1})+1} \cdot \frac{\mathbb{I}\{e_j<1/\lambda\}}{1-\lambda}\right] \nonumber \\
    &\leq& \E\left[\sum_{j\in\mathcal{I}_t}\frac{\alpha_j(1-\theta_j)}{R({j-1})+1} \cdot \frac{\mathbb{I}\{e_j\leq 1/\lambda\}}{1-\lambda} +\sum_{j\in\bar{\mathcal{I}}_t} \frac{\alpha_j}{R({j-1})+1} \cdot \frac{\mathbb{I}\{e_j\leq 1/\lambda\}}{1-\lambda}\right] \nonumber \\
    &=& \E\left[\widehat{\FDP}_{\text{e-Ada-GAIF}}(t)\right] \leq \alpha, \nonumber 
\end{eqnarray}
where the first inequality holds since $R({j-1})+1\leq (R(t)\vee 1)$ for every $j\in\{j\leq t:\delta_j=1\}$ by definition, the second inequality holds since $\mathbb{I}\{e_j\alpha_j\geq 1\}\leq e_j\alpha_j$, the third inequality uses the law of iterated expectations by conditioning on $\mathcal{F}_{j-1}$ and then applies the property of $e$-values, and the fourth inequality holds since $\E\left[\mathbb{I}\{e_j<1/\lambda\}\mid \mathcal{J}_{j-1}\right]\geq 1-\lambda$ by the property of $e$-values, the fifth inequality follows from $1-\theta_t\leq 1$, and the last inequality follows from the construction of $e$-Ada-GAIF, which completes the proof of FDR control for $e$-Ada-GAIF.}
\end{proof}

{\color{black}\subsection{Proof of Corollary \ref{the:FDR-OCTF-delay}}

\begin{proof}
{\color{black}
We prove the result in two steps.

\paragraph{Step 1: mFDR control within each sub-stream.}
Fix a sub-stream index $j\in\{1,2,\ldots,d+1\}$ and consider the sequence of testing times 
$\{t_k\}_{k\ge 1}$ such that $t_k\in S_j$. By construction of the sub-streams, the feedback 
associated with test $t_k$ is revealed at time $t_{k+1}$, which belongs to the same sub-stream. 
Therefore, when conditioning on the sub-stream index $j$, the testing procedure restricted to 
$\{t_k\}_{k\ge 1}$,  using an adjusted level $\alpha/(d+1)$, has fully observed feedback and coincides with the standard OCTF procedure 
(without delay).

Under Assumption~\ref{assump:conformal_setting}, the conformal $p$-values $\{p_{t_k}\}_{k\ge 1}$ 
computed using the calibration sets $\{\gC_{0t}^{\,j}\}$ are super-uniform under the null, 
conditional on the past within the same sub-stream. Moreover, the test levels 
$\{\alpha_{t_k}\}_{k\ge 1}$ are updated based only on information available prior to $t_k$ 
within the sub-stream. Hence, by the proof of finite-sample mFDR guarantee for OCTF in Appendix \ref{proof:mFDR-OCTF}, the rejection set 
restricted to sub-stream $S_j$, denoted by $\gR^{(j)}=\gR\cap S_j$, satisfies
\[
\mathbb E[V^{(j)}(t)]
\le \frac{\alpha}{d+1}\,\mathbb E[1\vee R^{(j)}(t)].
\]
and thus
\[
\mFDR^{(j)}(t) \le \alpha/(d+1)\le\alpha \quad \text{for all } t\in S_j.
\]




\paragraph{Step 2: Aggregation across sub-streams.}

Let $V(t)$ and $R(t)$ denote the numbers of false discoveries and total
rejections up to time $t$, respectively. Since the sub-streams
$\{S_j\}_{j=1}^{d+1}$ form a partition of $\mathbb N$,
\[
V(t)=\sum_{j=1}^{d+1}V^{(j)}(t),
\qquad
R(t)=\sum_{j=1}^{d+1}R^{(j)}(t).
\]
By Step~1,
\[
\mathbb E[V(t)]
=
\sum_{j=1}^{d+1}\mathbb E[V^{(j)}(t)]
\le
\frac{\alpha}{d+1}
\sum_{j=1}^{d+1}\mathbb E[1\vee R^{(j)}(t)].
\]

Since \(R^{(j)}(t)\le R(t)\) pointwise and \(1\vee(\cdot)\) is monotone,
\[
1\vee R^{(j)}(t)\le 1\vee R(t),
\]
for every \(j\). Hence,
\[
\sum_{j=1}^{d+1}\mathbb E[1\vee R^{(j)}(t)]
\le
(d+1)\mathbb E[1\vee R(t)].
\]

Therefore,
\[
\mathbb E[V(t)]
\le
\alpha\mathbb E[1\vee R(t)],
\]
which implies
\[
\mFDR(t)
=
\frac{\mathbb E[V(t)]}
{\mathbb E[R(t)\vee1]}
\le
\alpha.
\]
}

\end{proof}

\subsection{Proof of Theorem \ref{the:opt-FDR-OCTF}}
\begin{proof}
For simplicity, we define that when $\theta_t=1$, 
$$\hat{M}_{k,t}:=\gM^{\mathrm{EWMA}}_t(\gD_t;k)=\frac{\sum_{j\in\gC_{1t}} \rho^{t-1-j}\, \tilde{p}_{k,j}\,}{\sum_{j\in\gC_{1t}}\rho^{t-1-j}}\quad\text{and}\quad M_{k,t}:= \E[F^{(0)}_{k}(V(\mathbf{X}_t;k))\mid\theta_t=1].$$
And we set historical data $\gD_\gC$ as empty set for notational convenience. The final guarantee will remain valid by replacing $t$ with $n+t$ when considering historical samples.

We decompose the error by
\begin{align}\label{eq:error_decom_EWMA}
 \hat{M}_{k,t}-{M}_{k,t}=&\underbrace{\frac{\sum_{j\in\gC_{1t}} \rho^{t-1-j}\, (\tilde{p}_{k,j}-F^{(0)}_{k}(V(\X_j;k)))\,}{\sum_{j\in\gC_{1t}}\rho^{t-1-j}}}_{\text{(I) null CDF estimation}}\nonumber\\
 &+\underbrace{\frac{\sum_{j<t} \rho^{t-1-j}\theta_t\, (F^{(0)}_{k}(V(\X_j;k))-\E[F^{(0)}_{k}(V(\X_j;k))\mid\theta_j=1])\,}{\sum_{j\in\gC_{1t}}\rho^{t-1-j}}}_{\text{(II) gap of EWMA estimation}}\nonumber\\
 &+\underbrace{\frac{\sum_{j<t} \rho^{t-1-j}\theta_j\, (\E[F^{(0)}_{k}(V(\X_j;k))\mid\theta_j=1]-\E[F^{(0)}_{k}(V(\X_t;k))\mid\theta_t=1])\,}{\sum_{j\in\gC_{1t}}\rho^{t-1-j}}}_{\text{(III) drift bias}}.  
\end{align}
Then we check the three term one by one.

\textbf{(I) Null CDF estimation.} 
Firstly, we condition on the past states $(\theta_s)_{s<t}$ and $\{\theta_t=1\}$ to give results. In this case, $\gC_{0t}$ and $\gC_{1t}$ are fixed.
Denote 
$$\hat{F}_{k,t}^{(0)}(v)=\frac{\sum_{s \in \gC_{0t} \cup \{t\}} \mathbb{I}\{V(\mathbf{X}_s; k) \leq v\}}{1 + |\gC_{0t}|}.$$
On the event $\{\theta_t=1\}$, the score $V(\X_t;k)$ is drawn from the non-null distribution and thus is not a null
sample for estimating $F_0^k$. Let
\[\hat{F}_{k,t}^{(0),\text{null}}
(v):=\frac{1}{|\gC_{0t}|}\sum_{s \in \gC_{0t}} \mathbb{I}\{V(\mathbf{X}_s; k) \leq v\}.
\]
Then we can decompose $\hat{F}_{k,t}^{(0)}$ as
\[
\hat{F}_{k,t}^{(0)}(v)=\frac{|\gC_{0t}|}{|\gC_{0t}|+1}\,\hat{F}_{k,t}^{(0),\text{null}}(v)
+\frac{1}{|\gC_{0t}|+1}\,\mathbb{I}\{V(\X_t;k)\le v\}.
\]
Therefore, for all $v$,
\begin{align*}
\big|\hat{F}_{k,t}^{(0)}(v)-F_{k}^{(0)}(v)\big|
&\le \frac{|\gC_{0t}|}{|\gC_{0t}|+1}\,\big|\hat{F}_{k,t}^{(0),\text{null}}(v)-F_{k}^{(0)}(v)\big|
+\frac{1}{|\gC_{0t}|+1}\,\big|\mathbb{I}\{V(\X_t;k)\le v\}-F_{k}^{(0)}(v)\big|\\
&\le \sup_u \big|\hat{F}_{k,t}^{(0),\text{null}}(u)-F_{k}^{(0)}(u)\big|+\frac{1}{|\gC_{0t}|+1},
\end{align*}
since $\big|\mathbb{I}\{V(\X_t;k)\le v\}-F_{k}^{(0)}(v)\big|\le 1$.

Conditional on $(\theta_s)_{s<t}$, the collection $\{V(\X_s;k)\}_{s\in\gC_{0t}}$ are i.i.d. from $F_{k}^{(0)}$.
By the DKW inequality, for any $\varepsilon_1>0$,
\[
\Pr\!\left(\sup_v\big|\hat{F}_{k,t}^{(0),\text{null}}(v)-F_{k}^{(0)}(v)\big|\leq \varepsilon_1\mid (\theta_s)_{s<t},\theta_t=1\right)
\geq 1-2\exp\left\{-2|\gC_{0t}|\varepsilon_1^2\right\}.
\]
Combining the two displays yields
\[
\Pr\!\left(\sup_v\big|\hat{F}_{k,t}^{(0)}(v)-F_{k}^{(0)}(v)\big|\leq \varepsilon_1+\frac{1}{|\gC_{0t}|+1}\mid (\theta_s)_{s<t},\theta_t=1\right)
\geq 1-2\exp\left\{-2|\gC_{0t}|\varepsilon_1^2\right\}.
\]
As $\tilde p_{k,j}=\hat{F}_{k,t}^{(0)}\!\left(V(\X_j;k)\right)$, we further have
\[
\Pr\!\left(\sup_{j\in\gC_{1t}}\big|\tilde{p}_{k,j}-F_k^{(0)}(V(\mathbf{X}_j;k))\big|
\leq \varepsilon_1+\frac{1}{|\gC_{0t}|+1}\mid (\theta_s)_{s<t},\theta_t=1\right)
\geq 1-2\exp\left\{-2|\gC_{0t}|\varepsilon_1^2\right\}.
\]
Since $|\gC_{0t}|=\sum_{s=1}^{t-1}\mathbb{I}\{\theta_s=0\}$ and $\Pr(\theta_s=0)\ge \underline{\pi}$,
by Hoeffding's inequality,
\[
\Pr\Big(|\gC_{0t}|\ge (t-1)\underline{\pi}/2\Big)=\Pr\Big(|\gC_{0t}|-(t-1)\underline{\pi}\ge -(t-1)\underline{\pi}/2\Big)\ge 1-\exp\left\{-(t-1)\underline{\pi}^2/2\right\}.
\]
Therefore, on the event $\P_t=\big\{|\gC_{0t}|\ge (t-1)\underline{\pi}/2\big\}$, we have
$$\Pr\!\left(\sup_{j\in\gC_{1t}}\big|\tilde{p}_{k,j}-F_k^{(0)}(V(\mathbf{X}_j;k))\big|
\leq \varepsilon_1+\frac{2}{(t-1)\underline{\pi}}\mid \P_t,\theta_t=1\right)
\geq 1-2\exp\left\{-\{(t-1)\underline{\pi}\}\varepsilon_1^2\right\}.$$
Then by the independence of $\P_t$ and $\{\theta_t=1\}$, invoking Lemma \ref{lem:decondition_event} leads to
\begin{align}\label{appen_eq:bound_nullCDF}
&\Pr\left(\left|\frac{\sum_{j\in\gC_{1t}} \rho^{t-1-j}\, (\tilde{p}_{k,j}-F_k^{(0)}(V(\X_j;k)))\,}{\sum_{j\in\gC_{1t}}\rho^{t-1-j}}\right|\leq\varepsilon_1+\frac{2}{(t-1)\underline{\pi}}\mid \theta_t=1 \right)\nonumber\\
\geq& 1-2\exp\left\{-\{(t-1)\underline{\pi}\varepsilon_1^2\right\}-\exp\left\{-(t-1)\underline{\pi}^2/2\right\}.
\end{align}

\textbf{(II) Gap of EWMA estimation.} 
Define $\Omega_t=\sum_{j\in\gC_{1t}}\rho^{t-1-j}=\sum_{j=1}^{t-1}\rho^{t-1-j}\theta_j$. It follows that $\E[\Omega_t]=\sum_{j=1}^{t-1}\rho^{t-1-j}\E[\theta_j]\geq (1-\overline{\pi})(1-\rho^{t-1})/(1-\rho)$. By Hoeffding inequality, we have
\begin{align}\label{appen_eq:hoeff_ewma}
    \Pr\left(\Omega_t\geq\frac{1-\overline{\pi}}{2}\frac{1-\rho^{t-1}}{1-\rho}\right)\geq1-\exp\left\{-\frac{(1-\overline{\pi})^2(1-\rho^2)}{2(1-\rho^{2(t-1)})}\right\}.
\end{align}

Then we can apply the Hoeffding inequality to the weighted sum as $$\left\{\theta_j\big(F_k^{(0)}(V(\X_j;k))-\E[F_k^{(0)}(V(\X_j;k))\mid\theta_j=1]\big)\right\}_{j< t}$$ are bounded between $[0,1]$ and independent:

\begin{align}\label{appen_eq:EWMA_gap}
    &\Pr\left(\left|\sum_{j<t} \rho^{t-1-j}\theta_j\, (F^{(0)}_{k}(V(\X_j;k))-\E[F^{(0)}_{k}(V(\X_j;k))\mid\theta_j=1])\right|\leq\varepsilon_2\Omega_t\right)\nonumber\\
    \geq& \Pr\left(\left|\sum_{j<t} \rho^{t-1-j}\theta_j\, (F^{(0)}_{k}(V(\X_j;k))-\E[F^{(0)}_{k}(V(\X_j;k))\mid\theta_j=1])\right|\leq\varepsilon_2\Omega_t,\Omega_t\geq\frac{1-\overline{\pi}}{2}\frac{1-\rho^{t-1}}{1-\rho}\right)\nonumber\\
    \geq& 1-\Pr\left(\left|{\sum_{j<t} \rho^{t-1-j}\theta_j\, (F_k^{(0)}(V(\X_j;k))-\E[F_k^{(0)}(V(\X_j;k))\mid\theta_j=1])\,}\right|\geq\varepsilon_2\left\{\frac{1-\overline{\pi}}{2}\frac{1-\rho^{t-1}}{1-\rho}\right\}\right)\nonumber\\
    &-\Pr\left(\Omega_t\leq\frac{1-\overline{\pi}}{2}\frac{1-\rho^{t-1}}{1-\rho}\right)\nonumber\\
    \geq& 1-2\exp\left\{-\frac{2\varepsilon_2^2(1-\rho^{2})\left\{\frac{1-\overline{\pi}}{2}\frac{1-\rho^{t-1}}{1-\rho}\right\}^2}{1-\rho^{2(t-1)}}\right\}-\exp\left\{-\frac{(1-\overline{\pi})^2(1-\rho^2)}{2(1-\rho^{2(t-1)})}\right\}.
\end{align}

\textbf{(III) Drift bias.} 
By Stieltjes integrals, 
\begin{align*}
 \left|\E[F_k^{(0)}(V(\X_j;k))\mid\theta_j=1]-\E[F_k^{(0)}(V(\X_t;k))\mid\theta_t=1]\right|&=|\int F_k^{(0)}(v)dF_{k,j}^{(1)}(v)-\int F_k^{(0)}(v)dF_{k,t}^{(1)}(v)|\\
 &=|\int (F_{k,j}^{(1)}(v)-F_{k,t}^{(1)}(v))dF_k^{(0)}(v)|\\
 &\leq \|F_{k,j}^{(1)}-F_{k,t}^{(1)}\|_\infty.
\end{align*}
So relating it with our slowly varying distribution condition leads to
\begin{align*}
   & \left|\frac{\sum_{j\in\gC_{1t}} \rho^{t-1-j}\, (\E[F_k^{(0)}(V(\X_j;k))\mid\theta_j=1]-\E[F_k^{(0)}(V(\X_t;k))\mid\theta_t=1])\,}{\sum_{j\in\gC_{1t}}\rho^{t-1-j}}\right| \\
   &\leq \gamma\frac{\sum_{j=1}^{t-1}\rho^{t-1-j}(t-j)}{\Omega_t}\leq \gamma\frac{1}{\Omega_t(1-\rho)^2}.
\end{align*}
The last inequality comes from the fact that $\sum_{j=1}^{t-1}(t-j)\rho^{t-1-j}=\{1-t\rho^{t-1}+(t-1)\rho^t\}/(1-\rho)^2\leq 1/(1-\rho)^2$.


Since we have already analyzed the high-probability lower bound of $\Omega_t$ in \eqref{appen_eq:hoeff_ewma}, it follows that
\begin{align}\label{appen_eq:slow_drift}
\Pr\left(|\bar{\Delta}_t|
\le 2\gamma\{(1-\overline{\pi})(1-\rho^{t-1})(1-\rho)\}^{-1}\right)
\ge 1-\exp\left\{-\frac{(1-\overline{\pi})^2(1-\rho^2)}
{2(1-\rho^{2(t-1)})}\right\},
\end{align}
where
\[
\bar{\Delta}_t := \frac{\sum_{j\in\gC_{1t}} \rho^{t-1-j}\,
\Big(\E[F_k^{(0)}(V(\X_j;k))\mid \theta_j=1]
- \E[F_k^{(0)}(V(\X_t;k))\mid \theta_t=1]\Big)}
{\sum_{j\in\gC_{1t}} \rho^{t-1-j}}.
\]

 Combining the results in \eqref{appen_eq:bound_nullCDF}, \eqref{appen_eq:EWMA_gap} and \eqref{appen_eq:slow_drift} together into \eqref{eq:error_decom_EWMA}, we have for any $\varepsilon_1\in(0,\sqrt{\underline{\pi}/2})$ and $\varepsilon_2\in(0,1-\rho)$ 
 $$ |\hat{M}_{k,t}-{M}_{k,t}|\leq \varepsilon_1+\varepsilon_2+2\gamma\{(1-\overline{\pi})(1-\rho^{t-1})(1-\rho)\}^{-1}+2\{(t-1)\underline{\pi}\}^{-1}$$
with probability
$$1-3\exp\left\{-\{(t-1)\underline{\pi}\varepsilon_1^2\right\}-4\exp\left\{-\frac{\varepsilon_2^2(1+\rho)({1-\overline{\pi}})^2({1-\rho^{t-1}})}{2(1+\rho^{t-1})(1-\rho)}\right\}.$$

    Suppose $t>1+1/(1-\rho)$ such that $\rho^{t-1}<1/2$, then $$\frac{(1+\rho)({1-\rho^{t-1}})}{2(1+\rho^{t-1})(1-\rho)}\geq \frac{1}{3}\frac{1+\rho}{1-\rho}.$$
    Then the part (i) is verified. 

   For part (ii), 
$$\max_{k\in[K]} |\hat{M}_t^k-{M}_t^k|\leq \varepsilon_1+\varepsilon_2+4\gamma\{(1-\overline{\pi})(1-\rho)\}^{-1}+2\{(t-1)\underline{\pi}\}^{-1}$$
    holds with probability 
$$1-3K\exp\left\{-\{(t-1)\underline{\pi}\varepsilon_1^2\right\}-4K\exp\left\{-\frac{\varepsilon_2^2({1-\overline{\pi}})^2}{6(1-\rho)}\right\}.$$

By Assumption \ref{assum:margin},
$$\Pr(\hat{k}_t\neq k^*_t\mid\theta_t=1)\leq\Pr(\max_{k\in[K]}|\hat{M}_{k,t}-{M}_{k,t}|\geq c/2).$$

To enforce $\max_k|\hat M_{k,t}-M_{k,t}|\le c/2$,
it suffices to let 
\begin{equation*}
\varepsilon_1=\varepsilon_2 \le\frac{c}{8} ,\;\;\frac{4\gamma}{(1-\overline\pi)(1-\rho)}\leq \frac{c}{8},\;\;\frac{2}{(t-1)\underline\pi}\leq \frac{c}{8}.
\end{equation*}
So it is required that
$$t\geq 1+\frac{16}{c\underline{\pi}},\quad \frac{\gamma}{1-\rho}\leq \frac{(1-\overline{\pi})c}{32}.$$
And we have the explicit bound
\begin{align}\label{eq:misselect_deltafree}
\Pr(\hat k_t\neq k_t^*)
\le\;&
3K\exp\!\left\{-(t-1)\underline\pi\frac{c^2}{64}\right\}+
4K\exp\!\left\{-\frac{c^2(1-\overline\pi)^2}{384(1-\rho)}\right\}.\nonumber
\end{align}
And this bound can be directly generalized  when the historical data are not empty by replacing $t$ with $n+t$. 

Finally, we can jointly guarantee the score selection consistency for all non-null points:
\begin{align*}
    &\Pr(\exists i\le t:\theta_i=1,\ \hat k_i\neq k_i^*)\leq\sum_{i=1}^t\Pr(\theta_i=1)\Pr(\hat k_i\neq k_i^*\mid\theta_i=1)\\
    \leq&  3K(1-\underline{\pi})\sum_{i=1}^t \exp\!\left\{-(n+i-1)\underline\pi\frac{c^2}{64}\right\}
+4K(1-\underline{\pi})t\exp\!\left\{-\frac{c^2(1-\overline\pi)^2}{384(1-\rho)}\right\}\\
\leq&\frac{3K(1-\underline{\pi})\exp\!\left\{-n\underline\pi{c^2}/{64}\right\}}{1-\exp\!\left\{-\underline\pi{c^2}/{64}\right\}}
+4K(1-\underline{\pi})t\exp\!\left\{-\frac{c^2(1-\overline\pi)^2}{384(1-\rho)}\right\}.
\end{align*}

\end{proof}

\subsection{Proof of Theorem \ref{the:FDR-Opt-perm}}

\begin{proof}
It suffices to prove Proposition \ref{prop:online_conf_val_exch} for the permuted $p$-value. Naturally, we can view score selection as part of the score construction such that $V(\cdot;\hat{k}_t)$ can be defined as $V(\cdot;\gD_t)$. Here we do not require $V(\cdot;\gD_t)$ is symmetric to past null data and allow it perform non-symmetrically. For convenience, we set $\theta_t=\theta_{t+1}=\cdots=\theta_T=0$.

For any $\sigma'\in\Omega_t$ and given $\{\X_i : i\in \gC_{0t}\cup\{t\}\}=\{\x_i : i\in \gC_{0t}\cup\{t\}\}$ as the realizations, define $\bar{\sigma}=\sigma'\cdot\sigma$ and
\begin{align*}
    &p_{t}^{\rm opt}((\gD_T)_{\sigma'})\\
    =&\frac{1}{|\Omega_t|}\sum_{\sigma\in\Omega_t} \I\{V(\x_{{\sigma'(\sigma(i))}};((\gD_t)_\sigma)_{\sigma'})<V(\x_{\sigma'(t)};(\gD_t)_{\sigma'})\}\\
    &+\frac{1}{|\Omega_t|}\xi_t\cdot \sum_{\sigma\in\Omega_t} \mathbb{I}\{V(\x_{\sigma'(\sigma(i))};((\gD_t)_\sigma)_{\sigma'}))=V(\x_{\sigma'(t)};(\gD_t)_{\sigma'})\}\\
=&\frac{\sum_{\bar{\sigma}\in\Omega_t} \I\{V(\x_{{\bar{\sigma}(i)}};(\gD_t)_{\bar{\sigma}})<V(\x_{\sigma'(t)};(\gD_t)_{\sigma'})\}+\xi_t\cdot \sum_{\bar{\sigma}\in\Omega_t} \mathbb{I}\{V(\x_{\bar{\sigma}(i)};(\gD_t)_{\bar{\sigma}}))=V(\x_{\sigma'(t)};(\gD_t)_{\sigma'})\}}{|\Omega_t|}\\
    =&\frac{\sum_{{\sigma}\in\Omega_t} \I\{V(\x_{{{\sigma}(i)}};(\gD_t)_{{\sigma}})<V(\x_{\sigma'(t)};(\gD_t)_{\sigma'})\}+\xi_t\cdot \sum_{\sigma\in\Omega_t} \mathbb{I}\{V(\x_{{\sigma}(i)};(\gD_t)_{{\sigma}}))=V(\x_{\sigma'(t)};(\gD_t)_{\sigma'})\}}{|\Omega_t|}.
\end{align*}
The last equality comes from the fact that $\sigma'\in\Omega_t$ is a bijection such that $\{\sigma:\sigma\in\Omega_t\}=\{\sigma'\cdot\sigma:\sigma\in\Omega_t\}$.

Follow the same discussion in the proof of Proposition \ref{prop:online_conf_val_exch}.
Let $\{X_i:i\in \gC_{0t}\cup\{t\}\}=\{x_i:i\in \gC_{0t}\cup\{t\}\}$ as a set of realizations and 
$$q=Q_\alpha(V(x_{\sigma(t)};(\gD_t)_\sigma):\sigma\in\Omega_t).$$

For any $\sigma'\in\Omega_t$ such that $V(\x_{\sigma'(t)};(\gD_t)_{\sigma'})>q$, we have
$$p_{t}((\gD_T)_{\sigma'})\geq \frac{N_=+N_-}{|\Omega_t|}\geq \alpha;$$
Here
$$N_-=\sum_{{\sigma'}\in\Omega_t}\mathbb{I}\{(V(x_{{\sigma'}(t)};(\gD_t)_{\sigma'})<q\},\quad N_{=}=\sum_{{\sigma'}\in\Omega_t}\mathbb{I}\{V(x_{{\sigma'}(t)};(\gD_t)_{\sigma'})=q\}.$$

For any ${\sigma'}\in\Omega_t$ such that $V(\x_{{\sigma'}(t)};(\gD_t)_{\sigma'})<q$, we have
$$p_{t}((\gD_T)_{\sigma'})\leq  \frac{N_-}{|\Omega_t|}\leq \alpha.$$

Hence
\begin{align*}
    \Pr(p_t\leq\alpha\mid\Phi_t)&=\frac{1}{|\Omega_t|}\sum_{{\sigma'}\in\Omega_t}\Pr(p^{\rm opt}_t((\gD_T)_{\sigma'})\leq\alpha)\\
    &{=}\frac{1}{|\Omega_t|}\sum_{{\sigma'}\in\Omega_t}\left(\mathbb{I}\{V(\x_{{\sigma'}(t)};(\gD_t)_{\sigma'})<q\}+\mathbb{I}\{V(\x_{{\sigma'}(t)};\gD_t)=q\}\frac{\alpha|\Omega_t|-N_-}{N_=}\right)\\
    &{=}\frac{1}{|\Omega_t|}\left(N_-+N_=\frac{\alpha|\Omega_t|-N_-}{N_=}\right)=\alpha,
\end{align*}

As for the independence of null $p$-values, by a similar discussion in Proposition \ref{prop:online_conf_val_exch}, we can still have 
\[p_{t'}(\gD_T)=p_{t'}((\gD_T)_{\sigma'})\]
for all $t'\in\{t+1,\dots,T\}$ and any $\sigma'\in\Omega_t$.  Then by the same iterate conditional expectation without relying on additional symmetry discussion, we can provide the independence results.

\end{proof}

}

{\color{black}

\subsection{Proof of Proposition \ref{prop:LF_fdphat_control}}

\begin{proof}
The inequality $\widehat{\mathrm{FDP}}_{\mathrm{GAIF}}(t) \le \alpha$ is equivalent to
\begin{equation}\label{eq:wealth_target_gaif}
    \sum_{m=1}^{t} \alpha_m \;\le\; \alpha \bigl(1 \vee R(t)\bigr) + \sum_{j \in \mathcal{I}_t} \theta_j\, \alpha_j,
\end{equation}
where $R(t) = \sum_{m=1}^t \delta_m$.
Since 
\[
    {\alpha}_m = \gamma_m s_0 + (\alpha - s_0)\,\gamma_{m - \tau_1}\,\mathbb{I}\{\tau_1 < m\}
    + \alpha \sum_{k \ge 2:\,\tau_k < m} \gamma_{m - \tau_k}
    + \sum_{j \in \mathcal{I}_m} \theta_j\,\alpha_j\,\gamma_{m - j},
\]
summing over $m = 1, \dots, t$ and exchanging the order of summation, we bound each inner sum separately.

\medskip
\noindent\textbf{Base wealth terms.}
For any fixed origin $a \in \{0\} \cup \{\tau_k : k \ge 1\}$, define
\[
    S_{\mathrm{base}}(a) \;=\; \sum_{m > a}^{t} \gamma_{m - a}.
\]
Since the index $m - a$ increments by exactly 1 at each step, \textcolor{black}{and $\gamma_i\ge 0$ for all $i$,}
\[
    S_{\mathrm{base}}(a) \;=\; \sum_{i=1}^{t - a} \gamma_i \;\le\; \sum_{i=1}^{\infty} \gamma_i \;\le\; 1.
\]
Since $R(t-1)=\#\{k:\tau_k<t\}$, applying this bound with $\alpha^{(1)} = \alpha - s_0$ and $\alpha^{(k)} = \alpha$ for $k \ge 2$ yields:
\[
\sum_{m=1}^{t}
\Bigl[
s_0\,\gamma_m
+
\sum_{k \ge 1:\,\tau_k < m}
\alpha^{(k)}\,\gamma_{m-\tau_k}
\Bigr]
\le
s_0
+
(\alpha-s_0)\,\mathbb{I}\{R(t-1)\ge 1\}
+
\alpha\,(R(t-1)-1)_+.
\]
Since
\[
s_0
+
(\alpha-s_0)\,\mathbb{I}\{R(t-1)\ge 1\}
+
\alpha\,(R(t-1)-1)_+
=
\begin{cases}
s_0, & R(t-1)=0,\\[2mm]
\alpha R(t-1), & R(t-1)\ge 1,
\end{cases}
\]
and $s_0\le \alpha$, we obtain
\[
s_0
+
(\alpha-s_0)\,\mathbb{I}\{R(t-1)\ge 1\}
+
\alpha\,(R(t-1)-1)_+
\le
\alpha\,(1\vee R(t)).
\]
Therefore,
\[
\sum_{m=1}^{t}
\Bigl[
s_0\,\gamma_m
+
\sum_{k \ge 1:\,\tau_k < m}
\alpha^{(k)}\,\gamma_{m-\tau_k}
\Bigr]
\le
\alpha\,(1\vee R(t)).
\]

\medskip
\noindent\textbf{Feedback wealth terms.}
\textcolor{black}{By definition, $\mathcal{I}_t=\{j\in[t]:\theta_j \text{ is revealed by time } t\}$; since once revealed a label stays revealed, $\{\mathcal{I}_t\}_{t\in\mathbb{N}}$ is automatically nondecreasing, i.e., $\mathcal{I}_m\subseteq\mathcal{I}_t$ for all $m\le t$. Consequently, in the double sum $\sum_{m=1}^t\sum_{j\in\mathcal{I}_m}\theta_j\alpha_j\gamma_{m-j}$, every index $j$ that appears for some $m\le t$ already belongs to $\mathcal{I}_t$. Hence, after exchanging the order of summation, it suffices to sum over $j\in\mathcal{I}_t$.}
For each $j \in \mathcal{I}_t$, define
\[
    S_{\mathrm{fb}}(j) \;=\; \sum_{m=j+1}^{t} \mathbb{I}\{j \in \mathcal{I}_m\}\,\gamma_{m - j}.
\]
Since a revealed label stays revealed, if $j \in \mathcal{I}_t$ then $\mathbb{I}\{j \in \mathcal{I}_m\} = 1$ for all $m \ge m_0$ where $m_0$ is the first time $j$ enters $\mathcal{I}_m$. Hence, using $\gamma_i\ge 0$,
\[
    S_{\mathrm{fb}}(j) \;\le\; \sum_{m=j+1}^{t} \gamma_{m-j} \;=\; \sum_{i=1}^{t-j} \gamma_i \;\le\; 1.
\]
The total feedback contribution therefore satisfies:
\[
    \sum_{j \in \mathcal{I}_t} \theta_j\,\alpha_j \cdot S_{\mathrm{fb}}(j)
    \;\le\; \sum_{j \in \mathcal{I}_t} \theta_j\,\alpha_j.
\]

\medskip
\noindent\textbf{Combining.}
Adding the two contributions:
\[
    \sum_{m=1}^{t} \alpha_m
    \;\le\; \alpha\,(1 \vee R(t)) + \sum_{j \in \mathcal{I}_t} \theta_j\,\alpha_j,
\]
which is exactly \eqref{eq:wealth_target_gaif}. Rearranging gives $\widehat{\mathrm{FDP}}_{\mathrm{GAIF}}(t) \le \alpha$.
\end{proof}

\subsection{Proof of Proposition \ref{prop:SF_fdphat_control}}

\begin{proof}
The inequality $\widehat{\mathrm{FDP}}_{\mathrm{Ada\text{-}GAIF}}(t) \le \alpha$ is equivalent to
\begin{equation}\label{eq:wealth_target}
    \sum_{m=1}^{t} \alpha_m \kappa(p_m) \;\le\; \alpha \bigl(1 \vee R(t)\bigr) + \sum_{j \in \mathcal{I}_t} \theta_j\, \alpha_j\, \kappa(p_j),
\end{equation}
where $R(t) = \sum_{m=1}^t \delta_m$ and $\kappa(p) = \frac{\mathbb{I}\{p > \lambda\}}{1-\lambda}$.

Denote
\begin{equation*}
\begin{aligned}
\tilde{\alpha}_m
&= (1-\lambda)\Bigl[
s_0\, \gamma_{m - C_{0+}}
+ (\alpha - s_0)\, \gamma_{m - \tau_1 - C_{1+}}
+ \alpha \sum_{j \geq 2} \gamma_{m - \tau_j - C_{j+}}
\Bigr] \\
&\quad + \sum_{j:\, j \in \mathcal{I}_m} 
\gamma_{m - j - C_{j+}^{*}}\, \alpha_j\, \theta_j\, \mathbb{I}\{p_j > \lambda\}.
\end{aligned}
\end{equation*}
Since $\alpha_m = \min\{\lambda, \tilde{\alpha}_m\} \le \tilde{\alpha}_m$, and $\kappa(p_m) = \frac{\mathbb{I}\{p_m > \lambda\}}{1-\lambda}$, the factor $(1-\lambda)$ in the denominator of $\kappa(p_m)$ cancels the leading $(1-\lambda)$ in $\tilde{\alpha}_m$. Hence, for each $m$:
\begin{equation}\label{eq:per_step}
\begin{aligned}
    \alpha_m \kappa(p_m) \;\le\; \mathbb{I}\{p_m > \lambda\} \Bigl[
        &\;s_0\,\gamma_{m - C_{0+}(m)}
        + (\alpha - s_0)\,\gamma_{m - \tau_1 - C_{1+}(m)}\,\mathbb{I}\{\tau_1 < m\} \\
        &+ \alpha \sum_{k \ge 2:\,\tau_k < m} \gamma_{m - \tau_k - C_{k+}(m)} \\
        &+ \sum_{j \in \mathcal{I}_m} \theta_j\,\alpha_j\,\kappa(p_j)\,\gamma_{m - j - C_{j+}^*(m)}
    \Bigr].
\end{aligned}
\end{equation}

Summing \eqref{eq:per_step} over $m = 1, \dots, t$ and exchanging the order of summation, we bound each resulting inner sum separately.

\medskip
\noindent\textbf{Base wealth terms.}
For any fixed origin $a \in \{0\} \cup \{\tau_k : k \ge 1\}$, define
\[
    S_{\mathrm{base}}(a) \;=\; \sum_{m > a}^{t} \mathbb{I}\{p_m > \lambda\}\,\gamma_{m - a - C_{a+}(m)}.
\]
Since $C_{a+}(m)$ counts the number of indices $i \in (a, m)$ with $p_i \le \lambda$, \textcolor{black}{on the event $\{p_m>\lambda\}$ the effective index satisfies}
\[
\textcolor{black}{m - a - C_{a+}(m) = \sum_{i=a+1}^{m}\mathbb{I}\{p_i > \lambda\},}
\]
\textcolor{black}{which increments by exactly 1 relative to its value at the previous occurrence of $\{p_m>\lambda\}$, and stays unchanged whenever $p_m\le\lambda$ (in which case the term is annihilated by the indicator $\mathbb{I}\{p_m>\lambda\}$ in $S_{\mathrm{base}}(a)$).} Consequently, the non-zero terms in $S_{\mathrm{base}}(a)$ map bijectively onto a contiguous prefix of $\{\gamma_i\}_{i \ge 1}$, that is, they correspond exactly to $\gamma_1,\gamma_2,\ldots,\gamma_{\Delta_{a,t}}$ without gaps. Therefore, \textcolor{black}{using $\gamma_i\ge0$,}
\[
    S_{\mathrm{base}}(a) \;=\; \sum_{i=1}^{\Delta_{a,t}} \gamma_i \;\le\; \sum_{i=1}^{\infty} \gamma_i \;\le\; 1,
\]
where $\Delta_{a,t} = \sum_{i=a+1}^{t} \mathbb{I}\{p_i > \lambda\}$.

Applying this bound with $\alpha^{(1)} = \alpha - s_0$ and $\alpha^{(k)} = \alpha$ for $k \ge 2$, we obtain
\[
\sum_{m=1}^{t}
\mathbb{I}\{p_m>\lambda\}
\Bigl[
s_0\,\gamma_{m-C_{0+}(m)}
+
\sum_{k\ge 1:\,\tau_k<m}
\alpha^{(k)}
\gamma_{m-\tau_k-C_{k+}(m)}
\Bigr]
\le
s_0
+
(\alpha-s_0)\mathbb{I}\{R(t-1)\ge 1\}
+
\alpha(R(t-1)-1)_+.
\]
Since
\[
s_0
+
(\alpha-s_0)\mathbb{I}\{R(t-1)\ge 1\}
+
\alpha(R(t-1)-1)_+
=
\begin{cases}
s_0, & R(t-1)=0,\\[2mm]
\alpha R(t-1), & R(t-1)\ge 1,
\end{cases}
\]
and $s_0\le \alpha$, it follows that
\[
s_0
+
(\alpha-s_0)\mathbb{I}\{R(t-1)\ge 1\}
+
\alpha(R(t-1)-1)_+
\le
\alpha(1\vee R(t)).
\]
Therefore,
\[
\sum_{m=1}^{t}
\mathbb{I}\{p_m>\lambda\}
\Bigl[
s_0\,\gamma_{m-C_{0+}(m)}
+
\sum_{k\ge 1:\,\tau_k<m}
\alpha^{(k)}
\gamma_{m-\tau_k-C_{k+}(m)}
\Bigr]
\le
\alpha(1\vee R(t)).
\]

\medskip
\noindent\textbf{Feedback wealth terms.}
\textcolor{black}{As in the proof of Proposition~\ref{prop:LF_fdphat_control}, $\{\mathcal{I}_t\}_{t\in\mathbb{N}}$ is nondecreasing by definition, so terms with $j\notin\mathcal{I}_t$ do not occur in the double sum, and after exchanging the order of summation it suffices to sum over $j\in\mathcal{I}_t$.}
For each $j \in \mathcal{I}_t$, define
\[
    S_{\mathrm{fb}}(j) \;=\; \sum_{m=j+1}^{t} \mathbb{I}\{p_m > \lambda\}\,\mathbb{I}\{j \in \mathcal{I}_m\}\,\gamma_{m - j - C_{j+}^*(m)}.
\]
As $m$ ranges over the non-candidate times after $j$, \textcolor{black}{on the event $\{p_m>\lambda\}$} this effective index $m - j - C_{j+}^*(m) = \sum_{i=j+1}^{m}\mathbb{I}\{p_i > \lambda\}$ increases by one each time. Hence the corresponding $\gamma$-indices are all distinct. In the full and instant feedback case, these indices form a consecutive prefix of $\{\gamma_i\}_{i\ge1}$. Under general feedback, the indicator $\mathbb I\{j\in\mathcal I_m\}$ may remove some terms before the feedback for $j$ becomes available, but it cannot duplicate any $\gamma$-index. Consequently, \textcolor{black}{using $\gamma_i\ge0$,}
\[
S_{\mathrm{fb}}(j)
\le
\sum_{i=1}^{\infty}\gamma_i
\le
1 .
\]

Hence the total feedback contribution satisfies
\[
    \sum_{j \in \mathcal{I}_t} \theta_j\,\alpha_j\,\kappa(p_j) \cdot S_{\mathrm{fb}}(j)
    \;\le\; \sum_{j \in \mathcal{I}_t} \theta_j\,\alpha_j\,\kappa(p_j).
\]

\medskip
\noindent\textbf{Combining.}
Adding the two contributions:
\[
    \sum_{m=1}^{t} \alpha_m\,\kappa(p_m)
    \;\le\; \alpha\,(1 \vee R(t)) + \sum_{j \in \mathcal{I}_t} \theta_j\,\alpha_j\,\kappa(p_j),
\]
which is exactly \eqref{eq:wealth_target}. Rearranging gives $\widehat{\mathrm{FDP}}_{\mathrm{Ada\text{-}GAIF}}(t) \le \alpha$.
\end{proof}

}

\newpage
\vskip 0.2in
\bibliography{main}
\end{document}